\documentclass[10pt]{article}
\usepackage[margin=1.3in]{geometry}

\usepackage{amsmath, amsfonts}
\usepackage{bm, commath}
\usepackage{natbib}
\usepackage{caption}
\usepackage{graphicx}
\usepackage{subcaption}
\usepackage{amsmath, amsfonts, amsthm}
\usepackage{float}
\usepackage{booktabs,siunitx}
\usepackage{url}
\usepackage{multirow}
\usepackage{enumitem}
\usepackage{amssymb}
\usepackage{bm}
\usepackage{mathrsfs}
\usepackage{enumitem}
\usepackage{mathtools}

\usepackage{bm}
\usepackage{natbib}
\usepackage{graphicx}
\usepackage{mathrsfs}
\usepackage{xcolor}
\usepackage{booktabs}
\usepackage{amsthm}

\usepackage{sectsty}
\sectionfont{\bfseries\large\sffamily}%
\subsectionfont{\bfseries\sffamily\normalsize}%
\subsubsectionfont{\normalfont\sffamily\normalsize}%
\paragraphfont{\normalfont\itseries\sffamily\smallsize}%

\usepackage{authblk}
\usepackage{enumitem}% http://ctan.org/pkg/enumitem
\usepackage[nopar]{lipsum}

\usepackage{setspace}
\newcommand{\E}{\mathbb{E}}
\newcommand{\Var}{\mathrm{Var}}
\newcommand{\Cov}{\mathrm{Cov}}
\newtheorem{theorem}{Theorem}

\newtheorem{corollary}{Corollary}
\newtheorem{remark}{Remark}
\newtheorem{example}{Example}
\newtheorem{assumption}{Assumption}

\newtheorem{lemma}{Lemma}

\DeclareMathOperator*{\argmin}{\arg\min}

\def\bgamma{\bm \gamma}
\def\bmu{\mu}
\def\bbeta{\bm \beta}

\def\btheta{\theta}
\def\balpha{\bm \alpha}

\def\var{{\rm Var}}
\def\cov{{\rm Cov}}
\def\Z{Z} 
\def\X{X} 
\def\x{X} 
\def\W{\bm W} 
\def\z{Z}
\def\Y{Y}
\def\Ya{y_a}
\def\Yai{y_{a,i}}
\def\I{{\mathcal I}}
\def\bsigma{\bm{\Sigma}}
\def\aipw{\hat{\btheta}_{{}_{\rm AIPW}}}
\def\jcal{\hat{\btheta}_{{}_{\rm JC}}}
\def\vjcal{{\bm V}_{{}_{\rm JC}}}

\usepackage{pifont}% http://ctan.org/pkg/pifont
\newcommand{\diag}{\text{diag}}

\newcommand{\suma}{\sum_{i:A_i=a}}
\newcommand{\sumn}{\sum_{i=1}^{n}}

\usepackage{apptools}
\AtAppendix{\counterwithin{lemma}{section}}
\AtAppendix{\counterwithin{theorem}{section}}
\AtAppendix{\counterwithin{remark}{section}}
\AtAppendix{\counterwithin{assumption}{section}}

\usepackage[plain,noend]{algorithm2e}

\usepackage{caption}
\makeatletter
\renewcommand{\algocf@captiontext}[2]{#1\algocf@typo. \AlCapFnt{}#2} % text of caption
\def\@algocf@capt@plain{top}
\renewcommand{\algocf@makecaption}[2]{%
  \addtolength{\hsize}{\algomargin}%
  \sbox\@tempboxa{\algocf@captiontext{#1}{#2}}%
  \ifdim\wd\@tempboxa >\hsize%     % if caption is longer than a line
    \hskip .5\algomargin%
    \parbox[t]{\hsize}{\algocf@captiontext{#1}{#2}}% then caption is not centered
  \else%
    \global\@minipagefalse%
    \hbox to\hsize{\box\@tempboxa}% else caption is centered
  \fi%
  \addtolength{\hsize}{-\algomargin}%
}
\makeatother

\def\bgamma{\gamma}
\def\bmu{\mu}
\def\bbeta{\beta}

\def\btheta{\theta}
\def\balpha{\alpha}
\def\var{{\rm Var}}
\def\cov{{\rm Cov}}

\def\Z{Z} 
\def\X{X} 
\def\x{x} 
\def\W{W} 
 
\def\z{z}
\def\Y{Y}
\def\I{{\mathcal I}}
\def\bsigma{\Sigma}
\def\aipw{\hat{\btheta}_{{}_{\rm AIPW}}}
\def\jcal{\hat{\btheta}_{{}_{\rm JC}}}
\def\jcf{\hat{\btheta}_{{}_{\rm CF}}}
\def\vjcal{{V}_{{}_{\rm JC}}}

\begin{document}

\setlength{\abovedisplayskip}{3pt}
\setlength{\belowdisplayskip}{3pt}
\setlength{\abovedisplayshortskip}{0pt}
\setlength{\belowdisplayshortskip}{0pt}

%% The left and right page headers are defined here:
\markboth{Bannick et al.}{A General Form of Covariate Adjustment in Randomized Clinical Trials}

%% Here are the title, author names and addresses
\title{A General Form of Covariate Adjustment in Randomized Clinical Trials}

\author[1]{Marlena S. Bannick}
\affil[1]{Department of Biostatistics, University of Washington\thanks{Correspondence to Ting Ye (tingye1@uw.edu). This work was supported by National Institute of Allergy and Infectious Diseases [NIAID 5 UM1 AI068617].}}
\author[2]{Jun Shao}
%\affil[2]{School of Statistics, East China Normal University}
\affil[2]{Department of Statistics, University of Wisconsin-Madison}

\author[3]{Jingyi Liu}
\author[3]{Yu Du}
\author[3]{Yanyao Yi}
\affil[3]{Global Statistical Sciences, Eli Lilly and Company}

\author[1]{Ting Ye}
\maketitle

\begin{abstract}
In randomized clinical trials, adjusting for baseline covariates can improve credibility and efficiency for demonstrating and quantifying treatment effects. This article studies the augmented inverse propensity weighted (AIPW) estimator, which is a general form of covariate adjustment that uses linear, generalized linear, and non-parametric or machine learning models for the conditional mean of the response given covariates. Under covariate-adaptive randomization, we establish general theorems that show a complete picture of the asymptotic normality, {efficiency gain, and applicability of AIPW estimators}. In particular, we provide for the first time a rigorous theoretical justification of using machine learning methods with cross-fitting for dependent data under  covariate-adaptive randomization.
Based on the general theorems, we offer insights on the conditions for guaranteed efficiency gain and universal applicability {under different randomization schemes}, which also motivate a joint calibration strategy using some constructed covariates after applying AIPW. 
Our methods are implemented in the \textsf{R} package \textsf{RobinCar}.
\end{abstract}

\textsf{{\bf Keywords}: Augmentation; Covariate-adaptive randomization; G-computation;  Model-assisted; Multiple treatment arms; Nonlinear adjustment
}

	\section{Introduction}
	\label{sec: intro}
	
	In randomized clinical trials, the credibility and precision of estimated treatment effects are of utmost importance. Prognostic baseline factors can be utilized at the analysis stage of a clinical trial -- often called \textit{covariate adjustment} -- to decrease the variance of estimated treatment effects and increase power of hypothesis tests \citep{ICHE9,ema:2015aa,fda:2019aa}. 
 Based on extensive research and discussion in recent decades, the U.S. Food and Drug Administration (FDA) recently released guidance for sponsors and trial practitioners on best practices for covariate adjustment \citep{fda:2019aa}. Among them is the importance of \emph{model-assisted} analyses. In other words, although a working model between the responses and covariates may be used to assist estimating the treatment effect, the resulting estimator must be valid  even if the working model is misspecified and ``under approximately the same minimal statistical assumptions that would be needed for unadjusted estimation'' \citep{fda:2019aa}.
 
Adjusting for covariates at the analysis stage using a linear model has a long practical history and strong theoretical justification \citep[among others]{Yang:2001aa, Tsiatis:2008aa, Lin:2013aa, Wang:2019aa, Liu:2020aa, Li:2020aa, Ye2022bio, Ye2021better}. In contemporary perspectives, the classical ANCOVA is regarded as a model-assisted approach that uses the prognostic variables in a linear regression to decrease the variances of treatment effect estimators \citep{Yang:2001aa, Tsiatis:2008aa}. ANHECOVA, or ANCOVA2, also uses a linear model but adds treatment-by-covariate interactions to guarantee efficiency gain even if the model is misspecified \citep{Tsiatis:2008aa, Lin:2013aa, Ye2021better}. Generalized linear models are commonly used for trials with binary or count outcomes to improve the fit to the data \citep{Freedman:2008ab,Moore:2009tu,Rosenblum:2010wt,Steingrimsson:2017we,Wang:2021wg,Guo:2021ua, cohen2021noharm}. 

 Many proposals for model-assisted covariate adjustment, including those listed above, fall under the umbrella of the {\em augmented inverse propensity weighted} (AIPW) method \citep{robins1994estimation, glynn2010introduction, Guo:2021ua} or {\em g-computation} \citep{Freedman:2008ab}. Both methods 
 adjust for baseline covariates by estimating the conditional mean response given covariates under either a parametric or nonparametric working model such as a flexible machine learning model. The AIPW and g-computation approaches yield identical estimators when
the selected working model meets the {\em prediction unbiasedness} condition (see (\ref{PU}) in Section 2). If this condition is not met, estimators from g-computation may have large bias but the AIPW estimators are always asymptotically unbiased.     In fact, the AIPW estimator can be viewed as a de-biased g-computation estimator \citep{firth1998robust} and is a general form of covariate adjustment in the sense that any regular consistent and asymptotically normal estimator for the unconditional treatment effect is exactly equal to or asymptotically equivalent to an AIPW estimator \citep{Tsiatis:2008aa}.  Therefore, we focus on the AIPW method for the rest of the article.

 A complementary approach to covariate adjustment at the analysis stage is to account for prognostic factors at the design stage of a trial using \emph{covariate-adaptive randomization} (CAR), such as stratified permuted block randomization \citep{Zelen:1974aa} and Pocock and Simon's minimization \citep{Taves:1974aa, Pocock:1975aa}. CAR is widely used because it balances important prognostic factors across treatment groups, which enhances the credibility of the trial results \citep{ema:2015aa}. Like model-assisted covariate adjustment, CAR can also decrease the variance of the estimated treatment effect. This last point is nuanced because unlike simple randomization, CAR generates a dependent sequence of treatment assignments. For instance, in stratified permuted block randomization, patients within each stratum are randomized in ``blocks" with a predetermined block size to achieve the desired allocation ratio by the end of each block.
 %If the block size is four in a two-arm trial with equal allocation, then precisely two out of four patients in a block are assigned to each arm in a random fashion.  %using a set of randomly permuted values of (0, 0, 1, 1). This process is repeated for each new block of four patients enrolled within a specific stratum, 
 This leads to a dependency in treatment assignments at the stratum level.
 As emphasized by \cite{fda:2019aa},   to fully realize the power of CAR, variance estimators should account for this dependence. Importantly, covariate-adjusted estimators as described above can be used in conjunction with CAR. 
 
 To facilitate its practical uses, there is a need for a comprehensive theory of AIPW estimators with arbitrary -- and potentially nonparametric -- working models under CAR.
In contrast to previous theoretical contributions that have focused on specific forms of estimators {(e.g.,  M-estimators under parametric models and CAR in \cite{Wang:2021wg}),} 
	we establish a suite of theoretical results that show a complete picture of 
	the asymptotic properties of AIPW estimators of the vector of population mean responses under different treatments, in the context of CAR. Our theory covers the use of machine learning working models and cross-fitting, which is the first rigorous theoretical justification of these methods under CAR that produces dependent data. Beyond asymptotic normality,
we provide  high-level conditions and results for \emph{universal applicability} (i.e., the asymptotic distribution of AIPW estimator is invariant to CAR), and \emph{guaranteed efficiency gain} (i.e., the AIPW estimator is asymptotically at least as efficient as the unadjusted estimator, even in cases of working model misspecification).  

 Furthermore, we provide insights into the conditions  for guaranteed efficiency gain and universal applicability. Our theoretical results motivate our proposal of a  joint calibration estimator that achieves both guaranteed efficiency gain and universal applicability regardless of the initial working model. 
 For  assessment of accuracy and inference, we derive estimators of asymptotic covariance matrices that are both robust to model misspecification and correct under CAR, bridging the gap between the practical guidelines in \cite{fda:2019aa} and the theoretical underpinnings in statistical literature. All estimators for our simulation study and data application are computed using our \textsf{RobinCar R} package, available at \textsf{https://github.com/mbannick/RobinCar}. 

	\section{Setup, Estimators, and Assumptions}
	\label{sec: general setup and estimators}
	
	Consider a trial with $k$ treatments, where treatment $a=1,...,k$ has a pre-specified assignment proportion $\pi_a >0$ and 
	$\sum_{a=1}^k \pi_a =1$. 
	Let %$\bm \pi = (\pi_1, ..., \pi_k)^T$, 
 $y_a$ be the potential response of a patient under treatment  $a$,  and  $\Y = (y_1,...,y_k)^T$, 
	where  ${c}^T$ denotes the transpose of ${c = (c_1,...,c_k)}$. We consider estimation and inference on a given function of the target parameter vector $\theta = 
	(\theta_1, \dots, \theta_k)^T$, where $\theta_a = E(y_a)$, the population mean of $y_a$ under treatment  $a$.  
	Let $\X$ denote the vector of baseline covariates used for adjustment at the analysis stage, and let $\Z$ be a discrete baseline covariate vector used in CAR, where $\Z$ has a fixed number $L<\infty$ of distinct joint levels $\z_1,...,\z_L$ {with $P(\Z=\z_l)>0$ for all $l$}. $\X$ and $\Z$ may share some components. 
 Importantly, we assume that the distribution of baseline or pre-randomization covariates $\X$ and $ \Z$ is not affected by treatment. This is true for clinical trials and is the reason why we can robustly improve efficiency by adjusting for baseline covariates.

	A random sample of $n$ patients is obtained. For the $i$th patient, let  $\Y_i$, $\X_i$, and $\Z_i$ be  realizations of  
	$\Y$, $\X$, and $\Z$, respectively, $i=1,...,n$. 
	Throughout, we assume that $\Y_i, \X_i, \Z_i$, $i=1,..., n$, are independent and identically distributed with finite second order moments.
	
	Each patient is assigned to a single treatment based on a randomization scheme. The assigned treatment for patient $i$ is denoted by $A_i$ taking values $a=1,...,k$. Simple randomization assigns treatments randomly without using any covariates, resulting in $A_i$'s being independent and identically distributed with $P(A_i=a) = \pi_a$, independent of both the potential responses and covariates. Under CAR other than simple randomization, the assigned treatment for patient $i$ depends on the treatment assignments and covariates for all previous patients. Thus, the sequence of treatment assignments $A_1,...,A_n$ are interrelated and depend on $\Z_i$'s used in randomization.
 All commonly used CAR schemes, including those introduced in Section 1 and simple randomization treated as a special case of CAR, satisfy the following mild assumption \citep{Baldi-Antognini:2015aa, hu2022multi}.

	\begin{assumption}[Treatment Assignments]\label{assump: car}
		(i) Given covariate vector $ (\Z_1,...,\Z_n)$, the assignment vector 
		$ (A_1,...,A_n)$ is conditionally independent of 
		$\{\Y_i, \X_i, \ i=1,...,n\}$; (ii) For each $a$, $P(A_i=a\mid Z_1,...,Z_n) = \pi_a $; (iii)	For every level $\z$ of $\Z$, the sequence $\sqrt{n} \{n_a(\z) /n(\z) - \pi_a\} $ is bounded in probability as $n \to
		\infty$,   
		where $n(\z)$
		is the number of patients with $\Z=\z$ and $n_a(\z)$ is the number of
		patients with $\Z=\z$ and  assigned to treatment $a$.
	\end{assumption}

 	{Assumption \ref{assump: car}(i) means that the treatment assignment vector $ (A_1,...,A_n)$ depends solely on the covariate vector $ (\Z_1,...,\Z_n)$. Assumption \ref{assump: car}(ii) ensures that the allocation ratio is preserved at $\pi_a$ for every step; although this implies that $A_i$ is independent of $(\Z_1,...,\Z_n)$, the whole vector $ (A_1,...,A_n)$ is not independent of $(Z_1,...,\Z_n)$ under CAR other than simple randomization. \linebreak
  Assumption \ref{assump: car}(iii) requires a rate at which $n_a(\z) /n(\z)$ converges to $\pi_a$ and is satisfied by most commonly used CAR schemes.  For example, 
  for simple randomization, $\sqrt{n} \{n_a(\z) /n(\z) - \pi_a\} $ converges in distribution to a normal distribution; 
  for stratified permuted block randomization, $n \{n_a(\z) /n(\z) - \pi_a\}  $ is bounded in probability,  as  $| n_a(\z)- n(\z)  \pi_a |$ is bounded by the block size; for Pocock and Simon's minimization, Assumption \ref{assump: car}(iii) holds as shown in \cite{hu2022multi}.
 }
 
	For patient $i$ receiving $A_i = a$, the only observed potential response is the one under treatment $a$ and is denoted as $y_{a, i}$.
 Thus, the observed data are $y_{A_i, i}$, $X_i$, $Z_i$, $i=1,...,n$. 
	Without covariate adjustment, the unadjusted estimator of $\theta$ is the 
	sample mean vector $\bar{Y} = (\bar{y}_1,...,\bar{y}_k)^T$, where each $\bar{y}_a$ is the sample mean of $y_{a,i}$'s for those assigned to treatment $a$.
	Throughout, $\bar{ Y}$ is considered as the 
 benchmark for comparison with covariate-adjusted estimators.
 
	We consider a model-assisted approach  to adjust  for covariate $\X$ using a working model for the conditional response mean $E(y_a \mid \X=  x)$.  Working models can be arbitrary and include parametric  or  nonparametric machine learning methods such as random forest. 
 	Let $\hat\mu_a( x ) $ be the estimated value of $E(y_a \mid \X=  x)$ obtained by fitting a working model under treatment $a$. To estimate the population response mean vector $\btheta$,  we utilize the well-known augmented inverse propensity weighting (AIPW) estimator \citep{robins1994estimation, Tsiatis:2008aa}, given by 
	\begin{align}\label{eq: aipw}
	\aipw =	\bigg( \, \bar{y}_a -  \frac{1}{n_a} \sum_{i: A_i=a }  \hat\mu_a(\X_i) + \frac1n \sum_{i=1}^{n} \hat\mu_a(\X_i), \quad a=1,...,k \, \bigg)^T  ,
	\end{align}
	where $n_a$ is the number of patients in treatment arm $a$. The AIPW estimator in \eqref{eq: aipw} is a general form of covariate adjustment  \citep{Tsiatis:2008aa} and is
	robust against working model misspecification, because $\bar y_a$ is  asymptotically correct for estimating $\theta_a$ and 
	$n_a^{-1} \sum_{i: A_i=a } \hat\mu_a(\X_i) - n^{-1} \sum_{i=1}^{n} \hat\mu_a(\X_i)$  has  asymptotic mean zero, as the distribution of $\X$ is not affected by treatments in randomized clinical trials.   
If
 \begin{equation}\label{PU}
   \bar{y}_a =    \frac{1}{n_a} \sum_{i: A_i = a} \hat \mu_a (\X_i)  \qquad \mbox{almost surely}, 
%   \vspace{-2mm}
 \end{equation}
which  is referred to as prediction unbiasedness \citep{Guo:2021ua}, 
		then $\aipw$ is the same as  the g-computation estimator $\big( n^{-1} \sum_{i=1}^{n} \hat\mu_a(\X_i), a=1,...,k \, \big)^T$ \citep{Freedman:2008ab}, also known as the generalized Oaxaca-Blinder
	estimator \citep{oaxaca1973male, blinder1973wage}. However, without  prediction unbiasedness property (\ref{PU}), the g-computation estimator may have a large bias. The AIPW estimator  in \eqref{eq: aipw} corrects the bias of g-computation estimator and is  
always asymptotically unbiased. To demonstrate this, we give a simple simulation example that compares the g-computation and AIPW approaches in Section 6.

	To estimate $f(\btheta )$, such as a linear contrast of $\btheta$ or risk ratio, we use the AIPW estimator $f( \aipw )$. This estimator can be seen as taking an ``augment-then-contrast" approach because it first augments the response means and then contrasts the augmented response mean estimators.

 	In addition to Assumption 1 for CAR, our asymptotic results presented in Sections 3-4 require the following two  assumptions that are standard for establishing the asymptotic properties of AIPW estimators. 
%\vspace{-2mm}	

	\begin{assumption}[Stability] \label{assump: stability}
		For every $a$, there exists a function ${\mu}_a$ of $\X$ with finite $E\{ \mu_a^2 (X)\}$ such that, as $n \to \infty$, 
		$\|\hat{\mu}_a - \mu_a\|_{L_2} \to 0 $ in probability with respect to the randomness of $\hat{\mu}_a$ as a function of observed data, where  $\| h \|_{L_2} = [ E\{ h^{2}(\X) \} ]^{1/2}$ denotes the $L_2$ norm of a function $h$  of $\X$ with finite $E\{ h^2(\X) \}$ and $E$ is the expectation with respect to the distribution of $X$.  
	\end{assumption}
 
	\begin{assumption}[Donsker condition] \label{assump: simple}
		For every $a$, there exists a class $\mathcal{F}_a$ of functions  of $\X$ such that (i) $\mu_a \in \mathcal{F}_{a}$, where $\mu_a$ is the function in Assumption 2; (ii) $P(\hat\mu_a  \in \mathcal{F}_{a} )\rightarrow 1$ as $n \to \infty$;  {(iii) $\int_0^1 \sup_Q \sqrt {\log {N( \mathcal{F}_{a}, \|\cdot \|_{L_2(Q)} ,s)} } \ d s < \infty$,  where $Q$ is any finitely supported probability distribution on the range of $\X$}, $N( \mathcal{F}_{a}, \|\cdot \|_{L_2(Q)} ,s)$ is the $s$-covering number of the metric space $(\mathcal{F}_{a}, \|\cdot \|_{L_2(Q)} )$, defined as the size of the smallest collection $\mathcal G_{a,s} \subset \mathcal{F}_{a}$ such that every $f \in \mathcal{F}_{a}$ satisfies $\| f - g \|_{L_2(Q)} \leq s$ for some $g \in \mathcal G_{a,s}$,   and $\|\cdot \|_{L_2(Q)}$ is 
  the $L_2$ norm with respect to $Q$. 
  	\end{assumption}
%\vspace{-0.5mm}

 	The function  $\mu_a$ in Assumption 2 can be viewed as a limit of $\hat\mu_a$,  and it is not  equal to the conditional mean $E(y_a\mid \X = \x )$ when the working model is misspecified. 
Under Assumption \ref{assump: simple}, $\hat \mu_a$ and $\mu_a$  fall in a function class $\mathcal{F}_a$ with small complexity. 
%the function class  to be Donsker.
			Examples  of models that meet this assumption include smooth parametric models of fixed dimension 
		and nonparametric regression models with sufficiently many bounded derivatives 
		({van der Vaart and Wellner, 1996; Kosorok, 2008}). If $\hat \mu_a$ is obtained from complex machine learning methods,
		Assumption \ref{assump: simple} may not hold \citep{chernozhukov2017double}. 
  In Section 3.2, we relax Assumption \ref{assump: simple} using cross-fitting. 
 
	\section{Asymptotic Results} 
	\label{sec: asymptotic theory}
	
	\subsection{Asymptotic Linearity and Normality}
	
We present a suite of theoretical results that characterize the asymptotic behavior of AIPW estimators under CAR.

    \begin{theorem}[Asymptotic linearity and normality of $\aipw$ under CAR]\label{theo:asymptotic-normality}
    Suppose that Assumptions \ref{assump: car}-\ref{assump: simple} hold. 
    \vspace{2mm}
 
\noindent
    (i) $\aipw$ has the following asymptotically linear form, regardless of the CAR scheme used:
        \begin{align}\label{influence-function}
        \aipw -  \theta = \frac{1}{n} \sum_{i=1}^{n} \Big( \phi_a(A_i, y_{a,i}, \X_i, \theta_a,  \pi_a, \mu_a ), \ a=1,...,k\Big)^T  + o_p(n^{-1/2})
        \end{align}
        where  $o_p(n^{-1/2})$ denotes a term with 
        $n^{1/2} o_p(n^{-1/2}) \to 0$  in probability as $n \to \infty$, %$\W_i = (A_i, Y_i, \X_i)$, 
        \begin{align*}
  \phi_a(A_i, y_{a,i}, \X_i, \theta_a,  \pi_a, \mu_a ) & =         \frac{I(A_i= a)}{\pi_a } \left[ y_{a,i} - \mu_a (\X_i) - \theta_a + E\{ \mu_a(\X) \}\right] \\
  & \quad + \mu_a(\X_i) -  E\{ \mu_a(\X) \} 
        \end{align*}
  is the $a$th component of influence  function,  and $I(A_i=a)$ is the indicator function of $A_i=a$.
\vspace{2mm}
 
\noindent
(ii) Assume further that,   almost surely,
        \begin{align*}
            &\sqrt{n} \left(\frac{n_a( z_l)}{n( z_l)} - \pi_a, \ a=1,...,k, \ l=1,...,L \right)^T \, \bigg| \, Z_1,...,Z_n  \xrightarrow{d} N\left(0,  V_{ \Omega} \right), \tag{B}
        \end{align*}
        where     $\xrightarrow{d} $ denotes convergence in distribution
	as $n \to \infty$,  $n_a( z)$ and $n( z)$ are defined in Assumption 1, 
$V_{ \Omega} = {\rm diag} \big[
		\frac{\Omega(\z_l)}{P(\Z=\z_l)}, \, {l=1,...,L} \big]$, 
  $ \Omega( z)$ is a $k\times k$ covariance matrix depending on the CAR scheme, and
  ${\rm diag}\big[ C_j, \, j=1,...,J \big]$  denotes the block diagonal matrix whose $j$th block is the matrix $ C_j$, $j=1,...,J$. 
  Then, 
  \begin{align}\label{AN}
      \sqrt{n} (\aipw -  \theta) \ \xrightarrow{d} \ N( 0,  V ),  
  \end{align}
  where
  \begin{align*}
       V &= \mathrm{diag}\big[\pi_a^{-1} \var\left\{y_a - \mu_a( X) \right\}, a = 1, ..., k \big] + \cov\{ Y -  \mu (\X),  \mu (\X)\} \\
%+ \cov\{ Y,  \mu( X) \} + \cov\{ \mu( X),  Y\} - \var \{ \mu( X)\} \\
      &\quad  + \cov\{ \mu(\X),  Y\}- \E\left[\left\{ R_Y( Z) -  R_X( Z) \right\} \left\{ \Omega_{\rm SR} -  \Omega( Z) \right\} \left\{ R_Y( Z) -  R_X( Z) \right\} \right],
  \end{align*}
   $ \mu ( X)= \left( \mu_1(\X) ,..., \mu_k(\X)  \right)^T$ is the vector of $\mu_a (  X )$'s from Assumption \ref{assump: stability}, 
 $ \Omega_{\rm SR} = \mathrm{diag}[ \pi^T] -  \pi  \pi^T$,
 $ \pi^T = (\pi_1,...,\pi_k)$, 
 $  R_Y (\Z)  =  {\rm diag} \big[ \pi_a^{-1} E( Y_a - \theta_a \mid \Z ), \, a=1,...,k  \big]$, and  $
		 R_X (\Z)  =  {\rm diag} \big[ \pi_a^{-1} E[ \mu_a(\X) -E\{\mu_a(\X)\} \mid \Z], \, a=1,...,k  \big] $.

    \end{theorem}

To the best of our knowledge, this is the first result for asymptotic linearity and normality of AIPW estimator $\aipw$ allowing for nonparametric estimates of the conditional means, under CAR. 
The proof of all theorems are in the Supplemental Material.
A novel technique in proving (\ref{influence-function}) in Theorem \ref{theo:asymptotic-normality} shows that, under CAR, for a given $a$,
 \begin{align*}
     \frac{1}{n} \sum_{i=1}^n I(A_i=a) \left[\left\{\hat{\mu}_a( X_i) - \mu_a(\X_i)\right\} - E \left\{\hat{\mu}_a( X) - \mu_a( X) \right\}\right] = o_p(n^{-1/2}),
 \end{align*}
using a modification to the standard empirical process result to account for the fact that $A_i$'s are dependent.

Besides Assumptions 1-3, the asymptotic normality of 
$\aipw$ under CAR requires condition (B)  on the limiting distribution of the allocation fractions by stratum and treatment group in randomization. Condition
(B) is satisfied under simple randomization with 
 $ \Omega ( z) = \Omega_{\rm SR}$ 
and stratified permuted block randomization with $ \Omega( z) =  0$ for all $ z$.  
It is also satisfied under 
   most commonly used CAR schemes that are implemented separately for each stratum.    Pocock and Simon's minimization is a notable exception under which condition (B) is not satisfied -- it assigns patients by minimizing a balance measure across the marginal levels of covariates, thereby generating correlations across different strata. Later in Section 3.3, we will provide a strategy for valid inference under Pocock and Simon's minimization.

The  sample mean $\bar{ Y}$ can be treated as a special case of AIPW with $\hat\mu_a = \mu_a =0$ for all $a$. Thus, Theorem \ref{theo:asymptotic-normality} tells us that, under Assumption 1 and (B), (\ref{AN}) holds for $\bar{ Y}$ with 
$$  V = \mathrm{diag}\big[ \pi_a^{-1} \var (y_a), a = 1, ..., k \big]  - \E\left[ R_Y( Z) \left\{ \Omega_{\rm SR} -  \Omega( Z) \right\}  R_Y( Z)   \right]. $$
The matrix $  \Omega_{\rm SR}  -\Omega( Z) $  is typically positive definite for CAR other than simple randomization. As a result, using CAR reduces the asymptotic variance of $\bar{ Y}$ compared to simple randomization. This also highlights the important point that using conventional formulas for inference based on  $\bar{ Y}$ under simple randomization can be conservative. 

   \subsection{Machine Learning and Cross-Fitting}
	\label{sec: ml}

	If one constructs $\hat\mu_a(\X)$  using machine learning methods such as random forest, then $\hat \mu_a$ and its limit may be too complicated to satisfy  Assumption 3 (Donsker condition). Cross-fitting -- separating the data used for fitting the machine learning models and the data used for making predictions -- is a technique that allows us to relax Assumption 3 \citep{chernozhukov2017double}.
	
    Specifically, let $\I_1,...,\I_J$ be a random partition of $\{ 1,..., n\}$ with a fixed integer $J$. {Define} $n^{(j)} = n/J $ as the size of $\I_j$, and  let $\I^c_j $ be the set containing indices in $\{1,..., n\}$ but not in $\I_j$, $j=1,...,J$. For each partition $j$, we use data with indices in $\mathcal{I}_j^c$ to obtain $\hat{\mu}_{a,j}(\X )$ as an estimate of $E(y_a \mid \X )$, $j=1,...,J$.
The cross-fitted estimator of $\btheta$ is defined as 
	\begin{align}
 \jcf= \bigg( \frac{1}{J} \sum_{j=1}^{J} \frac{1}{n^{(j)}}\sum_{i\in \I_j} \bigg[\frac{I(A_i=a)}{\hat{\pi}_{a, j
	}}\{y_{a,i} - \hat{\mu}_{a,j}(\X_i) \} + \hat{\mu}_{a,j}(\X_i) \bigg]  , \ \ \ a=1,...,k \bigg)^T ,\label{eq: cf}
	\end{align}
	where  $ \hat\pi_{a,j}=n_{a}^{(j)}/n^{(j)}$ and $n_{a}^{(j)}$ counts indices in  $\I_j$ with $A_i =a$. 

\vspace{2mm}

%	\noindent
	{\em Assumption 2*}. 
	{\em For every $a$, there exists a function ${\mu}_a $ of $X$ with finite $L_2$ norm
		such that
		$P(\|\hat \mu_{a,j}  - \mu_a \|_{L_2 } \leq \delta_n )\to 1$
		as $n \to \infty$, 
		where $\hat\mu_{a,j}$ is given in  \eqref{eq: cf} for a fixed $j$ and 
		$\{ \delta_n \}$ is a sequence of positive numbers converging to 0}. 
	\vspace{2mm}

 With cross-fitting, we use Assumption 2* to replace Assumptions 2-3. 
 Assumption 2* is slightly stronger than Assumption 2, since 
 it requires $\|\hat \mu_{a,j}  - \mu_a \|_{L_2 } \to 0$ at a rate at least $\delta_n$, though $\delta_n$ can be any sequence tending to 0. Assumption 2* is  achievable
by many machine learning methods \citep{chernozhukov2017double}. 
 
 We have the following result for $\jcf$ in (\ref{eq: cf}) under CAR.
        \begin{theorem}[Asymptotic linearity and normality of $\jcf$ under CAR]\label{theo:ml}
Suppose that Assumptions 1 and 2* hold. Then all results in Theorem 1 hold with $\aipw$ replaced by $\jcf$ defined in \eqref{eq: cf}. 
        \end{theorem}
        
Under simple randomization, $(A_i, y_{A_i, i}, \X_i) $'s are independent and \cite{chernozhukov2017double} showed that  $\jcf$ is asymptotically linear. Under CAR other than simple randomization, however, $(A_i, y_{A_i,i}, \X_i) $'s are dependent. To handle the dependency arising from CAR, we need to define a new empirical process term, which not only involves conditioning on the auxiliary samples in $\mathcal{I}_j^c$ (used to estimate nuisance parameters) to render the cross-fitted nuisance parameter estimators non-stochastic, but also requires conditioning on all the treatment and strata indicators $(A_1,\dots, A_n, Z_1\dots, Z_n)$ to make the samples conditionally independent, though not identically distributed. Hence, a set of empirical process results needs to be extended to cases when data are not identically distributed. We provide details in the Supplementary Material.

\subsection{Guaranteed Efficiency Gain and Universal Applicability}\label{sec:intuition}

Our next result establishes a high-level condition for $\aipw$ to have guaranteed efficiency gain over the  sample mean $\bar{ Y} $ under CAR schemes that satisfy {\rm (B)}. 

  \begin{theorem}[Guaranteed efficiency gain under CAR]\label{corollary:efficiency}
  Assume Assumptions 1-3 for $\aipw$  or Assumptions 1 and 2* for $\jcf$.  
    For all CAR schemes satisfying  {\rm (B)}, a sufficient condition for $\aipw$ or  $\jcf$ having guaranteed efficiency gain over the sample mean  $\bar{ Y} $ under CAR is  
       \begin{equation}
        \begin{array}{c}
        {\rm diag} \big[ \pi_a^{-1} \cov \{ y_a-\mu_a(\X), \mu_a(\X)\}, a=1,...,k \big] - \cov \{\Y -  \mu (\X),  \mu (\X) \} \vspace{2mm} \\
		= E\left[ \{  R_Y (\Z) -  R_X(\Z)\} \{  \Omega_{\rm SR} - \Omega(\Z) \}   R_X(\Z)  \right] . \end{array} \tag{G1}
     \end{equation}
    \end{theorem}

{Condition (G1) is a novel condition for guaranteed efficiency gain under CAR. It can be simplified 
when $ \Omega(\Z) =\Omega_{\rm SR} $ under simple randomization. See also the result in Theorem 4.}

 Condition (B) is required for  Theorems  \ref{theo:asymptotic-normality}-\ref{corollary:efficiency}, and it requires knowing the form of 
 $ \Omega ( z )$ for a particular CAR scheme. 
 Importantly, condition (B) is not satisfied for Pocock and Simon's minimization.
 To  bypass condition (B), we establish the following theorem for $\aipw$ and $\jcf$ under a different condition (U) for universal applicability.  The result also strengthens the asymptotic normality in Theorems  \ref{theo:asymptotic-normality}-\ref{corollary:efficiency}, i.e., 
 the asymptotic distribution  of $\aipw$ or $\jcf$ is invariant across CAR schemes.

	\begin{theorem}[Universal applicability]\label{theo:universality}
 Assume Assumptions 1-3 for $\aipw$  or Assumptions 1 and 2* for $\jcf$.  	If
  \begin{align*}
      E \{  \Y -  \mu (\X)   \mid \Z\} =  \theta - E\{  \mu (\X )\}  \quad \mbox{almost surely,} \tag{U}
  \end{align*}
		then results \eqref{influence-function} and \eqref{AN}   hold 
  for $\aipw$ or $\jcf$ with 
 $$ 
    V \!=\! \mathrm{diag}\left[\pi_a^{-1} \var\left\{y_a - \mu_a( X) \right\}, a = 1, ..., k \right] \!+\! \cov\{ Y -  \mu (\X),  \mu (\X)\}\!+\! \cov\{ \mu(\X),  Y\}, $$
  i.e., 
		the asymptotic distribution of $\aipw$ or $\jcf$  is invariant to CAR schemes. Furthermore, (G1) in Theorem 3 reduces to the following condition 
  for $\aipw$ or $\jcf$   having guaranteed efficiency gain over $\bar Y$, 
  \begin{align*}
      \cov \{  Y -  \mu (\X), \mu (\X) \} = 0. \tag{G2}
  \end{align*}
	\end{theorem}

The universal applicability in Theorem 4 provides valid inference under Pocock-Simon's minimization without condition (B) or knowledge of $\Omega(z)$. Additionally, it also facilitates inference about functions of $ \theta$ using a single and universally valid inference procedure applicable to all CAR schemes satisfying Assumption 1. In terms of conditions for asymptotic normality, Theorem \ref{theo:universality} complements Theorems 1-3 in the sense that 
condition (U) in Theorem 4 is a restriction on the limit function  $ \mu (\X) $, whereas condition (B) in Theorems 1-3 is on the CAR schemes and the information about  $ \Omega( z)$.  In terms of sufficient conditions on guaranteed efficiency gain, (U)+(G2) in Theorem 4 implies (G1), but (G1) has to be applied with condition (B).  

A situation under which (U), (G1), and (G2) are all satisfied is when the working model is correctly specified, i.e., $\bmu (\X) =E(\Y \mid \X)$, and 
	$\hat \mu_a$ is correctly constructed with $\X $ including $\Z$. In fact, in this situation, the AIPW estimator also achieves the semiparametric efficiency bound (see, e.g. \cite{Tsiatis:2008aa}).

Note that (U) may not hold 
even if a nonparametric method such as the machine learning is applied,
when the working model is misspecified.  
 To restore universal applicability  if the initial $\hat \mu_a$ does not satisfy condition (U), which is particularly important for Pocock and Simon's minimization, 
a simple approach 	is to replace $\hat \mu_a (\X_i)$  with 
$$
	\hat{\mu}^{(U)}_a (\X_i) = \hat \mu_a(\X_i) + \frac{1}{n_a(\Z_i)} \sum_{j: A_j = a, \Z_j = \Z_i} \{ y_{a,j} - \hat \mu_a(\X_j) \} 
$$
	for every $a$. Since the limiting value of $\hat{\mu}^{(U)}_a(\X)$ is
	$  \mu_a(\X) \!+ \! E \{ y_a - \mu_a(\X) \mid \Z\} $
	for every $a$, condition (U) is satisfied  with $\bmu (\X) $ replaced by  $\big(\mu_1(\X) \!+ \! E \{ Y_1 - \mu_1(\X) | \Z\}, \dots,\mu_k(\X) \!+ \! E \{ Y_k - \mu_k(\X) \mid \Z\} \big)^T$. This can be viewed as  a
 ``stratum-specific internal bias calibration'' and is called $Z$-calibration for simplicity.
 
  When the initial $\hat\mu_a$ does not satisfy (G1) or (G2),  under simple randomization, \cite{cohen2021noharm} proposed a linear calibration strategy that replaces $\hat \mu_a (X)$ with $\tilde \gamma_a^T \hat \mu (X)$, where $\hat \mu (X) = (\hat \mu_1 (\X), ..., \hat \mu_k (\X)  )^T$ and $\tilde \gamma_a$ is 
 {the coefficient vector of $\hat \bmu (\X)$ from a linear regression of $y_a$  on ``covariate" $\hat \bmu (\X)$ plus an intercept} with data under treatment $a$. 
  However, this strategy only works under simple randomization.

	\section{Joint Calibration for Guaranteed Efficiency Gain \\ and Universal Applicability} \label{subsec: joint calibration} 
	
We aim to derive an estimator of $\theta$ that has both universal applicability and guaranteed efficiency gain. {In Section \ref{sec: asymptotic theory}.3   we have introduced two calibration techniques, the $Z$-calibration for universal applicability and linear calibration for guaranteed efficiency gain under simple randomization. It 
	seems  to suggest that we can first achieve  universal applicability  through $Z$-calibration  to obtain $\hat{\mu}^{(U)} (X) = \left(\hat{\mu}^{(U)}_1(\X),\dots,\hat{\mu}^{(U)}_k(\X)\right)^T$, and then, to obtain guaranteed efficiency gain, we apply a second calibration, the linear calibration proposed by 
	\cite{cohen2021noharm} with $\hat \mu (X)$ replaced by $\hat{\mu}^{(U)} (X)$. However, this second calibration step may lose universal applicability and thus does not necessarily lead to guaranteed efficiency gain under CAR.  
	%Likewise, it does not work by first obtaining guaranteed efficiency gain through the LC  and then achieving universal applicability using $Z$-calibration, because guaranteed efficiency gain may be lost after the second calibration.} 
	
	To simultaneously achieve both  universal applicability  and guaranteed efficiency gain,
	we propose the following joint calibration with  regressor $\hat \W = \left(\Z^T, \hat\bmu(\X)^T \right)^T$  under each treatment $a=1,...,k$.   We assume that the regressors have been appropriately adjusted to avoid collinearity such that $\Sigma_W = \var (W) $ is positive definite, where $W = \left(\Z^T, \bmu(\X)^T \right)^T$. 
 Define 
    $$\hat \mu_a^* (\hat W_i) =  \hat {\gamma}_a^T  \hat\W_i , \qquad a=1,...,k, $$
 where $\hat \W_i$ is the value of $\hat \W$ from patient $i$ and {$\hat\gamma_a$ is the coefficient vector of $\hat W$ from a linear regression of $y_a$ on  covariate $\hat W$ plus an intercept with data under treatment $a$}. %Note that  $\hat {\gamma}_a$ should vary with $a$. 
 The joint calibration estimator $\jcal$ is defined as the AIPW estimator in \eqref{eq: aipw} or the cross-fitted estimator in (\ref{eq: cf})
 with $\hat\mu_a(\X_i)$ replaced by $\hat \mu_a^* (\hat W_i)$ for $a=1,..., k$.
%{\blue  This joint calibration does not add any serious computational complexity}. 

%In the special case where $\hat\bmu(\X)$ is obtained by fitting {\red arm-specific} linear models that adjust for $\X$, the resulting joint calibrated $\hat \mu_a^* (\X)$ is the same as the estimated $E(y_a \mid \X)$ under  ANHECOVA working model in \cite{Ye2021better}. {\red [TY: also need to mention the ANHECOVA including Z.]}

 Theorem \ref{theo: double} summarizes the asymptotic properties of $\jcal$, showing that joint calibration can simultaneously achieve universal applicability and guaranteed efficiency gain.
 
	\begin{theorem}[Asymptotic Properties of $\jcal$ under CAR]\label{theo: double} 
Assume Assumptions 1-3 for $\aipw$  or Assumptions 1 and 2* for $\jcf$. \\
  (i)  Regardless of which CAR scheme is used,
 \begin{align*}
     \sqrt{n} (	\jcal  -  \theta     ) \ \xrightarrow{d} \ N( 0, \vjcal ),
 \end{align*}
 where $\vjcal = {\rm diag} \big[ \pi_a^{-1} \var (y_a -  \gamma_a^T  W), a=1,...,k \big] +  \Gamma^T  \Sigma_W  \Gamma$,  $\gamma_a = \bsigma_W^{-1} \cov (\W , y_a)$,  and $\Gamma = (\bgamma_1,..., \bgamma_k)$.\\
(ii) For all randomization schemes satisfying (B), $\jcal  $ has guaranteed efficiency gain over the sample mean $\bar{ Y}$.
	\end{theorem}

An application of joint calibration and Theorem \ref{theo: double} is illustrated in the following example where a generalized linear model (GLM) is used as a working model.  

\begin{example}	
	Consider a GLM with a canonical link as a working model for $E(y_a \mid \X ) $,  $a=1,...,k$. In this case,
 $\hat\mu_a(\X) = g^{-1}  (\hat\alpha_a + \hat\bbeta_a^T \X) $, 
 where $\hat{\alpha}_a$ and $\hat \bbeta_a$ are fitted values using the GLM and 
 $g^{-1} ( \cdot )$ is a known differentiable and invertible function whose inverse $g$ is the canonical link function. The specific form of $g$ depends on the problem being considered.   
	For example, 
	if responses are binary, then $g^{-1} (t) = {\rm expit} (t) = \exp(t)/\{1+\exp(t)\}$ (logistic working model);
	if responses are non-negative integer valued, then $g^{-1}(t) = \exp(t)$ (Poisson working model); if responses are continuous, then $g^{-1} (t) = t$ (linear working model). Since the working model for each $a$ is a finite-dimensional parametric model, Assumptions \ref{assump: stability}-\ref{assump: simple} are satisfied. 
	The limit of $\hat\mu_a(\X) $ is
	$\mu_a(\X) = g^{-1}  (\alpha_a + \bbeta_a^{ T} \X) $,  
	where $\alpha_a$ and $ \bbeta_a $ are the limits of 
	$\hat\alpha_a$ and $\hat\bbeta_a$, respectively, satisfying 
	\begin{equation}\label{pscore}
	E \left[ \{ Y_{a} -  g^{-1}( \alpha_a + \bbeta_a ^{T}\X) \}  {\textstyle {1 \choose \X}}  \right] = 0, \qquad a=1,...,k.
	\end{equation}
	The first line of (\ref{pscore}) implies that $E \{  \mu (\X)\} = \btheta$. The sample analog of (6) implies that $\hat{\mu}_a(X)$ satisfies the prediction unbiasedness property in (\ref{PU}). {If all joint levels of the discrete strata variable $\Z$ are included in $\X$, then the second line of (\ref{pscore}) implies that} $	E [ \{ y_a-  \mu_a (\X)  \} \Z ] = 0$ and thus  $	E \{   y_a -  \mu_a (\X)   \mid \Z \} = 0$ for every $a$.  Consequently, condition (U) in Theorem 3 is satisfied. When the working model is misspecified and $g^{-1}$ is nonlinear, however, neither (G1) nor (G2) is satisfied. Thus, the AIPW estimator $\aipw$ that uses the GLM conditional mean estimate $\hat{\mu}_a(X)$ may not have guaranteed efficiency gain over the benchmark sample mean vector $\bar \Y$ 
\citep{Guo:2021ua,cohen2021noharm}.

 	To achieve both guaranteed efficiency gain and universal applicability, we can apply  the joint calibration and use $\jcal$. 
	With $\hat \mu_a(\X)$ replaced by $\hat\mu_a^*(\X) = \hat\gamma_a^T \hat W_i$ defined in  joint calibration, both (U) and (G2) in Theorem 4 hold and, hence, 
	$\jcal$ achieves both guaranteed efficiency gain over $\bar Y$ and universal applicability.
\end{example} 

\section{Robust Variance Estimation}
\label{sec: var}
The results in Theorems 1-5  allow us to make asymptotically valid inference about $ \theta$ such as constructing confidence intervals for functions of $ \theta$, provided that we can derive 
an estimator of $ V$ in (\ref{AN}) that is consistent and robust 
against misspecification of 	working models. 

 Let $s^2_{a} $ be the sample variance of $y_{a,i} $ for patients in treatment arm $a$, $\hat Q$ be the $k\times k$ matrix whose $(a,b)$th component  is the sample covariance of $y_{a,i}$ and $\hat \mu_b (\X_i)$ for patients in arm  $a$, 
	$\hat \Sigma$ be the sample covariance matrix of $(\hat \mu_1 (\X_i), \dots, \hat \mu_k (\X_i))^T$ based on the entire sample, $\hat Q_{aa}$ and $\hat \Sigma_{aa}$ be respectively the $(a,a)$th component of $\hat Q$ and $\hat \Sigma$, $\hat { R}_Y(\z) = {\rm diag}\big[ \pi_a^{-1}  \{\bar y_{a} ( z) - \hat{\btheta}_{a}\} ,  a=1,\dots, k \big]$, $\hat { R}_X(\z) = {\rm diag}\big[ \pi_a^{-1}  \{\bar \mu_{a} ( z) - \bar \mu_{a}\}, a=1,\dots, k \big]$, $\bar y_{a} ( z)$ be the sample mean of $y_{a,i}$ for patients in treatment arm $a$ and stratum $\z$, $\hat{\btheta}_{a}$ be the $a$th component in $\aipw$ or $\jcf$, 
	$\bar\mu_a (\z) $ be the sample mean of $\hat \mu_a (\X_i)$'s for patients in stratum $\z$, and  $\bar\mu_a  $ be the sample mean of $\hat \mu_a (\X_i)$'s for all patients. Then, $ V $ in (\ref{AN}) can be consistently estimated by 
	\begin{align*}
	 \hat { V} = & \
	{\rm diag}\left[ \pi_a^{-1} ( s_a^2 - 2 \hat Q_{aa} + \hat\Sigma_{aa}), \, a=1,...,k\right] + 
	\hat Q + \hat Q^T- \hat \Sigma \\
	& \
	- \sum_{\z } \frac{n(\z)}{n}  \{ \hat{ R}_Y (\z) - \hat{ R}_X(\z)\} \{  \Omega_{\rm SR} - \Omega(\z) \}   \{ \hat{ R}_Y (\z) - \hat{ R}_X(\z)\} , 
	\end{align*}
when the randomization scheme satisfies condition (B), regardless of whether the working model is misspecified or not. This estimator requires knowing $ \Omega(\z) $.
 When condition (U) is satisfied, $\hat { V} $ can be simplified by dropping the last term, without requiring condition (B) and knowing $ \Omega(\z) $. 
 
For the joint calibration estimator $\jcal$, $\vjcal$ in Theorem 5 can be consistently estimated  using the formula of $\hat V$, {but with the last term dropped, $\hat \mu_a $ replaced by $\hat \mu_a^*$, and $\hat\theta_a$ replaced by the $a$th component of $\jcal$}.

 }
 
 For estimation and inference about a differentiable function of $\theta$, a robust standard error can be obtained based on $\hat { V}$ and the delta method.

			\section{Simulations}
			\label{sec: simu}

First,  in a small simulation we show that the g-computation estimator may be biased when prediction unbiasedness property (\ref{PU}) does not hold, while the AIPW estimator is approximately unbiased. Figure \ref{fig:neg-bin} shows the density plot of point estimates for 1,000 simulated trials using g-computation and AIPW estimators with a negative binomial working model, which is commonly used to model count responses. We can see that the g-computation estimator has a positive bias in estimating the treatment effect $\theta_2 - \theta_1 =  3.25$, but the AIPW estimator corrects for this bias.

\begin{figure}[ht]
    \centering
    \includegraphics[scale=.7]{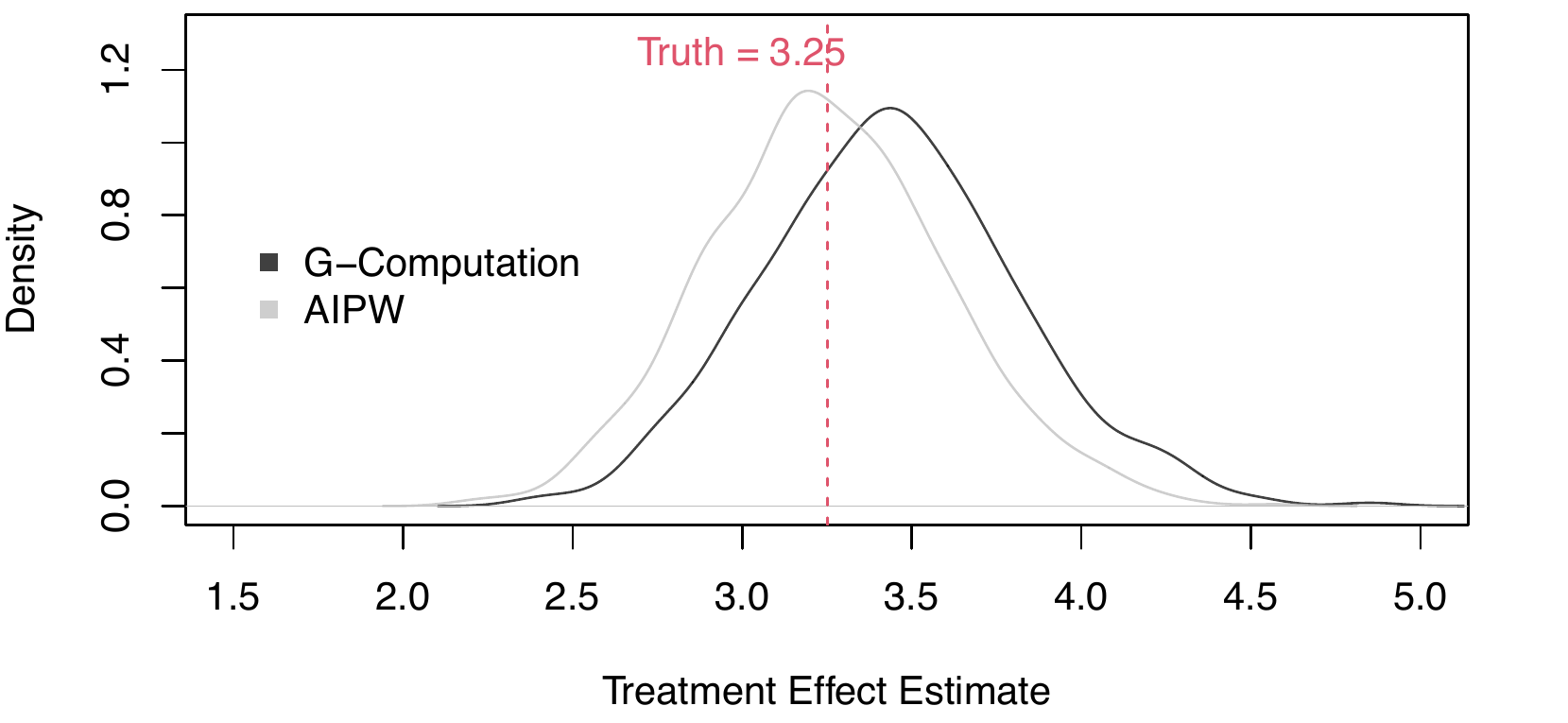}
    \caption{Treatment effect estimates for 1000 simulated two-arm trials with $n = 500$ based on a Poisson data generating process, under simple randomization with $\pi_1 = 1/3$. The response means conditional on  $X = (X_c, X_b)^T$  are  $E(y_1 |X) = \log(1 + X_{c} + 6 X_{c}^3 + X_{b})$ and $E(y_2|X) = X_{c} + 10 X_{c}^2 + X_{b}$, where
    $X_c$ is uniformly distributed on the interval $(-5,5)$, $X_b$ is binary with $P(X_b= 1) =0.5$, and $X_c$ and $X_b$ are independent. The working model  is negative binomial in covariate $X$ including treatment-by-covariate interactions, with an unknown dispersion parameter.}
    \label{fig:neg-bin}
\end{figure}

	Next, we perform a simulation study to assess and compare the finite-sample performance of all aforementioned estimators, in a two-arm trial ($k=2$) with binary responses. We consider the following two cases. \vspace{2mm}

 \noindent
 Case I. $X = (X_c, X_b)^T$, where $X_c$ is uniformly distributed on the interval  $(-5,5)$, $X_b$ is binary with $P(X_b= 1) =0.5$, and $X_c$ and $X_b$ are independent. Given $X$, the true models for potential binary responses $y_1$ and $y_2$ are 
            \begin{align*}
           E(y_1 \mid X) &= {\rm expit} \{0.5  + 0.5 X_{c} + 0.5 X_{b}- 0.2X_{c}^2 \} \\
                E(y_2 \mid X) &= {\rm expit}\{ 0.2 + 0.5 X_{c} + 0.5 X_{b} \}.
               \end{align*}
 We consider $\pi_1 = \pi_2 = 0.5$
 under three randomization methods: 
			simple randomization, stratified permuted block randomization with a block size 6, and Pocock and Simon's minimization. The $Z$ used in CAR has joint levels of  $X_{b}$ and $X_{c}$ dichotomized at its median. 
   
 \vspace{2mm}

 \noindent
 Case II. 
$\X = (X_{c1}, X_{c2}, X_{c3}, X_{b})^T$,  where $X_{c1}, X_{c2}$, and  $X_{c3}$ are  uniformly distributed  on $(-5, 5)$, $X_b$ is binary with $P(X_b= 1) =0.5$, and $X_{c1}$, $X_{c2}$, $X_{c3}$, and $ X_b$ are independent. Given $X$, the true models for potential binary responses $y_1$ and $y_2$ are 
               \begin{align*}
           E(y_1 \mid X) &= {\rm expit} \{0.2  - 0.5X_{c1} + 0.5 X_{c2} + X_{c3}+ 0.2 X_b  \\
                & \quad + X_{c1}(X_{c2} + X_{c3}) -0.2 X_{c1}^2 X_{b}   - 0.02 X_{c1}^2 (1-X_{b})   \}, \\
                E(y_2 \mid X) &=  1 -0.02  X_{c1}^2 - 0.02 X_{c2}^2 .
               \end{align*}
           We consider $\pi_1 = 2/3$ and $\pi_2 = 1/3$,  under three randomization methods given in Case I. The $Z$ used in CAR has joint levels of  $X_{c1}$, $X_{c2}$, and $X_{c3}$ dichotomized at their medians. 
\vspace{2mm}   
  
  The estimand of interest for both cases is the average treatment effect $\theta_2 - \theta_1 $, where $\theta_a = \E(y_a)$. 
    The following estimators of $\theta_2 - \theta_1$ are considered: First, the 
 sample mean $\bar \Y$. Second, the AIPW, linear calibration (LC), and joint calibration (JC) estimators, using  
    the logistic model described in Example 1 of Section 4 as a working model for each treatment arm, i.e.,  $\hat \mu_a (X) = {\rm expit} ( \hat\alpha_a + \hat\beta_a^T X)$ for each $a$. {In Case I, this logistic working model is correct for treatment arm 2 but incorrect for arm 1; in Case II, the working models are incorrect for both arms, with deviations more substantial than those in Case I.}  Third, the AIPW and JC estimators, implemented using the random forest method by the \textsf{SuperLearner R} package separately for each arm, with 5-fold cross-fitting. 
    Lastly, the AIPW estimator based on the true response model, which is an oracle estimator since it cannot be computed in practice when the true response model is  unknown, and is considered for benchmark purposes.  
			
			The simulation results with {$n= 1,000$ based on $5,000$ simulations} are reported in Table \ref{tab:simulation-main}. The results contain the simulation average of bias, simulation standard deviation (SD), simulation averages of the robust estimated standard error (SE)  developed in Section 5
			and the naive SE obtained by treating CAR as simple randomization, and simulation coverage probability (CP) of the 95\% confidence intervals based on asymptotic normality and the two SEs.

			The following is a summary for the results in Table \ref{tab:simulation-main}. \vspace{2mm}
   
			\begin{enumerate}[topsep=1pt]
				\item  All estimators have negligible biases compared to their SDs.
              Our proposed robust SEs are very accurate in all cases. As a result, the CPs of confidence intervals using correct SEs are all close to the nominal level 95\%. All these results support our asymptotic theory developed in Sections 3-4. 
				\item In terms of SD, the performance  of an estimator is primarily influenced  by the proximity of its working model to the true model regarding the conditional mean of the response given the covariates. 
    The oracle estimator performs the best since it uses the true model.  In Case II, where the working logistic model significantly deviates from the true model, estimators employing random forest clearly outperform the AIPW estimators that rely on the logistic working model. The estimators using the logistic working model and random forest are comparable in Case I, where the logistic working model is either identical to or not significantly different from the true model. Although the sample mean consistently performs the worst, it can sometimes be comparable to the AIPW estimators under poor working models. 
\item   The use of  CAR improves SD, but the effect is not as large as modeling since CAR uses dichotomized covariates. 
   Under CAR, if universal applicability does not hold, the naive SE is not the same as the correct SE and tends to overestimate, leading to an inflated CP. 
Thus, it is important to use the correct SE that accounts for CAR.  Additionally, if universal applicability does not hold, there is no theoretically valid SE and confidence interval under Pocock and Simon's minimization. 
\item  {The JC estimator  performs almost identically under simple randomization, stratified permuted block, and Pocock and Simon’s minimization, which confirms   
 our asymptotic theory on its universal applicability}. Although applying JC usually reduces SD, the magnitude of the reduction is not substantial when $Z$ is already used either in AIPW or in CAR, indicating that modeling and CAR are more useful than JC in efficiency gain.   

 \end{enumerate}
   
\section{Application to Clinical Trial Data}
			\label{sec: real}
			
			We analyze data from a phase 3 randomized clinical trial comparing a new basal insulin (insulin peglispro) to insulin glargine as the control in
			 insulin-naive patients with type-2 diabetes \citep{davies2016basal}. 
			A total of 1516 patients were randomized in a 2:1 ratio to insulin peglispro or insulin glargine using permuted block randomization (with a block size of 6), stratified according to {country (23 levels)}, baseline hemoglobin A1C (HbA1C) ($\leq$ 8.5 and $>$ 8.5 \%), low-density lipoprotein cholesterol (LDL; $<$100 and $\geq$100 mg/dL), {baseline sulfonylurea or meglitinide use}. The endpoint we consider is an indicator of having an average of more than two hypoglycemic events per 30 days, during weeks 52-78 of follow-up. The estimand that we are interested in is the linear contrast of treatment group means.
			
			We report the unadjusted analysis based on sample means. In addition, we consider two methods to adjust for the following covariates: age, sex, body mass index (BMI), country, baseline HbA1c, baseline LDL, baseline sulfonylurea or meglitinide use, {and baseline hypoglycemic events count}. As we did in the simulation section, the first method uses the logistic model described in Example 1 of Section 4 as a working model for each treatment arm.
			The second method fits a separate ensemble of random forests for each arm using the \textsf{SuperLearner R} package.  We also perform linear and joint calibration on the logistic working model, and we perform joint calibration on the random forest.  
			
		Our results in Table \ref{tb: BIAQ} show that patients receiving insulin peglispro were significantly less likely on average to experience more than two hypoglycemic events per 30 days within the 52-78 week follow-up period. All the point estimators are close to each other (with negligible differences compared to the standard errors). {Covariate adjustment 
  %tends to produce smaller standard errors and  p-values. AIPW 
  using the logistic working model has a 5.6\% variance reduction compared to the unadjusted sample mean, which roughly corresponds to saving 5.6\% sample size. 
  %with correct SE. 
  This saving is amplified to 23\% if the naive variance estimator is used for unadjusted sample mean, showing the importance of using correct variance estimators. The random forest method does not yield any efficiency gain over the use of logistic working model, which may be because the logistic model is not far away from the true model. }

{The joint calibration enjoys universality, and it is the only method for which the naive standard errors developed under simple randomization are correct. 
For all other methods, the naive standard errors are too conservative because stratified permuted block randomization was applied. 
Note that the unadjusted sample mean is actually based on stratified permuted block randomization and is much better than the one under simple randomization, which explains why it cannot be substantially further improved by covariate adjustment in the analysis stage}. 
			   
\section{Conclusion and Recommendation}

 From our theoretical and empirical findings, we draw the following conclusions and recommendations. When adjusting for baseline covariates, while AIPW with any misspecified working model can still lead to valid inference of the treatment effect, using a working model that closely mirrors the true model can significantly improve efficiency compared to methods that do not adjust for covariates. Machine learning methods, compared to parametric approaches, stand out for their flexibility and capability to closely approximate the true model.  Joint calibration provides a simple strategy that ensures universal applicability (allowing for the use of a universal inference procedure) and provides a guaranteed efficiency gain  over the unadjusted method. In all cases, it is critical to use variance estimators which are robust to model misspecification, and which take the covariate-adaptive randomization scheme into account.

% TABLES -------------------------------------------------

% latex table generated in R 4.1.2 by xtable 1.8-4 package
% Thu Jan 25 15:50:48 2024
\begin{table}[htbp!]
\small
\centering
\caption{Simulation results for estimating $\theta_2 -\theta_1 = 0.167$ under case I and $\theta_2 -\theta_1 = 0.164$ under case II, based on $n= 1000$
                              and 5000 simulation replicates. Entries have all been multiplied by 100. The naive SE and CP use the method ignoring CAR, with empty entries  under simple randomization, oracle, or JC. Under  minimization, some entries are marked with ``--'' because the correct SE and CP are not available. }\label{tab:simulation-main}
                            
\begin{tabular}{cccccccccc}
  \hline
& & & & & & \multicolumn{2}{c}{Correct} & \multicolumn{2}{c}{Naive} \\  \smallskip 
& CAR & Working model & Method & Bias & SD & SE  & CP  & SE  & CP  \\ \hline
%\\ \smallskip 
Case I & Simple  & 
& Sample mean & 0.09 & 3.13 & 3.12 & 94.86 &  &  \\ \cline{4-8}
 % \smallskip 
&& Oracle & AIPW & 0.07 & 2.66 & 2.67 & 95.00 &  & \\ \cline{4-8}
 %\smallskip 
&& Logistic & AIPW  & 0.05 & 2.74 & 2.76 & 95.14 & &  \\ %\smallskip 
& & &LC    & 0.06 & 2.70 & 2.71 & 95.12 &  &  \\ %\smallskip 
&& & JC   & 0.06 & 2.67 & 2.67 & 95.04 & &  \\ \cline{4-8}
 %\smallskip 
&& Random forest & CF  & 0.08  & 2.83 & 2.84 & 95.00 &  &  \\ 
%\smallskip 
&& & JC  & 0.07 & 2.68 & 2.69 & 95.04 &  &  \\ \cline{2-10}
%\smallskip 
& Permuted block & & Sample mean & 0.07 & 2.80 & 2.81 & 95.00 & 3.12& 97.36 \\ \cline{4-10}
  % \smallskip 
&& Oracle & AIPW & 0.07 & 2.65 & 2.66 & 95.16 &  &  \\ 
\cline{4-10}
&& Logistic & AIPW  & 0.06 & 2.75 & 2.75 & 95.06 & 2.76 & 95.08 \\ 
%\smallskip 
&& & LC    & 0.06 & 2.69 & 2.69 & 95.04 & 2.71 & 95.18 \\ %\smallskip 
& & & JC    & 0.06 & 2.67 & 2.67 & 94.90 &  &  \\ 
 \cline{4-10}
&& Random forest& CF & 0.07 & 2.71 & 2.73 & 95.14 & 2.84 & 96.18 \\ %\smallskip 
&& & JC & 0.07 & 2.67 & 2.69 & 95.00 &  &   \\ \cline{2-10}
&Minimization & & Sample mean & 0.09 & 2.83 & -- & -- & 3.12 & 96.56 \\ 
 \cline{4-10}
&& Oracle & AIPW & 0.10 & 2.68 & 2.66 & 94.58 &  &  \\ 
  \cline{4-10}
& & Logistic & AIPW  & 0.08 & 2.76 & -- & -- & 2.76 & 94.68 \\ %\smallskip 
& & & LC   & 0.08 & 2.71 & -- & -- & 2.71 & 94.66 \\ %\smallskip 
& & & JC    & 0.09 & 2.69 & 2.67 & 94.58 & &  \\ \cline{4-10}
&& Random forest & CF  & 0.09 & 2.74 & -- & -- & 2.84 & 95.60 \\ %\smallskip 
&& & JC  & 0.09 & 2.70 & 2.69 & 94.82 & &  \\ 
 \hline
 Case II & Simple  & 
& Sample mean & -0.91 & 3.32 & 3.29 & 94.92 &  &  \\ \cline{4-8}
 % \smallskip 
&& Oracle & AIPW & -1.03 & 2.59 & 2.55 & 94.36 &  & \\ \cline{4-8}
 %\smallskip 
&& Logistic & AIPW  & -0.69 & 3.29 & 3.25 & 94.68 & &  \\ %\smallskip 
& & &LC    & -0.68 & 3.29 & 3.25 & 94.64 &  &  \\ %\smallskip 
&& & JC   & -1.08 & 3.02 & 2.96 & 94.44 & &  \\ \cline{4-8}
 %\smallskip 
&& Random forest & CF  & -1.11  & 2.79 & 2.80 & 95.06 &  &  \\ 
%\smallskip 
&& & JC  & -1.15 & 2.75 & 2.74 & 94.64 &  &  \\ \cline{2-10}
%\smallskip 
& Permuted block & & Sample mean & 0.05 & 2.96 & 2.97 & 94.74 & 3.29 & 96.68 \\ \cline{4-10}
  % \smallskip 
&& Oracle & AIPW & -0.15 & 2.59 & 2.54 & 94.68 &  &  \\ 
\cline{4-10}
&& Logistic & AIPW  & 0.12 & 2.96 & 2.96 & 94.74 & 3.25 & 96.36 \\ 
%\smallskip 
&& & LC    & 0.13 & 2.96 & 2.96 & 94.68 & 3.25 & 96.34 \\ %\smallskip 
& & & JC    & 0.10 & 2.96 & 2.95 & 94.62 &  &  \\ 
 \cline{4-10}
&& Random forest& CF & -0.01 & 2.74 & 2.75 & 95.12 & 2.80 & 95.26 \\ %\smallskip 
&& & JC & -0.05 & 2.74 & 2.74 & 95.12 &  &   \\ \cline{2-10}
&Minimization & & Sample mean & -0.59 & 3.01 & -- & -- & 3.28 & 96.76 \\ 
 \cline{4-10}
&& Oracle & AIPW & -0.49 & 2.62 & 2.54 & 94.38 &  &  \\ 
  \cline{4-10}
& & Logistic & AIPW  & -0.53 & 3.01 & -- & -- & 3.24 & 96.48 \\ %\smallskip 
& & & LC   & -0.53 & 3.01 & -- & -- & 3.24 & 96.48 \\ %\smallskip 
& & & JC    & -0.58 & 3.00 & 2.95 & 94.68 & &  \\ \cline{4-10}
&& Random forest & CF  & -0.52 & 2.78 & -- & -- & 2.80 & 95.10 \\ %\smallskip 
&& & JC  & -0.48 & 2.77 & 2.74 & 94.68 & &  \\ 
 \hline
 \end{tabular}
\end{table}
	
\begin{table}[htbp!]
				\centering
				\caption{Analysis results in the real data example. Entries (except for the p-values) have all been multiplied by 100. 
    Permuted block randomization is applied, stratified according
to country,  baseline HbA1C,  LDL, and sulfonylurea or meglitinide use.
The endpoint is an indicator of having an average of more than two hypoglycemia events in 30 days. Adjusted 	covariates in the analysis stage are age, sex, BMI, country, baseline HbA1C, LDL cholesterol,  sulphonylurea or meglitinide use, and baseline hypoglycemic event count. The naive SE and p-value use the method under simple randomization. 
					The entry under naive SE or p-value is empty when naive = correct (i.e., when JC is applied).}\label{tb: BIAQ} 
 \begin{tabular}{cccccccc}
					\toprule
     & & & \multicolumn{2}{c}{Correct} & &\multicolumn{2}{c}{Naive} \\ 
					Working model & Method & Estimate  & SE  & p-value  && SE   &  p-value  \\ 
					\hline 
					% \multicolumn{8}{c}{\emph{Continuous endpoint: HbA1C}} \\ \hline 
					& Sample mean & -8.58   & 2.46       & $<$0.001        &  & 2.73     & 0.002       \\
					\hline 
					 Logistic & AIPW  & -8.75   & 2.39       & $<$0.001        &  & 2.50     & $<$0.001      \\
					& LC          & -8.73   & 2.39       & $<$0.001        &  & 2.50     & $<$0.001      \\
					& JC          & -8.94   & 2.39       & $<$0.001        &  &        &          \\
					\hline  
					Random forest & CF        & -8.51   & 2.45       & 0.001        &  & 2.70      & 0.002      \\
					& JC          & -8.26   & 2.42      & 0.001        &  &      &          \\ 
     % \hline\hline
					% \multicolumn{8}{c}{\emph{Binary endpoint: I(HbA1C$< $7\%)}} \\ \hline 
					% & Sample Mean& 11.65 & 4.92 & 0.018 &  & 5.18 & 0.025 \\
					% \hline 
					% Heterogeneous GLM & AIPW        & 12.86 & 4.80  & 0.007 &  & 4.93 & 0.009 \\
					% &LC          & 12.88 & 4.80  & 0.007 &  & 4.93 & 0.009 \\
					% &JC          & 12.46 & 4.81 & 0.010  &  &      &       \\
					% \hline 
					% Homogeneous GLM &AIPW        & 12.99 & 4.81 & 0.007 &  & 4.94 & 0.009 \\
					% &LC          & 13.04 & 4.79 & 0.007 &  & 4.92 & 0.008 \\
					% &JC          & 12.59 & 4.81 & 0.009 &  &      &       \\
					% \hline 
					% Random Forest& CF       & 12.74 & 5.02 & 0.011 &  & 5.14 & 0.013 \\
					% & JC         & 10.99 & 4.93 & 0.026 &  &      &  \\
					\bottomrule
				\end{tabular}
			\end{table}

\newpage
\bibliographystyle{apalike}
\bibliography{reference}

\newpage
\appendix
\doublespacing

\section{Asymptotic Linearity and Normality of AIPW Estimator under CAR}

\subsection{Helper Lemmas for AIPW Under CAR}

We first prove several lemmas that we use throughout the proof of Theorem \ref{theo:asymptotic-normality}, especially part (i). Throughout, we use the following notation: $\rho_{P} := \norm{\cdot}_{L^2(P)}$, and $1_a: A \to 1(A = a)$. Furthermore, we define the following function classes, for $\mathcal{F}_a$ defined as in Assumption \ref{assump: simple}:
\begin{align*}
	\mathcal{H}_a &:= \{(a', x) \to 1_a(a') \mu_a(\X): \mu_a \in \mathcal{F}_a \} \\
	\mathcal{G}_a &:= \{z \to \int \mu_a(\X) dP_0(\X|\Z): \mu_a \in \mathcal{F}_a\} \\
	\mathcal{D}_{\mathcal{F}} &:= \{f_1 - f_2: f_1, f_2 \in \mathcal{F}\} \quad \text{where $\mathcal{F}$ is any function class} \\
	\mathcal{D}_{\mathcal{F}}(\delta) &:= \{f_1 - f_2: f_1, f_2 \in \mathcal{F}, \rho_{P_0}(f_1 - f_2) < \delta\} \quad \text{where $\mathcal{F}$ is any function class}.
\end{align*}
\begin{lemma}\label{lemma:pi-a}
	For any $f$, $\int 1(A_i=a) f(X_i) dP(A_i, X_i) = \pi_a \int f(X) dP(X)$.
\end{lemma}
\begin{proof}
Note that Assumption \ref{assump: car} (i) and (ii) imply that $A_i$ is independent of $\{(\Y_i, \X_i, \Z_i), i=1,...,n\}$, and thus $A_i\perp \X_i$. This directly implies the result.
\end{proof}

\begin{lemma}\label{lemma:bounded-g}
	The function class $\mathcal{G}_a$ is $P_0$-Donsker.
\end{lemma}
\begin{proof}
	First, let $\mathcal{S} := \{z \to s_{\theta}(z) \equiv\sum_{\ell=1}^L \theta_{\ell} 1(z = \ell): \theta \in \Theta \subset \mathbb{R}^L \}$. Second, for every strata $z$, we know that $\int\mu_a(\X)dP_0(\X|\Z) < \infty$ since we have assumed that $\int \mu_a(\X) dP_0(\X) < \infty$ for all $\mu_a \in \mathcal{F}_a$, and each strata $z$ has positive probability. Therefore, $\mathcal{G}_a \subseteq \mathcal{S}$. Therefore, it suffices to show that $\mathcal{S}$ is $P_0$-Donsker. This follows by \citet{vandervaartAsymptoticStatistics1998} Example 19.7, because each function $s_{\theta}(z) \in \mathcal{S}$ is Lipschitz in its $L$-dimensional indexing parameter $\theta$, i.e.,
	\begin{align*}
		|s_{\theta_1}(z) - s_{\theta_2}(z)| \leq \sup_{\ell} | \theta_{1,\ell} - \theta_{2,\ell} | \leq \norm{\theta_1 - \theta_2}_2.
	\end{align*}
\end{proof}

\begin{lemma}\label{lemma:g-2norm}
	Define the functions $\hat{g}_a : z \to \int \hat\mu_a(\X) dP_0(\X|\Z)$, and $g_a: z \to \int \mu_a(\X) dP_0(\X|\Z)$. Then $P_0 (\hat{g}_a - g_a)^2 \leq P_0(\hat{\mu}_a - \mu_a)^2$.
\end{lemma}
\begin{proof}
		Using only properties of iterated expectation and the law of total variance:
		\begin{align*}
			P_0 (\hat{g}_a - g_a)^2 &\equiv \norm{\int \hat{\mu}_a(\X) dP_0(\X|\Z) - \int \mu_a(\X) dP_0(\X|\Z)}_{L^2(P_0)}^2 \\
			&= \norm{\int \big\{\hat{\mu}_a(\X) - \mu_a(\X)\big\} dP_0(\X|\Z)}_{L^2(P_0)}^2 \\
			&= P_{0} \left(\left[P_0(\mu_a(\X) - \mu'_a(\X) \big| \Z)\right]^2 \right) \\
                &\leq P_{0} \left[P_0 \left\{ (\mu_a(\X) - \mu'_a(\X))^2 \big| \Z \right\} \right]  \\
			&= P_0 (\hat{\mu}_a - \mu_a)^2,
		\end{align*}		
  where the inequality is from $\left[P_0(\mu_a(\X) - \mu'_a(\X) \big| \Z)\right]^2\leq P_0 \left\{ (\mu_a(\X) - \mu'_a(\X))^2 \big| \Z \right\} $ due to a conditional version of Jensen's inequality. 
\end{proof}

\subsubsection{Translating Between Covering Numbers}
\begin{lemma}\label{lemma:h-cover}
The covering number for the function class $\mathcal{H}_a$ with respect to $\rho_{P_0}$ is no more than that of the function class $\mathcal{F}_a$ with respect to $\rho_{P_0}$.
	\begin{align*}
				N(\epsilon, \mathcal{H}_a, \rho_{P_0}) \leq N(\epsilon, \mathcal{F}_a, \rho_{P_0}).
	\end{align*}
\end{lemma}
\begin{proof} By definition, all functions $h_a \in \mathcal{H}_a$ can be written as $1_a\mu_a$ for some $\mu_a \in \mathcal{F}_a$. Let $T_{\epsilon}$ be a minimal $\epsilon$-cover of $\mathcal{F}_a$. Take any $h_a \equiv 1_a \mu_a$. We know that for all $\mu_a$, there exists a $\mu_a' \in T_{\epsilon}$ such that $\rho_{P_0}(\mu_a - \mu_a') < \epsilon$.  taking $h'_a = 1_a \mu_a'$, we have that
\begin{align*}
	\rho_{P_0}\left({h_a - h_a'}\right)^2 &\equiv \int 1_a(\mu_a - \mu_a')^2 dP_0 \leq \int (\mu_a - \mu_a')^2 dP_0 \equiv \rho_{P_0}\left(\mu_a - \mu_a'\right)^2 < \epsilon^2\\
	&\implies \rho_{P_0}\left({h_a - h_a'}\right) < \epsilon.
\end{align*}
Thus, the covering number for $\mathcal{H}_a$ is no more than the covering number for $\mathcal{F}_a$.
\end{proof}
\begin{lemma}\label{lemma:dh-cover}
	With $\delta > 0$, and $\delta' = \delta/\sqrt{\pi_a}$,
	\begin{align*}
		N(\epsilon, \mathcal{D}_{\mathcal{H}_a}(\delta), \rho_{P_n}) \leq N(\epsilon, \mathcal{D}_{\mathcal{F}_a}(\delta'), \rho_{P_n}).
	\end{align*}
\end{lemma}
\begin{proof}
    This is very similar to Lemma \ref{lemma:h-cover}, but requires more care because the size of the function class is dependent on $\delta$, so it is not guaranteed that an $h \in \mathcal{D}_{\mathcal{H}_a}(\delta)$ can be written as $1_a f$ for some $f \in \mathcal{D}_{\mathcal{F}_a}(\delta)$. We need to modify $\delta$ for the $\mathcal{D}_{\mathcal{F}_a}$ class.
    
    Let $T_{\epsilon}$ be a minimal $\epsilon$-cover for $\mathcal{D}_{\mathcal{F}_a}(\delta')$. Then $\forall f \in \mathcal{D}_{\mathcal{F}_a}(\delta')$, $\exists f' \in T_{\epsilon}$ such that $\norm{f-f'}_{L^2(P_n)} < \epsilon$. Take an arbitrary $h \in \mathcal{D}_{\mathcal{H}_a}(\delta)$. Then we can write $h$ as $1_a f$ for some $f \in \mathcal{D}_{\mathcal{F}_a}(\delta')$, because (recalling Lemma \ref{lemma:pi-a}): 
	\begin{align*}
		\rho^2_{P_0}(h) &= \rho^2_{P_0}(1_af) = P_0 1_a f^2 = \pi_a P_0 f^2 = \pi_a \rho^2_{P_0}(f) < \pi_a (\delta')^2 = \delta^2 \\
		&\implies \rho_{P_0}(h) < \delta.
	\end{align*}
	Furthermore, we have that, using $h' = 1_a f'$ for the $f'$ that covers $f$,
	\begin{align*}
		\rho^2_{P_n}(h - h') = \rho^2_{P_n}(1_a(f - f')) &= \frac{1}{n} \sum_{i=1}^{n} 1_a(A_i) (f(\X_i) - f'(\X_i))^2 \\
		&\leq \frac{1}{n} \sumn (f(\X_i) - f'(\X_i))^2 \leq \rho^2_{P_n}(f - f') < \epsilon^2.
	\end{align*}
	Since $\rho_{P_n}(h - h') < \epsilon$, the covering number for $\mathcal{D}_{\mathcal{H}_a}(\delta)$ is no more than that of $\mathcal{D}_{\mathcal{F}_a}(\delta')$.
\end{proof}

\subsubsection{Studying a non-standard Empirical Process}
In this subsection, we prove lemmas that allow us to study the empirical process $\sqrt{n}(P_n - P_0^n)$, where  $P_0^n(\cdot) = \frac1n \sumn P_0(\cdot|\mathcal{A}_n,\mathcal{Z}_n)$.

\begin{lemma}\label{lemma:cond-empirical}
	The following bound holds, for any subset $\mathcal{S} \subseteq \mathcal{D}_{\mathcal{H}_a}$ such that $\mathcal{S}$ is a symmetric function class, i..e, $\mathcal{S} = \mathcal{S} \cup -\mathcal{S}$:
	\begin{align*}
		\E(\norm{P_n - P_0^n}_{\mathcal{S}} | \mathcal{A}_n,\mathcal{Z}_n) \leq 2\E (\norm{R_n}_{\mathcal{S}}|\mathcal{A}_n,\mathcal{Z}_n)
	\end{align*}
  where $R_n$ is the Rademacher process with $\varepsilon_i$ as Rademacher random variables, i.e.,
  \begin{align*}
  	R_n(s) \equiv \frac{1}{n} \sum_{i=1}^{n} \varepsilon_i s(A_i, \X_i).
  \end{align*}
\end{lemma}
\begin{proof}
	This result is due to the symmetrization argument in \citet{vandervaartWeakConvergenceEmpirical1996} Lemma 2.3.1. In typical symmetrization proofs, we have a $P_0$ and an unconditional expectation. However, since the external and internal expectations are both conditional on $(\mathcal{A}_n, \mathcal{Z}_n)$, and because $\mathcal{X}_n | \mathcal{A}_n,\mathcal{Z}_n$ are mutually independent (but not necessarily identically distributed), the result still holds, as we show. Let $X_i'$ be drawn from the same distribution as $X_i$, for $i = 1, \dots, n$. Thus, $\E(s(A_i, X_i)|\mathcal{A}_n, \mathcal{Z}_n) = \E(s(A_i, X_i')|\mathcal{A}_n, \mathcal{Z}_n)$. For fixed values $X_1, \dots, X_n$:
 \begin{align*}
     \norm{P_n - P_0^n}_{\mathcal{S}} &= \sup_{s \in \mathcal{S}} \left|\frac1n\sum_{i=1}^{n} s(A_i, X_i) - \E(s(A_i, X_i') | \mathcal{A}_n, \mathcal{Z}_n)\right| \\
     &= \sup_{s \in \mathcal{S}} \frac1n \left|\E\left(\sum_{i=1}^{n} s(A_i, X_i) - s(A_i, X_i') \Bigg| \mathcal{A}_n, \mathcal{Z}_n\right)\right| \\
     &\leq \E\left(\sup_{s \in \mathcal{S}} \Bigg|\frac1n\sum_{i=1}^{n} s(A_i, X_i) - s(A_i, X_i') \Bigg|\Bigg| \mathcal{A}_n, \mathcal{Z}_n \right).
 \end{align*}
 Let $\epsilon_i$, $1, \dots, n$ be independent Rademacher random variables. Let $\E^*$ be the outer expectation over both $\X_1, \dots, \X_n$ and $\X'_1, \dots, \X'_n$. Then taking the same conditional expectation over $\X_1, \dots, \X_n$ as well,
 \begin{align*}
      \E(\norm{P_n - P_0^n}_{\mathcal{S}} | \mathcal{A}_n,\mathcal{Z}_n) \leq \E^*\left(\sup_{s \in \mathcal{S}} \Bigg|\frac1n\sum_{i=1}^{n} \epsilon_i \left[s(A_i, X_i) - s(A_i, X_i')\right] \Bigg|\Bigg| \mathcal{A}_n, \mathcal{Z}_n \right)
 \end{align*}
 because $\epsilon_i$ just keeps a function $s(A_i, X_i)$ the same or flips it to its negative value $-s(A_i, X_i)$. Since $X_i$ and $X_i'$ have the same distribution, this does not change the supremum over $\mathcal{S}$. By the triangle inequality,
 \begin{align*}
     \E^*\left(\sup_{s \in \mathcal{S}} \Bigg|\frac1n\sum_{i=1}^{n} \epsilon_i \left[s(A_i, X_i) - s(A_i, X_i')\right] \Bigg|\Bigg| \mathcal{A}_n, \mathcal{Z}_n \right) &\leq 2 \E^*\left(\sup_{s \in \mathcal{S}}  \Bigg|\frac1n\sum_{i=1}^{n} \epsilon_i s(A_i, X_i) \Bigg|\Bigg| \mathcal{A}_n, \mathcal{Z}_n \right).
 \end{align*}
 This completes the proof.
\end{proof}
\begin{lemma}\label{lemma:dudley}
	For any $\mathcal{S}$ as defined in Lemma \ref{lemma:cond-empirical},
	\begin{align*}
		\E(\norm{R_n}_{\mathcal{S}}|\mathcal{X}_n,\mathcal{A}_n, \mathcal{Z}_n) \leq 32 n^{-1/2} \int_0^{\infty} \sqrt{\log N(\epsilon, \mathcal{S},\rho_{P_n})} d\epsilon.
	\end{align*}
\end{lemma}
\begin{proof}
	This is a direct application of Dudley's entropy integral bound. We first note that with the definition of the Rademacher process in Lemma \ref{lemma:cond-empirical}, $R_n(s)$ can be equivalently expressed as $n R_n(s) = \sum_{i=1}^{n} \theta_i \varepsilon_i$, where $\theta_i \equiv s(A_i, \X_i)$, and $\theta \in \mathbb{T}_n \subseteq \mathbb{R}^n$, where $\mathbb{T}_n = \{(s(A_1, \X_1), ..., s(A_n, \X_n)): s \in \mathcal{S}\}$. The Rademacher process is a sub-Gaussian process with respect to the Euclidean metric $\norm{\cdot}_{2}$ and its indexing parameter $\theta \in \mathbb{R}^n$ (see Lemma 2.2.7 and discussion in \citet{vandervaartWeakConvergenceEmpirical1996}). Therefore, we can apply Theorem 5.22 from \citet{wainwrightHighDimensionalStatisticsNonAsymptotic2019} to upper bound the expected supremum (where the expectation is taken over the Rademacher random variables, with the data $(\mathcal{X}_n, \mathcal{A}_n, \mathcal{Z}_n)$ fixed) of the Rademacher process by
	\begin{align*}
		\E(n \norm{R_n}_{\mathcal{S}}|\mathcal{X}_n,\mathcal{A}_n,\mathcal{Z}_n) \leq 32 \int_0^{\infty} \sqrt{\log N(u, \mathbb{T}_n, \norm{\cdot}_2}) du.
	\end{align*}
	Finally, we note that a minimal $\epsilon$-cover for $\mathbb{T}_n$ with respect to $\norm{\cdot}_2$ is an $\epsilon n^{-1/2}$-cover for $\mathcal{S}$ with respect to the pseudometric $L^2(P_n) \equiv \rho_{P_n}$. This is because for any two $\theta, \theta' \in \mathbb{T}_n$ such that $\norm{\theta - \theta'}_2 < \epsilon$, we have $\sqrt{n} \norm{s - s'}_{L^2(P_n)} = \norm{\theta - \theta'}_2 < \epsilon$, where $s, s' \in \mathcal{S}$ are the corresponding functions to the indexing parameters $\theta, \theta'$. Therefore,
	\begin{align*}
		\E(n \norm{R_n}_{\mathcal{S}}|\mathcal{X}_n,\mathcal{A}_n,\mathcal{Z}_n) &\leq 32 \int_0^{\infty} \sqrt{\log N(un^{-1/2}, \mathcal{S}, \rho_{P_n}}) du \\
		&\leq 32 \sqrt{n} \int_0^{\infty} \sqrt{\log N(\epsilon, \mathcal{S}, \rho_{P_n}}) d\epsilon.
	\end{align*}
	Dividing by $n$ on both sides completes the proof.
	\end{proof}

\begin{lemma}\label{lemma:h-equicontinuous}
	Define $P_0^n(\cdot) = \frac1n \sumn P_0(\cdot|\mathcal{A}_n,\mathcal{Z}_n)$, and the empirical process $\mathbb{G}'_n(\cdot) = \sqrt{n}(P_n - P_0^n) (\cdot)$. Then $\mathbb{G}'_n$ is asymptotically uniformly equicontinuous on the class $\mathcal{H}_a$ with respect to the pseudometric $\rho_{P_0}$.
\end{lemma}
\begin{proof}
	To show that $\mathbb{G}'_n$ is asymptotically uniform equicontinuous, we must show that, for any deterministic sequence $\delta_n \to 0$, $\norm{\mathbb{G}'_n}_{\mathcal{D}_{\mathcal{H}_a}(\delta_n)} = o_P(1)$. By Markov's inequality, it suffices to show that $\E(\norm{\mathbb{G}'_n}_{\mathcal{D}_{\mathcal{H}_a}(\delta_n)}) = o(1)$.
	
	By Lemmas \ref{lemma:dh-cover}, \ref{lemma:cond-empirical}, and \ref{lemma:dudley}, we have that, for any deterministic $\delta_n \to 0$,
	\begin{align*}
		\E(\norm{\mathbb{G}'_n}_{\mathcal{D}_{\mathcal{H}_a}(\delta_n)}|\mathcal{A}_n,\mathcal{Z}_n) &= \sqrt{n} \cdot \E(\norm{P_n - P_0^n}_{\mathcal{D}_{\mathcal{H}_a}(\delta_n)} | \mathcal{A}_n,\mathcal{Z}_n)\nonumber\\
		&\leq 2\sqrt{n} \cdot \E (\norm{R_n}_{\mathcal{D}_{\mathcal{H}_a}(\delta_n)}|\mathcal{A}_n,\mathcal{Z}_n) \quad &&\text{(Lemma \ref{lemma:cond-empirical})} \\
		&\leq 64 \E \left(\int_0^{\infty} \sqrt{\log N(\epsilon, \mathcal{D}_{\mathcal{H}_a}(\delta_n), \rho_{P_n}}) d\epsilon\Big|\mathcal{A}_n,\mathcal{Z}_n\right) \quad &&\text{(Lemma \ref{lemma:dudley})} \\
		&\leq 64  \E\left(\int_0^{\infty} \sqrt{\log N(\epsilon, \mathcal{D}_{\mathcal{F}_a}(\delta'_n), \rho_{P_n})} d\epsilon \Big|\mathcal{A}_n,\mathcal{Z}_n\right) && \text{(Lemma \ref{lemma:dh-cover})}.
	\end{align*}
	Let $\epsilon_n = \sup_{d\in\mathcal{D}_{\mathcal{F}_a}(\delta'_n)} \rho_{P_n}(d)$. Then we only need to integrate to $\epsilon_n$, the diameter of $\mathcal{D}_{\mathcal{F}_a}(\delta_n')$, because beyond that the covering number is 1. Let $Q$ be any finitely supported probability distribution. Then,
	\begin{align*}
		\E(\norm{\mathbb{G}'_n}_{\mathcal{D}_{\mathcal{H}_a}(\delta_n)}|\mathcal{A}_n,\mathcal{Z}_n) &\leq 64\E\left(\int_0^{\epsilon_n}  \sqrt{\log N(\epsilon, \mathcal{D}_{\mathcal{F}_a}(\delta'_n), \rho_{P_n})} d\epsilon \Big|\mathcal{A}_n,\mathcal{Z}_n\right) \\
		&\leq 64 \E\left(\int_0^{\epsilon_n/2} \sup_{Q} \sqrt{\log N(\epsilon, \mathcal{F}_a, \rho_{Q})} d\epsilon \Big|\mathcal{A}_n,\mathcal{Z}_n \right) \\
  \implies \E(\norm{\mathbb{G}'_n}_{\mathcal{D}_{\mathcal{H}_a}(\delta_n)}) &\leq 64 \E\left(\int_0^{\epsilon_n/2} \sup_{Q} \sqrt{\log N(\epsilon, \mathcal{F}_a, \rho_{Q})} d\epsilon \right)
	\end{align*}
	where we have used the fact that $N(\epsilon, \mathcal{D}_{\mathcal{F}_a}(\delta_n'), \rho_{Q}) \leq N(\epsilon/2, \mathcal{F}_a, \rho_{Q})^2$.
	Thus, by an application of the dominated convergence theorem and using Assumption \ref{assump: simple}, to show that $\E(\norm{\mathbb{G}'_n}_{\mathcal{D}_{\mathcal{H}_a}(\delta_n)}) = o(1)$, it suffices to show that $e_n = o_P(1)$, or equivalently that $e_n^2 = o_P(1)$.
	\begin{align*}
		\epsilon^2_n &= \sup_{d\in\mathcal{D}_{\mathcal{F}_a}(\delta'_n)} P_n d^2 \\
		&= \sup_{d\in\mathcal{D}_{\mathcal{F}_a}(\delta'_n)} \big\{(P_n - P_0) d^2 + P_0 d^2\big\} \\
		&\leq \sup_{d\in\mathcal{D}_{\mathcal{F}_a}(\delta'_n)}(P_n - P_0) d^2 + \sup_{d\in\mathcal{D}_{\mathcal{F}_a}(\delta'_n)}P_0 d^2 \\
		&\leq \sup_{d\in\mathcal{D}^2_{\mathcal{F}_a}(\delta'_n)}(P_n - P_0) d + (\delta_n')^2.
	\end{align*}
	Since $\delta_n \to 0$, which implies $\delta'_n \to 0$, it suffices to show that $\sup_{d\in\mathcal{D}^2_{\mathcal{F}_a}(\delta'_n)}(P_n - P_0) d \to 0$ in probability. To show this, we cite results from Chapter 2.10 \textit{Permanence of the Donsker Property} in \citet{vandervaartWeakConvergenceEmpirical1996}. Since $\mathcal{F}_a$ is $P_0$-Donsker, this implies that $\mathcal{D}_{\mathcal{F}_a} \equiv \mathcal{F}_a - \mathcal{F}_a$ is also $P_0$-Donsker (Theorem 2.10.6), and $\mathcal{D}^2_{\mathcal{F}_a}$ is $P_0$-Glivenko Cantelli (Lemma 2.10.14). Since, $\mathcal{D}^2_{\mathcal{F}_a}(\delta'_n) \subseteq \mathcal{D}^2_{\mathcal{F}_a}$, this completes the proof.
\end{proof}

\subsubsection{Asymptotic Negligibility of Two Empirical Process Terms}

We now prove Lemma \ref{lemma:random-func} that allows us to use the results of previous lemmas to show that the empirical process terms $(P_n - P_0)(\hat{g}_a - g_a)$ (Lemma \ref{lemma:g-negligible}) and $(P_n - P_0^n)(1_a\hat\mu_a - 1_a \mu_a)$ (Lemma \ref{lemma:h-negligible}) are both asymptotically negligible. Lemma \ref{lemma:random-func} is a version of \citet{vandervaartAsymptoticStatistics1998} Lemma 19.24 where the estimated function falls into $\mathcal{F}$ with probability tending to 1 rather than deterministically. Note that this automatically holds for $(P_n - P_0) (\hat{\mu}_a - \mu_a)$ since $\hat{\mu}_a$ and $\mu_a$ satisfies the Donsker condition of Assumption \ref{assump: simple}.

\begin{lemma}\label{lemma:random-func}
	Let $(\mathcal{F}, \rho)$ be totally bounded for some pseudometric $\rho$ such that $\rho(f)^2 \leq \rho_{P_0}(f)^2$. Let $\mathbb{X}_n$ be some empirical process that is asymptotically uniformly equicontinuous on $\mathcal{F}$ with respect to $\rho$. Then, if $P(\hat{f} \in \mathcal{F}) \to 1$ as $n \to \infty$ and $P_0(\hat{f} - f)^2 \to 0$ in probability, then $\mathbb{X}_n(\hat{f} - f) = o_P(1)$.
\end{lemma}
\begin{proof}
	First, note that for any sequence $\Delta_n = o_P(1)$, we can always find a deterministic sequence $\delta_n \to 0$ such that $P(\Delta_n > \delta_n) \to 0$ as $n \to \infty$. In this case, let $\Delta_n = 2P_0(\hat{f} - f)^2$, and notice that $\rho(\hat{f} - f)^2 < \Delta_n$ by assumption. Define the following events:
	\begin{align*}
		A_n &= \{|\mathbb{X}_n(\hat{f} - f)| > \epsilon\} \\
		B_n &= \{\norm{\mathbb{X}_n}_{\mathcal{D}_{\mathcal{F}}(\delta_n)} > \epsilon\} \\
		C_n &= \{\Delta_n \leq \delta_n\} \\
		D_n &= \{\hat{f} \in \mathcal{F} \}
	\end{align*}
	Then we can construct the following inequality:
	\begin{align*}
		P(A_n)
		&\leq P(A_n \cap D_n \cap C_n) + P(C_n^c) + P(D_n^c) \\
		&\leq P(B_n) + P(C_n^c) + P(D_n^c).
	\end{align*}
	$P(B_n) \to 0$ by the asymptotic uniformly equicontinuous property of $\mathbb{X}_n$ on $\mathcal{F}$, and $P(C_n^c)$ and $P(D_n^c)$ both $\to 0$ as $n \to \infty$ by assumption. Thus, $\mathbb{X}_n(\hat{f} - f) = o_P(1)$. 
\end{proof}

\begin{remark}\label{remark:vdv1924}
Note that Lemma \ref{lemma:random-func} holds if $\mathcal{F}$ is a Donsker class with $\mathbb{X}_n \equiv \mathbb{G}_n \equiv \sqrt{n}(P_n - P_0)$, the standard empirical process, by definition of a Donsker class. This is a version of \citet{vandervaartAsymptoticStatistics1998}, Lemma 19.24, where the only difference is that $\hat{f}$ falls into $\mathcal{F}$ with probability tending to 1, rather than deterministically.
\end{remark}

\begin{lemma}\label{lemma:g-negligible}
	Define the new random variables $\hat{g}_a(Z_i) = \int \hat\mu_a(\X) dP_0(\X|\Z_i)$, and $g_a(Z_i) = \int \mu_a(\X) dP_0(\X|\Z_i)$. Under the conditions in Assumption \ref{assump: simple}, $(P_n - P_0)(\hat{g}_a - g_a) = o_P(1/\sqrt{n})$.
\end{lemma}
\begin{proof}
Note that $\Z_i$ are i.i.d. random variables, and that $\mathcal{G}_a$ is $P_0$-Donsker by Lemma \ref{lemma:bounded-g}. Furthermore, $P_0(\hat{g}_a \in \mathcal{G}_a) \geq P_0(\hat{\mu}_a \in \mathcal{F}_a) \to 1$, by Assumption \ref{assump: simple}, and $P_0(\hat{g}_a - g_a)^2 \leq P_0(\hat{\mu}_a - \mu_a)^2 \to 0$, combining Lemma \ref{lemma:g-2norm} with Assumption \ref{assump: simple}. Since $\mathcal{G}_a$ is $P_0$-Donsker, we know that $\mathbb{X}_n := \sqrt{n}(P_n - P_0)$ is as asymptotically uniform equicontinuous on $\mathcal{G}_a$ using the standard deviation pseudometric, which is guaranteed to exist: $\rho(f) := \sqrt{P_0(f - P_0f)^2} \leq \rho_{P_0}(f)$. Then using Remark \ref{remark:vdv1924} with $\mathcal{F} := \mathcal{G}_a$, we conclude that $\mathbb{X}_n(\hat{g}_a - g_a) = o_P(1)$, which implies $(P_n - P_0)(\hat{g}_a - g_a) = o_P(1/\sqrt{n})$.
\end{proof}

\begin{lemma}\label{lemma:h-negligible}
	Under the conditions in Assumption \ref{assump: simple}, $(P_n - P_0^n)(1_a\hat\mu_a - 1_a \mu_a) = o_P(1/\sqrt{n})$.
\end{lemma}
\begin{proof}
By Assumption \ref{assump: simple}, $\mathcal{F}_a$ is $P_0$-Donsker with respect to $\rho_{P_0}$, which tells us that $\left(\mathcal{F}_a, \rho_{P_0} \right)$ is totally bounded. Applying Lemma \ref{lemma:h-cover}, we have that $(\mathcal{H}_a, \rho_{P_0})$ is also totally bounded. From Lemma \ref{lemma:h-equicontinuous}, we know that the empirical process $\mathbb{G}'_n := \sqrt{n}(P_n - P_0^n)$ is asymptotically uniformly equicontinuous on $\mathcal{H}_a$.

Let $h_a = 1_a \mu_a$, and $\hat{h}_a = 1_a \hat{\mu}_a$. Then, $P_0(\hat{h}_a \in \mathcal{H}_a) \geq P_0(\hat{\mu}_a \in \mathcal{F}_a) \to 1$ by Assumption \ref{assump: simple}, and as in Lemma \ref{lemma:h-cover}, $P_0(\hat{h}_a - h_a)^2 \leq P_0(\hat{\mu}_a - \mu_a)^2 \to 0$ in probability also by Assumption \ref{assump: simple}. Then, applying Lemma \ref{lemma:random-func} using $\mathbb{X}_n := \mathbb{G}'_n$, $\mathcal{F} := \mathcal{H}_a$, and $\rho := \rho_{P_0}$, we have $\mathbb{G}'_n(1_a\hat\mu_a - 1_a \mu_a) = o_P(1)$, which implies that $(P_n - P_0^n)(1_a\hat\mu_a - 1_a \mu_a) = o_P(1/\sqrt{n})$.
\end{proof}

\newpage

\subsection{AIPW Influence Function: Proof of Theorem \ref{theo:asymptotic-normality} (i)}
\begin{proof}[Proof of Theorem \ref{theo:asymptotic-normality} (i)]
Here we find the influence function of the AIPW estimator under Assumptions \ref{assump: car}-\ref{assump: simple}.
\begin{align*}
	\hat{\theta}_{\rm AIPW,a} = \frac{1}{n} \sum_{i=1}^{n} \hat{\mu}_a(\X_i) + \frac{1}{n_a} \suma (\Yai - \hat{\mu}_a(\X_i)).
\end{align*}
Our goal is to show that $\hat{\theta}_{AIPW,a}$ is asymptotically linear for $\theta_a$ with influence function given by {$\phi_a(A_i, \X_i, Y_i) := \frac{I(A_i= a)}{\pi_a } \{ \Yai - \mu_a (\X_i) - (\theta_a  - P_0\mu_a (\X_i)) \} + \mu_a(\X_i) - P_0\mu_a (\X_i)$}.
% {the influence function expression should be like this. The previous one is only for the case under prediction unbiasedness.}
Then we have the following expression:
\begin{align*}
	&\hat{\theta}_{AIPW,a} - \theta_a \\
	&= \frac{1}{n} \sum_{i=1}^{n} \hat{\mu}_a(\X_i) + \frac{1}{n_a} \suma (\Yai - \hat{\mu}_a(\X_i)) - \theta_a \\
	&= \frac{1}{n} \sum_{i=1}^{n} \hat{\mu}_a(\X_i) + \frac{1}{n_a} \suma (\Yai - \hat{\mu}_a(\X_i))- \theta_a - \frac{1}{n} \sum_{i=1}^{n} \mu_a(\X_i) + \frac{1}{n} \sum_{i=1}^{n} \mu_a(\X_i) \\
	&= {\frac{1}{n} \sum_{i=1}^{n} \mu_a(\X_i)- P_0\mu_a(\X_i) +\underbrace{ P_0\mu_a(\X_i)   -   \theta_a + \frac{1}{n} \sum_{i=1}^{n} \Big\{ \hat{\mu}_a(\X_i) - \mu_a(\X_i)\Big\} +  \frac{1}{n_a} \suma (\Yai - \hat{\mu}_a(\X_i))}_{(\star)}.}
\end{align*}
Focusing on $(\star)$, we have that, by adding and subtracting $\frac{1}{n_a} \suma \mu_a(\X_i)$ (just rearranging going from the first to second lines below), 
\begin{align*}
	(\star) &= P_0\mu_a(\X_i)   -   \theta_a+ \frac{1}{n} \sum_{i=1}^{n} \hat{\mu}_a(\X_i) - \frac{1}{n} \sum_{i=1}^{n} \mu_a(\X_i) + \frac{1}{n_a} \suma \Big\{\Yai - \hat{\mu}_a(\X_i) \Big\} \\
	&\qquad + \frac{1}{n_a} \suma \Big\{\mu_a(\X_i) - \mu_a(\X_i) \Big\} \\
	&= \underbrace{\frac{1}{n_a} \suma \Big\{\Yai - \mu_a(\X_i) - ( \theta_a - P_0\mu_a(\X_i)  )\Big\}}_{(\star\star)}   \\
	&\qquad +	\underbrace{\frac{1}{n_a} \suma \mu_a(\X_i) - \frac{1}{n} \sumn \mu_a(\X_i) - \frac{1}{n_a} \suma \hat{\mu}_a(\X_i) + \frac{1}{n} \sumn \hat{\mu}_a(\X_i)}_{(\star\star\star)}.
\end{align*}
For $(\star\star)$, we would like to show that it is:
\begin{align*}
	(\star\star) = \frac{1}{n} \sumn \frac{I(A_i=a)}{\pi_a} \Big\{\Yai - \mu_a(\X_i)  - (\theta_a  - P_0\mu_a (\X_i)) \Big\} + o_P(1/\sqrt{n}).
\end{align*}
For this to hold, by an application of Slutsky's theorem, since $\hat{\pi}_{a} \to \pi_a$ in probability, it suffices to show that $\frac{1}{n} \sumn I(A_i=a) \left\{\Yai - \mu_a(\X_i)  - (\theta_a  - P_0\mu_a (\X_i)) \right\} = O_p(1/\sqrt{n})$. In Section \ref{app:asymptotic-normality}, we break up the influence function that we derive in this section into five parts. Parts 1, 2, and 3 can be added together $(\phi^1_a + \phi^2_a + \phi^3_a)$ to obtain $(\star\star)$. Using the results of Section \ref{app:asymptotic-normality} about the convergence in distribution of $\E_n (\phi^1_a)$ and $\E_n(\phi^3_a)$, together with Remark \ref{remark:phi2-op1} that shows convergence in distribution of $\E_n(\phi^2_a)$, we see that $(\star\star)$ is $O_P(1/\sqrt{n})$.

For $(\star\star\star)$, we would like to show that it is asymptotically negligible, i.e., $o_P(1/\sqrt{n})$. With these two results, we have the desired influence function.

\textit{Asymptotic Negligibility for $(\star\star\star)$}

By noting that $\suma \mu_a(\X_i) = \sumn 1(A_i=a) \mu_a(\X_i)$ (and the same for $\hat{\mu}_a$), we have
\begin{align*}
	(\star\star\star) &= \frac{1}{n} \sum_{i=1}^{n} \Big\{1(A_i=a)\mu_a(\X_i) \cdot \hat{\pi}_a^{-1} - \mu_a(\X_i) \Big\} - \frac{1}{n} \sum_{i=1}^{n} \Big\{1(A_i=a)\hat{\mu}_a(\X_i) \cdot \hat{\pi}_a^{-1} - \hat{\mu}_a(\X_i) \Big\} \\
        &= \frac{1}{n} \sumn 1(A_i=a) \Big\{\mu_a(\X_i) - P_0 \mu_a \Big\} \cdot \hat{\pi}_a^{-1} - (P_n - P_0) \mu_a \\ &\quad\quad - \frac{1}{n} \sumn 1(A_i=a) \Big\{\hat\mu_a(\X_i) - P_0 \hat\mu_a \Big\} \cdot \hat{\pi}_a^{-1} + (P_n - P_0) \hat\mu_a \\
        &= \hat\pi_a^{-1} \frac{1}{n} \sumn 1(A_i=a) \Big\{\mu_a(\X_i) - P_0 \mu_a - \hat{\mu}_a(\X_i) + P_0 \hat{\mu}_a \Big\} + (P_n - P_0)(\hat\mu_a - \mu_a) \\
        &=  (\pi_a^{-1} + o_p(1)) \cdot \underbrace{ \frac{1}{n} \sumn 1(A_i=a) \Big\{\mu_a(\X_i) - P_0 \mu_a - \hat{\mu}_a(\X_i) + P_0 \hat{\mu}_a \Big\}}_{\mathcal{M}} + o_P(1\sqrt{n})
 \end{align*}
 where $(P_n - P_0)(\hat\mu_a - \mu_a)$ is $o_P(1/\sqrt{n})$ by Remark \ref{remark:vdv1924} combined with Assumption \ref{assump: simple}(iii).
 Since $\hat{\pi}_a \to \pi_a$ in probability, by an application of the continuous mapping theorem and Slutsky's theorem, it suffices to show that $\mathcal{M} = o_P(1/\sqrt{n})$.
 
 $\mathcal{M}$ can be broken into three parts. We represent a sequence of random variables $T_1, .., T_n$ with the notation $T_1^n$. Throughout, we use the fact that for any function $f$, $ P_0 f(\X_i) | \mathcal{A}_n, \mathcal{Z}_n = P_0 f(\X_i)|\Z_i$, and that $(\X_i, \Z_i),i=1,\dots, n$ are i.i.d. random variables. We use the notation $P_0^n(\cdot)$ to represent the expectation conditional on $\mathcal{A}_n,\mathcal{Z}_n$.
 \begin{align*}
 	\mathcal{M} &=  \frac{1}{n} \sumn 1(A_i=a) \Big\{\mu_a(\X_i) - P_0 \mu_a - \hat{\mu}_a(\X_i) + P_0 \hat{\mu}_a \Big\} \\
 	&\quad\quad + \frac1n \sumn [1(A_i=a) - \pi_a] \big(P_0 \mu_a(\X_i) | \Z_i - P_0 \hat{\mu}_a(\X_i)|\Z_i \big) \\
 	&\quad\quad - \frac1n \sumn [1(A_i=a) - \pi_a] \big(P_0 \mu_a(\X_i) | \Z_i - P_0 \hat{\mu}_a(\X_i)|\Z_i \big) \\
 	&\quad\quad + \pi_a \big(P_0\hat\mu_a - P_0 \mu_Aa/\big) - \pi_a \big(P_0\hat\mu_a - P_0 \mu_a \big) \\
 	&= \underbrace{\frac1n \sumn 1(A_i=a) \big(\mu_a(\X_i) - P_0 \mu_a(\X_i)|\Z_i - \hat\mu_a(\X_i) + P_0 \hat{\mu}_a(\X_i)|\Z_i \big)}_{\mathcal{M}_1} \\
 	&\quad\quad + \underbrace{\frac1n \sumn [1(A_i=a) - \pi_a] \big(P_0 \mu_a(\X_i)|\Z_i - P_0 \hat{\mu}_a(\X_i)|\Z_i - P_0 \mu_a + P_0\hat{\mu}_a \big)}_{\mathcal{M}_2} \\
 	&\quad\quad + \underbrace{\frac1n \sumn \pi_a \big(P_0 \mu_a(\X_i)|\Z_i - P_0 \hat{\mu}_a(\X_i) |\Z_i + P_0 \hat{\mu}_a - P_0 \mu_a \big)}_{\mathcal{M}_3}.
 \end{align*}
 For term $\mathcal{M}_1$, note that for a generic $f(X_i)$,
 \begin{align*}
 	\frac1n \sumn 1(A_i=a)P_0f(\X_i)|\Z_i &= \frac1n \sumn 1(A_i=a) P_0 f(\X_i)|\mathcal{A}_n,\mathcal{Z}_n \\
 	&= \frac1n \sumn P_0 1_a(A_i) f(X_i)|\mathcal{A}_n,\mathcal{Z}_n.
 \end{align*}
 We write the short-hand notation for this empirical measure as $P_0^n(\cdot) = \frac1n \sumn P_0(\cdot|\mathcal{A}_n,\mathcal{Z}_n)$. Furthermore, $\frac1n \sumn P_0(f(X_i)|\mathcal{A}_n,\mathcal{Z}_n) = \frac1n \sum_{i=1}^{n} P_0(f(\X_i)|\Z_i)$.
 \begin{align*}
 	\mathcal{M}_1 &= \frac1n \sumn 1(A_i=a) \big(\mu_a(\X_i) - P_0 \mu_a(\X_i)|\Z_i - \hat\mu_a(\X_i) + P_0 \hat{\mu}_a(\X_i)|\Z_i \big) \\
 	&= - (P_n - P_0^n)(1_a\hat{\mu}_a - 1_a\mu_a).
 \end{align*}
 By Lemma \ref{lemma:h-negligible}, we know that this term is $o_P(1/\sqrt{n})$.
 
 For the second term $\mathcal{M}_2$, we will exploit the fact that we have assumed that there are finitely many strata levels $\l \in \{1, ..., L\}$, $L < \infty$. We have:
 \begin{align*}
 	\mathcal{M}_2 &= \frac1n \sumn [1(A_i=a) - \pi_a] \big(P_0 \mu_a(\X_i)|\Z_i - P_0 \hat{\mu}_a(\X_i)|\Z_i - P_0 \mu_a + P_0\hat{\mu}_a \big) \\
 	&= \frac1n \sum_{l = 1}^{L} \sum_{i:Z_i=l} [1(A_i=a) - \pi_a] \big((P_0 \mu_a(\X_i)|\Z_i=l) - (P_0 \hat{\mu}_a(\X_i)|\Z_i=l) - P_0 \mu_a + P_0\hat{\mu}_a \big).
 \end{align*}
 Note that because we have broken this up by strata level $l$, we no longer have dependence of $P_0 \mu_a(\X_i) | \Z_i$ and $P_0 \hat{\mu}_a(\X_i)|\Z_i$ on the index $i$. Therefore, it can be written as:
 \begin{align*}
 	\mathcal{M}_2 &=  \bigg(\sum_{l = 1}^{L} \sum_{i:Z_i=l} \frac1n[1(A_i=a) - \pi_a]\bigg) \big(P_0 (\hat{\mu}_a - \mu_a) - P_0 (\hat{\mu}_a(\X_i) - \mu_a(\X_i))|\Z_i=l\big).
 \end{align*}
 By assumption, we have $P_0 (\hat{\mu}_a - \mu_a)^2 \to 0$ in probability. Applying $L(P)$ norm inequalities, $P_0 (\hat{\mu}_a - \mu_a) \leq \norm{\hat{\mu}_a - \mu_a}_{L^1(P_0)} \leq \norm{\hat{\mu}_a - \mu_a}_{L^2(P_0)} \to 0$ in probability. We have a similar term conditional on $\Z_i = l$. Here, we note that using iterated expectation, we have
 \begin{align*}
 	P_0(\hat{\mu}_a - \mu_a)^2 = \sum_{l=1}^{L} P(\Z_i=l)(P_0(\hat{\mu}_a - \mu_a)^2 | \Z_i = l).
 \end{align*}
 Since the LHS is $o_P(1)$, and since $P(\Z_i=l) > 0$ for all $l$, and $(P_0(\hat{\mu}_a - \mu_a)^2 | \Z_i = l)$ is a non-negative term, we know that within each strata level $l$, that $P_0(\hat{\mu}_a - \mu_a)^2 | \Z_i = l$ must also be $o_P(1)$. Otherwise we would contradict the LHS begin $o_P(1)$. Therefore, we have
 \begin{align}
 	\mathcal{M}_2 &= o_P(1)\bigg(\sum_{l = 1}^{L} \sum_{i:Z_i=l} \frac1n[1(A_i=a) - \pi_a]\bigg)\nonumber \\
 	&= o_P(1) \bigg(\sum_{l = 1}^{L} [n_a(l)/n- \pi_a(n_l/n)]\bigg).\label{strata-ai}
 \end{align}
For each strata level $l$, we have assumed that $n_a(l)/n_l - \pi_a = O_p(1/\sqrt{n})$. Since there are finitely many $L$, we have 
 	$\mathcal{M}_2 = o_P(1) L O_p(1/\sqrt{n}) = o_P(1/\sqrt{n})$.
 
 For $\mathcal{M}_3$, we have
 \begin{align*}
 	\mathcal{M}_3 &= \pi_a \frac1n \sumn \bigg\{P_0 \mu_a(\X_i)|\Z_i - P_0 \hat{\mu}_a(\X_i) |\Z_i\bigg\} + \pi_a P_0 (\hat{\mu}_a - \mu_a).
 \end{align*}
Note that in the sum above, we have only i.i.d. terms, because $(X_i, Z_i)$ are i.i.d. random variables. Define the new random variable $\hat{g}_a(Z_i) = \int \hat\mu_a(\X_i) dP_0(\X_i|\Z_i)$, and $g_a(Z_i) = \int \mu_a(\X_i) dP_0(\X_i|\Z_i)$, as in Lemma \ref{lemma:g-negligible}. Using iterated expectation, we have that $P_0 \hat{g}_a = P_0 \hat{\mu}_a$, and $P_0 g_a = P_0 \mu_a$. Then $\mathcal{M}_3$ can be written as $\pi_a (P_n - P_0)(\hat{g}_a - g_a)$, which is $o_P(1/\sqrt{n})$ by Lemma \ref{lemma:g-negligible}.

Putting terms $\mathcal{M}_1, \mathcal{M}_2$ and $\mathcal{M}_3$ together we see that $\mathcal{M} = o_P(1/\sqrt{n})$, which means $(\star\star\star) = o_P(1/\sqrt{n})$.
\end{proof}

\newpage
\subsection{Asymptotic Variance Decomposition}\label{app:asymptotic-normality}

To show asymptotic normality of (possibly cross-fitted) AIPW estimator, we use the influence function from Theorem \ref{theo:asymptotic-normality} (i) and compute its asymptotic variance under condition (B). The derivations in this section will be used not only for understanding the asymptotic behavior under (B), but also under (U), and deriving conditions for guaranteed efficiency gain. When $\mu_a(\X) $ is linear in $x$, the variance agrees with \citet{Ye2021better}'s Theorem 3.

First, we rewrite $\phi_a$ as the sum of five components. Recall that our expression for $\phi_a$ was
\begin{align*}
	\phi_a(A_i, Y_i, \X_i) = \frac{1_a(A_i)}{\pi_a } \{ \Yai - \tilde{\mu}_a (\X_i) \} + \tilde{\mu}_a(\X_i) - \theta_a.
\end{align*}
where $\tilde{\mu}_a(\X_i) = \mu_a(\X_i) + \theta_a - P_0\mu_a(\X_i)$. We can split up this influence function into five parts that will make it easier to derive the asymptotic distribution of $\hat{\theta}_{\rm AIPW}$. Let
\begin{align*}
	\phi_a^1: (A_i,X_i,Y_i,Z_i) &\to \frac{1_a(A_i)}{\pi_a} \left[\Yai - \tilde{\mu}_a(\X_i) - \E\left\{\Yai - \tilde{\mu}_a(\X_i) \big|\Z_i \right\}\right] \\
	\phi_a^2: (A_i,X_i,Y_i,Z_i) &\to \frac{\left(1_a(A_i) - \pi_a\right)}{\pi_a} \E \left\{\Yai - \tilde{\mu}_a(\X_i) \big|\Z_i \right\} \\
	\phi_a^3: (X_i,Y_i,Z_i) &\to \E \left\{\Yai - \tilde{\mu}_a(\X_i) \big|\Z_i \right\} \\
	\phi_a^4: (X_i,Z_i) &\to \tilde{\mu}_a(\X_i) - \E\left\{\tilde{\mu}_a(\X_i) |\Z_i\right\} \\
	\phi_a^5: (X_i,Z_i) &\to \E\left\{\tilde{\mu}_a(\X_i) |\Z_i\right\} - \theta_a.
\end{align*}

Then by construction, $\phi_a = \phi_a^1 + \dots + \phi_a^5$. Let $\E_n\{M\} = (1/n) \sum_{i=1}^{n} M_i$, and let $\bm \phi^j = (\phi_1^j, ..., \phi_k^j)^T$ be the vector of the $j$th component of $\phi_a$ across all treatment groups $a = 1, ..., k$. Similarly, let $Y = (Y_{1}, ..., Y_{k})^T$, and $\tilde{\mu}(\X) = (\tilde{\mu}_1(\X), ..., \tilde{\mu}_k(\X))^T$. Conditional on $(\mathcal{A}_n,\mathcal{Z}_n)$, we have that $\E(\bm \phi^1 + \bm \phi^4|\mathcal{A}_n,\mathcal{Z}_n) = \bm 0$. Using the Lindeberg-Feller central limit theorem (see Lemma \ref{lemma:lfclt}),
\begin{align*}
	\sqrt{n} \E_n(\bm \phi^1 + \bm \phi^4) \big|\mathcal{A}_n,\mathcal{Z}_n \xrightarrow{d} N\left(0, \bm V^{(1,4)}\right)
\end{align*}
where $\bm V^{(1,4)}$ is a matrix such that $\Var(\bm \phi^1 + \bm \phi^4 |\mathcal{A}_n,\mathcal{Z}_n) \to \bm V^{(1,4)}$ in probability. The variance of two random vectors (since $\bm \phi^1$ and $\bm \phi^4$ are conditionally mean zero), is given by
\begin{align*}
	\Var(\bm \phi^1 + \bm \phi^4 | \mathcal{A}_n,\mathcal{Z}_n) &= \E\left\{(\bm \phi^1 + \bm \phi^4)(\bm \phi^1 + \bm \phi^4 )^T \big| \mathcal{A}_n,\mathcal{Z}_n \right\} \\
	&= \E\left\{\bm \phi^1 (\bm \phi^1)^T\big| \mathcal{A}_n,\mathcal{Z}_n \right\} + \E\left\{\bm \phi^1 (\bm \phi^4)^T\big| \mathcal{A}_n,\mathcal{Z}_n \right\} \\
	&\quad\quad + \E\left\{\bm \phi^4 (\bm \phi^1)^T\big| \mathcal{A}_n,\mathcal{Z}_n \right\} + \E\left\{\bm \phi^4 (\bm \phi^4)^T\big| \mathcal{A}_n,\mathcal{Z}_n \right\} \\
	&= \Var(\bm \phi^1|\mathcal{A}_n,\mathcal{Z}_n) + \Var(\bm \phi^4|\mathcal{A}_n,\mathcal{Z}_n) + \\
	&\quad\quad \Cov(\bm \phi_1, \bm \phi_4|\mathcal{A}_n,\mathcal{Z}_n) + \Cov(\bm \phi_4, \bm \phi_1|\mathcal{A}_n,\mathcal{Z}_n).
\end{align*}
Since conditional on $(\mathcal{A}_n,\mathcal{Z}_n)$, the product of $1_j(A_i)$ and $1_k(A_i)$ is zero for $j \neq k$, then $\Var(\bm \phi^1|\mathcal{A}_n,\mathcal{Z}_n)$ is a diagonal matrix with $(a,a)$th entry given by
\begin{align*}
	\pi_a^{-2} 1_a(A_i) \Var(\Yai - \tilde{\mu}_a(\X_i)|\Z_i) \to \pi_a^{-1} \E(\Var(\Ya - \tilde{\mu}_a(\X)|\Z)).
\end{align*}
Similarly, $\Var(\bm \phi^4 | \mathcal{A}_n,\mathcal{Z}_n) \to \E(\Var(\tilde{\mu}(\X)|\Z))$, $\Cov(\bm \phi^1, \bm \phi^4 | \mathcal{A}_n,\mathcal{Z}_n) \to \E(\Cov(Y - \tilde{\mu}(\X), \tilde{\mu}(\X) | \Z))$. Therefore,
\begin{align*}
	\bm V^{(1,4)} &= \text{diag}\{\pi_a^{-1} \E(\Var(\Ya - \tilde{\mu}_a(\X)|\Z))\} + \E(\Var(\tilde{\mu}(\X)|\Z)) \\
	&\quad\quad + \E(\Cov(Y - \tilde{\mu}(\X), \tilde{\mu}(\X) | \Z)) + \E(\Cov(\tilde{\mu}(\X), Y - \tilde{\mu}(\X) | \Z)).
\end{align*}
Furthermore, we have that, since $Z_i$ are i.i.d. random variables, we have
\begin{align*}
	\sqrt{n} \E_n(\bm \phi^3 + \bm \phi^5) \xrightarrow{d} N\left(0, \bm V^{(3,5)}\right)
\end{align*}
where, similarly to the steps for $\bm V^{(1,4)}$ (though there is no need to condition on $(\mathcal{A}_n,\mathcal{Z}_n)$),
\begin{align*}
	\bm V^{(3, 5)} &= \Var(\bm \phi^3) + \Var(\bm \phi^5) + \Cov(\bm \phi^3, \bm \phi^5) + \Cov(\bm \phi^5, \bm \phi^3) \\
	&= \E(\E(Y - \tilde{\mu}|\Z)\E(Y - \tilde{\mu}|\Z)^T) + \Var(\E(\tilde{\mu}|\Z)) \\
	&\quad\quad + \Cov(\E(Y - \tilde{\mu}|\Z), \E(\tilde{\mu}|\Z)) + \Cov(\E(\tilde{\mu}|\Z), \E(Y - \tilde{\mu}|\Z)) \\
	&= \E(R(\Z) \pi \pi^T R(\Z)) + \Var(\E(\tilde{\mu}|\Z)) \\
	&\quad\quad + \Cov(\E(Y - \tilde{\mu}|\Z), \E(\tilde{\mu}|\Z)) + \Cov(\E(\tilde{\mu}|\Z), \E(Y - \tilde{\mu}|\Z))
\end{align*}
where $R(\Z) = \text{diag}\{\pi_a^{-1} \E(\Ya - \tilde{\mu}_a(\X)|\Z); a = 1, ..., k\}$.

\newpage
\subsection{Helper Lemmas for Asymptotic Normality}

We will use the following lemma that uses asymptotic independence of the various components of the influence function that we have defined in order to claim asymptotic normality under condition (B) (or (U) later on).
\begin{lemma}[Asymptotic Independence]\label{lemma:asymptotic-independence}
    Let $Q_n$, $R_n$, and $S_n$ be random variables such that $Q_n | \mathcal{A}_n, \mathcal{Z}_n \xrightarrow{d} N(\bm 0, \bm I)$, $R_n | \mathcal{Z}_n \xrightarrow{d} N(\bm 0, \bm I)$, and $S_n \xrightarrow{d} N(\bm 0, \bm I)$. Then jointly, $(Q_n, R_n, S_n) \xrightarrow{d} (Q, R, S)$, where $Q$, $R$, and $S$ are independent, standard multivariate normal random variables. (Also see \citet{Ye2021better} Proof of Theorem 3).
\end{lemma}
\begin{proof}
    Applying iterated expectation, and the dominated convergence theorem to switch limit and expectation, we have
\begin{align*}
	\lim_{n\to\infty} P(Q_n \leq Q, R_n \leq R, S_n \leq S) &= \lim_{n\to\infty} \E\left\{\E\left\{\E\left\{1(Q_n \leq Q)1(R_n \leq R) 1(S_n \leq S) \big| \mathcal{A}_n,\mathcal{Z}_n\right\} | \mathcal{Z}_n \right\} \right\} \\
	&= \lim_{n\to\infty} \E\left\{\E\left\{\E\left\{1(Q_n \leq Q) \big| \mathcal{A}_n,\mathcal{Z}_n\right\} 1(R_n \leq R)  \big|\mathcal{Z}_n \right\}1(S_n \leq S) \right\} \\
	&= \E\left\{ \E\left\{\lim_{n\to\infty}\E\left\{1(Q_n \leq Q) \big| \mathcal{A}_n,\mathcal{Z}_n\right\} \lim_{n\to\infty}1(R_n \leq R)  \big|\mathcal{Z}_n \right\} \lim_{n\to\infty} 1(S_n \leq S) \right\} \\
	&= \Phi(Q) \E\left\{\E\left\{  \lim_{n\to\infty}1(R_n \leq R)  \big| \mathcal{Z}_n \right\} \lim_{n\to\infty} 1(S_n \leq s) \right\} \\
 &= \Phi(Q) \E\left\{\lim_{n\to\infty} \E\left\{  1(R_n \leq R)  \big| \mathcal{Z}_n \right\} \lim_{n\to\infty} 1(S_n \leq s) \right\} \\
  &= \Phi(Q) \Phi(R) \E\left\{ \lim_{n\to\infty} 1(S_n \leq s) \right\} \\
  &= \Phi(Q) \Phi(R) \Phi(S) \\
    &= P(Q \leq Q, R \leq R, S \leq S)
\end{align*}
Therefore $(Q_n, R_n, S_n) \xrightarrow{d} (Q, R, S)$, where $Q$, $R$, and $S$ are independent, standard multivariate normal random variables.
\end{proof}
\begin{lemma}[Verification of Lindeberg-Feller Condition]\label{lemma:lfclt}
    Define $\bm \phi_1$ and $\bm \phi_4$ as in Section \ref{app:asymptotic-normality}, that is $\bm \phi_1 = (\phi^1_1, ..., \phi^1_k)$, and $\bm \phi_4 = (\phi^4_1, , ..., \phi^4_k)$ with
    \begin{align*}
	\phi_a^1: (A_i,X_i,Y_i,Z_i) &\to \frac{1_a(A_i)}{\pi_a} \left[\Yai - \tilde{\mu}_a(\X_i) - \E\left\{\Yai - \tilde{\mu}_a(\X_i) \big|\Z_i \right\}\right] \\
	\phi_a^4: (X_i,Z_i) &\to \tilde{\mu}_a(\X_i) - \E\left\{\tilde{\mu}_a(\X_i) |\Z_i\right\}.
    \end{align*} Then
	$\sqrt{n} \E_n(\bm \phi^1 + \bm \phi^4) \big|\mathcal{A}_n,\mathcal{Z}_n \xrightarrow{d} N\left(0, \bm V^{(1,4)}\right)$ where $\bm V^{(1,4)}$ is a matrix such that $\Var(\bm \phi^1 + \bm \phi^4 |\mathcal{A}_n,\mathcal{Z}_n) \to \bm V^{(1,4)}$ in probability.
\end{lemma}
\begin{proof}
    We aim to verify the Lindeberg condition for any linear combination of the $2k$ vector $\E_n(\phi^1_1, ..., \phi^1_k, \phi^4_1, ..., \phi^4_k)\bm c$ with $\bm c \in \mathbb{R}^{2k}$, conditional on $(\mathcal{A}_n, \mathcal{Z}_n)$. By the Cramer-Wold device, and the Continuous Mapping Theorem, this will complete the proof. For simplicity of notation, let $\bm W_i = (A_i, X_i, y_{1,i}, ..., y_{k,i}, Z_i)$. Let
    \begin{align*}
    K_i = \frac{1}{n} \left(\phi^1_1(\bm W_i), ..., \phi^1_k(\bm W_i), \phi^4_1(\bm W_i), ..., \phi^4_k(\bm W_i)\right)\bm c.
    \end{align*}
    We are interested in $\sum_{i=1}^{n} K_i$. We know that $\E(K_i|\mathcal{A}_n, \mathcal{Z}_n) = 0$, since each component $\phi^1_a(\bm W_i)$ and $\phi^4_a(\bm W_i)$ have conditional expectation zero. Furthermore,
    \begin{align*}
        \Var(K_i|\mathcal{A}_n, \mathcal{Z}_n) &= \frac{1}{n^2} \bm c^T \Var\left(\phi^1_1(\bm W_i), ..., \phi^1_k(\bm W_i), \phi^4_1(\bm W_i), ..., \phi^4_k(\bm W_i)|\mathcal{A}_n, \mathcal{Z}_n\right) \bm c.
    \end{align*}
    Lindeberg's condition for the random variables $K_i | (\mathcal{A}_n, \mathcal{Z}_n)$ holds if, for all $\epsilon > 0$, (following exactly the steps in \citet{Ye2021better} proof of their Theorem 2)
    \begin{align}\label{lindeberg}
        \sum_{i=1}^{n} \E \left[ \frac{K_i^2}{\Var(\sum_{i=1}^{n} K_i | \mathcal{A}_n, \mathcal{Z}_n)} I \left( \frac{K_i^2}{\Var(\sum_{i=1}^{n} K_i | \mathcal{A}_n, \mathcal{Z}_n)} > \epsilon \right)\middle| \mathcal{A}_n, \mathcal{Z}_n \right] \to 0
    \end{align}
    as $n \to \infty$. Let $\tau_{i,n} := \frac{\Var(K_i | \mathcal{A}_n, \mathcal{Z}_n)}{\Var(\sum_{i=1}^{n} K_i | \mathcal{A}_n, \mathcal{Z}_n)}$. By (conditional) independence, $\sum_{i=1}^{n} \tau_{i,n} = 1$, since the sum of the variances is the variance of the sums. Then we can rewrite \eqref{lindeberg} as:
    \begin{align*}
        \eqref{lindeberg} &= \sum_{i=1}^{n} (\tau_{i,n}) \E \left[ \frac{K_i^2}{\Var(K_i | \mathcal{A}_n, \mathcal{Z}_n)} I \left( \frac{K_i^2}{\Var(\sum_{i=1}^{n} K_i | \mathcal{A}_n, \mathcal{Z}_n)} > \epsilon \right)\middle|\mathcal{A}_n, \mathcal{Z}_n \right] \\
        &\leq \max_{i} \E \left[ \frac{K_i^2}{\Var(K_i | \mathcal{A}_n, \mathcal{Z}_n)} I \left( \frac{K_i^2}{\Var(\sum_{i=1}^{n} K_i | \mathcal{A}_n, \mathcal{Z}_n)} > \epsilon \right)\middle|\mathcal{A}_n, \mathcal{Z}_n \right] \\
        &= \max_{i} \E \left[ \frac{K_i^2}{\Var(K_i | \mathcal{A}_n, \mathcal{Z}_n)} I \left( \frac{K_i^2}{\Var(K_i|\mathcal{A}_n, \mathcal{Z}_n)} > \frac{\epsilon}{ \tau_{i,n}}\right)\middle|\mathcal{A}_n, \mathcal{Z}_n \right]
    \end{align*}
     Note that $W_i := \frac{K_i}{\sqrt{\Var(K_i | \mathcal{A}_n,    \mathcal{Z}_n)}}$ is a random variable which, conditional on $\mathcal{A}_n, \mathcal{Z}_n$, has mean zero and has unit variance. It is important to recall that we have finite strata levels $\{1, ..., L\}$. Therefore, quantities conditional on strata levels and treatment indicators have maximums that are achieved. For any $i$, we have that $\E(W_i^2 I(W_i^2 > \epsilon \max_i(\tau_{i,n}^{-1}))|\mathcal{A}_n, \mathcal{Z}_n) \to 0$ by the dominated convergence theorem. This is because point-wise, the function $w^2 I(w^2 > \epsilon \max_i(\tau_{i,n}^{-1})) \to 0$ as $n \to \infty$. Also, $w^2 I(w^2 > \epsilon \max_i (\tau_{i,n}^{-1})) \leq w^2$, and $\E(w^2|\mathcal{A}_n, \mathcal{Z}_n) = 1 < \infty$ is bounded, so we can switch limit and expectation. Therefore, $\lim_{n\to\infty} \E(W_i^2 I(W_i^2 > \epsilon \max_i(\tau_{i,n}^{-1}))|\mathcal{A}_n, \mathcal{Z}_n) =  \E(\lim_{n\to\infty} W_i^2 I(W_i^2 > \epsilon \max_i(\tau_{i,n}^{-1}))|\mathcal{A}_n, \mathcal{Z}_n) = 0$.
     
     The covariance elements of $\Var\left(\phi^1_1(\bm W_i), ..., \phi^1_k(\bm W_i), \phi^4_1(\bm W_i), ..., \phi^4_k(\bm W_i)|\mathcal{A}_n, \mathcal{Z}_n\right)$ are given by:
    \begin{align*}
        \Cov(\phi^1_a(\bm W_i), \phi^1_a(\bm W_i)|\mathcal{A}_n, \mathcal{Z}_n) &= \pi_a^{-2} 1_a(A_i) \Var(Y_{a,i} - \tilde{\mu}_a(X_i) | Z_i) \\
        \Cov(\phi^1_a(\bm W_i), \phi^1_b(\bm W_i)|\mathcal{A}_n, \mathcal{Z}_n) &= 0 \\
        \Cov(\phi^4_a(\bm W_i), \phi^4_b(\bm W_i)|\mathcal{A}_n, \mathcal{Z}_n) &= \Cov(\tilde{\mu}_a(X_i), \tilde{\mu}_b(X_i)|Z_i) \\
        \Cov(\phi^1_a(\bm W_i), \phi^4_b(\bm W_i)|\mathcal{A}_n, \mathcal{Z}_n) &= \frac{1}{\pi_a} 1_a(A_i) \Cov(Y_{a,i} - \tilde{\mu}_a(X_i), \tilde{\mu}_b(X_i)|Z_i)
    \end{align*}
    The elements are all functions of $(A_i, Z_i)$ which have finite levels. Therefore, the quadratic form of the covariance matrix, $\bm c^T \Var\left(\phi^1_1(\bm W_i), ..., \phi^1_k(\bm W_i), \phi^4_1(\bm W_i), ..., \phi^4_k(\bm W_i)|\mathcal{A}_n, \mathcal{Z}_n\right) \bm c\leq \bm C_{Z_i, A_i}$, which is positive semi-definite by nature of being a covariance matrix, can be bounded above and below by positive constants. Thus, $\max_i \Var(K_i | \mathcal{A}_n, \mathcal{Z}_n) = n^{-2} \max_{z_l; l\in \{1, ..., L\}, a \in \{1, ..., k\}} (C_{z_l, a})$. Proceeding in a similar fashion, 
    \begin{align*}
        \Var\left(\sum_{i=1}^{n} K_i | \mathcal{A}_n, \mathcal{Z}_n\right) &\geq n^{-2} \sum_{i=1}^{n} \bm c^T \Var\left(\phi^1_1(\bm W_i), ..., \phi^1_k(\bm W_i), \phi^4_1(\bm W_i), ..., \phi^4_k(\bm W_i)|\mathcal{A}_n, \mathcal{Z}_n\right) \\
        &\geq n^{-1} \min_{z_l; l\in \{1, ..., L\}, a \in \{1, ..., k\}} (C_{z_l, a}).
    \end{align*}
    Putting these pieces together,
    \begin{align*}
        \tau_{i,n} &\leq \max \tau_{i,n}
        \leq n^{-1} \left(\max_{z_l; l\in \{1, ..., L\}, a \in \{1, ..., k\}} (C_{z_l, a}) / \min_{z_l; l\in \{1, ..., L\}, a \in \{1, ..., k\}} (C_{z_l, a})\right) = o(1).
    \end{align*}.
\end{proof}

\newpage
\subsection{Asymptotic Normality  and Variance under Condition (B): Proof of Theorem \ref{theo:asymptotic-normality} (ii)}

We now apply the results of the previous sections to derive asymptotic normality and variance under condition (B). Under (B), we have that
\begin{align}\label{eq:B}
	\sqrt{n} \left(      {\textstyle \frac{n_a(z_l) }{n(z_l)}- \pi_a },  \, {}^{a=1,...,k}_{ \, l=1,..., L}  \right)^T \, \bigg| \, \mathcal{Z}_n \, \xrightarrow{d}  \,  N \left(0, {\rm diag} \left\{
\frac{\Omega(z_l)}{P(\Z=z_l)}, \, l=1,...,L \right\}
\right).
\end{align}
We first note that, with $\hat{p}_{z_l} = n(z_l) / n$,
\begin{align*}
	\frac{1}{n} \sum_{i=1}^{n} 1_a(A_i) \pi_a^{-1} \E(\Yai - \tilde{\mu}_a(\X_i) | \Z_i) &= \sum_{l=1}^{L}\frac{n_a(z_l)}{n(z_l)} \hat{p}_{z_l} R_a(z_l) \\
	\frac{1}{n} \sum_{i=1}^{n} \E(\Yai - \tilde{\mu}_a(\X_i) | \Z_i) &= \sum_{l=1}^{L} \pi_a  \hat{p}_{z_l} R_a(z_l).
\end{align*}
where $R_a(Z) = \pi_a^{-1} \E(\Yai - \tilde{\mu}_a(\X_i)|\Z)$.
Therefore, we can find the distribution of $\bm \phi_2 | \mathcal{Z}_n$ by performing a linear transformation on the normal distribution in \eqref{eq:B}. Define $R(\Z) = \text{diag}\{R_1(\Z), ..., R_k(\Z)\}$. Then, letting $\mathbb{L}_n := \left(      {\textstyle \frac{n_a(z_l) }{n(z_l)}- \pi_a },  \, {}^{a=1,...,k}_{ \, l=1,..., L}  \right)^T$, we have that
\begin{align*}
	\E_n(\bm \phi_2) = \left[
		\hat{p}_{z_1} R(z_1), ..., \hat{p}_{z_L} R(z_L)
	\right] \mathbb{L}_n.
\end{align*}
\begin{remark}\label{remark:phi2-op1}
    We need (B) to show asymptotic normality and derive its limiting variance, but we do not need (B) to show that $\E_n(\bm \phi_2)|\mathcal{Z}_n = O_P(1/\sqrt{n})$. This is apparent only from Assumption 1 (iii), which states that each element of the vector $\sqrt{n} \mathcal{L}_n = O_P(1)$, so the linear transformation of $\mathbb{L}_n$ above is $O_P(1/\sqrt{n})$.
\end{remark}

Therefore, $\sqrt{n} \E_n(\bm \phi_2) | \mathcal{Z}_n \xrightarrow{d} N(0, \bm V^{(2)})$ where $\bm V^{(2)}$ is the limit in probability of the variance-covariance matrix of the linear transformation conditioning on $\mathcal{Z}_n$:
\begin{align*}
	\bm V^{(2)}_{\mathcal{Z}_n} &= \left[\hat{p}_{z_1} R(z_1), ..., \hat{p}_{z_L} R(z_L)	\right]
	\text{diag}\left\{p_{z_l}^{-1} \Omega(z_l); l = 1, ... L \right\} \left[\hat{p}_{z_1} R(z_1), ..., \hat{p}_{z_L} R(z_L)	\right]^T \\
	&= \sum_{l=1}^{L} (\hat{p}_{z_l})^2 p_{z_l}^{-1} R(z_l) \Omega(z_l) R(z_l) \\
	&= \sum_{l=1}^{L} \hat{p}_{z_l} R(z_l) \Omega(z_l) R(z_l) + o_P(1) \xrightarrow{p} \E(R(\Z) \Omega(\Z) R(\Z)) \equiv \bm V^{(2)}.
\end{align*}

Define the following scaled random variables that are the empirical averages of the influence function components: $\bm M_n^{(1, 4)} := \left[\bm V^{(1,4)}\right]^{-1}\sqrt{n} \E_n(\bm \phi^1 + \bm \phi^4)$, $\bm M_n^{(2)} := \left[\bm V^{(2)}\right]^{-1}\sqrt{n} \E_n(\bm \phi^2)$, and $\bm M_n^{(3,5)} := \left[\bm V^{(3,5)}\right]^{-1}\sqrt{n} \E_n(\bm \phi^3 + \bm \phi^5)$. Then by the continuous mapping theorem, $\bm M_n^{(1, 4)} | \mathcal{A}_n,\mathcal{Z}_n \xrightarrow{d}  N(0, \bm I_k)$, $\bm M_n^{(2)} | \mathcal{Z}_n \xrightarrow{d}  N(0, \bm I_k)$, and $\bm M_n^{(3,5)} \xrightarrow{d}  N(0, \bm I_k)$. Furthermore, these three normal distributions are asymptotically independent. Applying Lemma \ref{lemma:asymptotic-independence},
$(\bm M_n^{(1, 4)}, \bm M_n^{(2)}, \bm M_n^{(3,5)}) \xrightarrow{d} (\bm M^{(1, 4)}, \bm M^{(2)}, \bm M^{(5)})$, where $\bm M^{(1,4)}$, $\bm M^{(2)}$, and $\bm M^{(3,5)}$ are independent multivariate normal random variables.
Therefore, by the continuous mapping theorem, we can scale each of them by their respective variance-covariance matrices, and add the three random variables together to get that under (B),
\begin{align*}
	\sqrt{n} (\hat{\theta}_{AIPW} - \theta) \equiv \sqrt{n} \E_n(\bm \phi^1 + \bm \phi^2 + \bm \phi^3 + \bm \phi^4 + \bm \phi^5) + o_P(1)\xrightarrow{d} N(0, \bm \Sigma_{(B)})
\end{align*}
where $\bm \Sigma_{(B)}$ is given by
\begin{align*}
\bm \Sigma_{(B)} &= \bm V^{(1,4)} + \bm V^{(2)} + \bm V^{(3,5)} \\
&= \text{diag}\{\pi_a^{-1} \E(\Var(\Ya - \tilde{\mu}_a(\X)|\Z))\} + \E(\Var(\tilde{\mu}(\X)|\Z)) \\
	&\quad\quad + \E(\Cov(Y - \tilde{\mu}(\X), \tilde{\mu}(\X) | \Z)) + \E(\Cov(\tilde{\mu}(\X), Y - \tilde{\mu}(\X) | \Z)) \\
	&\quad\quad + \E(R(\Z) \Omega(\Z) R(\Z)) + \E(R(\Z) \pi \pi^T R(\Z)) + \Var(\E(\tilde{\mu}|\Z)) \\
	&\quad\quad + \Cov(\E(Y - \tilde{\mu}(\X)|\Z), \E(\tilde{\mu}(\X)|\Z)) + \Cov(\E(\tilde{\mu}(\X)|\Z), \E(Y - \tilde{\mu}(\X)|\Z)) \\
	&= \text{diag}\{\pi_a^{-1} \E(\Var(\Ya - \tilde{\mu}_a(\X)|\Z))\} + \E(R(\Z) \Omega(\Z) R(\Z)) + \E(R(\Z) \pi \pi^T R(\Z))  \\
	&\quad\quad + \Var(\tilde{\mu}(\X)) + \Cov(Y - \tilde{\mu}(\X), \tilde{\mu}(\X)) + \Cov(\tilde{\mu}(\X), Y - \tilde{\mu}(\X)) \\
	&= \underbrace{\text{diag}\{\pi_a^{-1} \E(\Var(\Ya - \tilde{\mu}_a(\X)|\Z))\}}_{(I)} + \E(R(\Z) \Omega(\Z) R(\Z)) + \E(R(\Z) \pi \pi^T R(\Z))  \\
	&\quad\quad - \Var(\tilde{\mu}(\X)) + \Cov(Y, \tilde{\mu}(\X)) + \Cov(\tilde{\mu}(\X), Y)
\end{align*}
where we have simplified the expressions using the law of total variance and law of total covariance.
Under simple randomization $\Omega_{SR} = \text{diag}\{\pi_a; a = 1, ..., k\} - \pi \pi^T$. Then,
\begin{align*}
	(I) &= \text{diag}\{\pi_a^{-1} \Var(\Ya - \tilde{\mu}_a(\X))\} - \text{diag}\{\pi_a^{-1} \Var(\E(\Ya - \tilde{\mu}_a(\X)|\Z))\} \\
	&=  \text{diag}\{\pi_a^{-1} \Var(\Ya - \tilde{\mu}_a(\X))\} - \text{diag}\{\pi_a^{-1} \Var(R_a(\Z))\} \\
	&= \text{diag}\{\pi_a^{-1} \Var(\Ya - \tilde{\mu}_a(\X))\} - \E(R(\Z) \text{diag}\{\pi_a\} R(\Z))
\end{align*}
where the last equality is because $E(R(\Z)) = \bm 0$. Putting these together, we have
\begin{align*}
	\bm \Sigma_{(B)} &= \text{diag}\{\pi_a^{-1} \Var(\Ya - \tilde{\mu}_a(\X))\} \\
	&\quad\quad - \E(R(\Z) \text{diag}\{\pi_a\} R(\Z)) + \E(R(\Z) \Omega(\Z) R(\Z)) + \E(R(\Z) \pi \pi^T R(\Z))  \\
	&\quad\quad - \Var(\tilde{\mu}(\X)) + \Cov(Y, \tilde{\mu}(\X)) + \Cov(\tilde{\mu}(\X), Y) \\
	&= \text{diag}\{\pi_a^{-1} \Var(\Ya - \tilde{\mu}_a(\X))\} - \E(R(\Z) \{\Omega_{SR} - \Omega\} R(\Z)) \\
	&\quad\quad - \Var(\tilde{\mu}(\X)) + \Cov(Y, \tilde{\mu}(\X)) + \Cov(\tilde{\mu}(\X), Y).
\end{align*}
Note that $\tilde{\mu}_a$ can be replaced with $\mu_a$ in the above expression because $\Var(P_0 \mu_a(X) - \theta_a)$ is a constant.

\newpage
\section{Asymptotic Linearity and Normality of Cross-Fitted Estimator under CAR}

In this section, we show that all of our theoretical results also apply to AIPW estimators that use cross-fitting. Assumption 2* states that given a random subset $I_k$, for every $a$,  the nuisance parameter estimator $\hat{\mu}_{a,k}$ using the $I_k^c$ sample satisfies the following condition:  there exists a function $\mu_a$ with finite second order moment such that with probability $1-\Delta_n$, 
 $\|\hat \mu_{a,k}  - \mu_a \|_{L^2 (P_0)} \leq \delta_n $, where $\Delta_n=o(1)$ and $\delta_n = o(1)$. We implicitly assume that there exists an $\epsilon > 0$ such that $\epsilon < \pi_a < 1-\epsilon$, and $\epsilon < \hat \pi_{a,k} < 1-\epsilon$ almost surely. This is a very mild assumption: true and estimated treatment probabilities need to be bounded away from 0 and 1).

\subsection{Helper Lemmas for Cross-Fitting}

\begin{lemma}\label{lemma:mu-star}
	Let $\mu^*_a(\X) = \mu_a(\X) - P_0 \mu_a(\X) + \theta_a$. Let $$\mathcal{F}^*_{a,n} = \{x \to \mu(\X) - P_0 \mu + \theta_a: \mu \in \mathcal{F}_{a,n} \}.$$ Then $\sup_{\mu \in \mathcal{F}_{a,n}}\E([\mu(\X) - \mu_a(\X)]^2) = o_P(1)$ implies $\sup_{\mu^* \in \mathcal{F}^*_{a,n}} \E([{\mu}^*(\X) - \mu_a^*(\X)]^2) = o_P(1)$.\end{lemma}
\begin{proof}
	Let $P_{n,a}(\cdot) = \frac{1}{n_a} \sum_{i:A_i=a}(\cdot)$. We can write, for any $\mu^* \in \mathcal{F}^*_{a,n}$,
	\begin{align*}
		\E([{\mu}^*(\X) - \mu_a^*(\X)]^2) &= \E([\mu(\X) - \mu_a(\X) + P_0 (\mu_a - \mu)]^2) \\
		&\leq C \cdot \left\{ \E([\mu(\X) - \mu_a(\X)]^2) + \E([P_{0} \mu - P_0 \mu_a]^2)\right\} \\
		&= C \cdot \left\{ \E([\mu(\X) - \mu_a(\X)]^2) + [P_{0}(\mu - \mu_a)]^2\right\} \\
		&\leq 2C \cdot \left\{ \E([\mu(\X) - \mu_a(\X)]^2)\right\}
	\end{align*}
	for some $\mu \in \mathcal{F}_{a,n}$ and some constant $C$, where the last step is due to Jensen's Inequality. Taking a supremum over $\mu^* \in \mathcal{F}_{a,n}^*$ on the LHS is equivalent to taking a supremum over $\mu \in \mathcal{F}_{a,n}$ on the RHS. Thus, to show that $\sup_{\mu^* \in \mathcal{F}^*_{a,n}} \E([{\mu}^*(\X) - \mu_a^*(\X)]^2) = o_P(1)$, it suffices to show that $\sup_{\mu \in \mathcal{F}_{a,n}}$ over the RHS is $o_P(1)$, which is true by assumption.\end{proof}

\begin{lemma}\label{lemma:strata-contradiction}
	The condition that $\sup_{\mu^* \in \mathcal{F}^*_{a,n}}\E([\mu^*(\X) - \mu^*_a(\X)]^2) = o_P(1)$ implies that for each strata $Z$, $\sup_{\mu^* \in \mathcal{F}^*_{a,n}} \E([{\mu}^*(\X) - \mu_a^*(\X)]^2|\Z) = o_P(1)$.
\end{lemma}
\begin{proof}
	We use proof by contradiction. Assume that $\sup_{\mu^* \in \mathcal{F}^*_{a,n}}\E([\mu^*(\X) - \mu^*_a(\X)]^2) = o_P(1)$, and that there exists a $\delta > 0$ and a strata $Z'$ such that for all $\epsilon > 0$, $$\lim_{n\to\infty} P\left(\sup_{\mu^* \in \mathcal{F}^*_{a,n}} \E([{\mu}^*(\X) - \mu_a^*(\X)]^2|\Z') > \epsilon \right) > \delta.$$ Furthermore, by iterated expectation, $\E([{\mu}^*(\X) - \mu_a^*(\X)]^2) = \sum_{l=1}^{L} \E([{\mu}^*(\X) - \mu_a^*(\X)]^2|\Z=l) > \E([{\mu}^*(\X) - \mu_a^*(\X)]^2|\Z') P(Z')$. Since each strata has positive probability, there exists a $\rho > 0$ such that $\E([{\mu}^*(\X) - \mu_a^*(\X)]^2) > \E([{\mu}^*(\X) - \mu_a^*(\X)]^2|\Z') \rho$. Therefore, for all $\epsilon' > 0$ (where we can re-parameterize $\epsilon' = \epsilon \cdot \rho$),
	\begin{align*}
		P\left(\sup_{\mu^* \in \mathcal{F}^*_{a,n}} \E([{\mu}^*(\X) - \mu_a^*(\X)]^2) > \epsilon' \right) \geq P\left(\sup_{\mu^* \in \mathcal{F}^*_{a,n}} \E([{\mu}^*(\X) - \mu_a^*(\X)]^2|\Z') > \epsilon \right).
	\end{align*}
	Taking the limit of both sides as $n \to \infty$, we have that the RHS is $> \delta$ by assumption, for all $\epsilon > 0$. Therefore, the limit of the LHS is also $> \delta$, meaning that $\sup_{\mu^* \in \mathcal{F}^*_{a,n}} \E([{\mu}^*(\X) - \mu_a^*(\X)]^2)$ is not $o_P(1)$. This implies that if  $\sup_{\mu^* \in \mathcal{F}^*_{a,n}}\E([\mu^*(\X) - \mu^*_a(\X)]^2) = o_P(1)$, then each strata $Z$ must also have $\sup_{\mu^* \in \mathcal{F}^*_{a,n}} \E([{\mu}^*(\X) - \mu_a^*(\X)]^2|\Z) = o_P(1)$.
\end{proof}

\begin{lemma}\label{lemma:implies-sup}
	Let $\hat{\mu}_{a,k}$ be a function such that, with probability $1 - \Delta_n$, $\norm{\hat{\mu}_{a,k} - \mu_a}_{L^2(P_0)} \leq \delta_n$, where $\delta_n = o(1)$ and $\Delta_n = o(1)$. Then with probability $1 - \Delta_n$, $\hat{\mu}_{a,k} \in \mathcal{F}_{a,n}$, and 
	\begin{align*}
		\sup_{\mu\in \mathcal{F}_{a,n}} \norm{\hat{\mu}_{a,k} - \mu_a}_{L^2(P_0)}^2 \equiv \sup_{\mu\in \mathcal{F}_{a,n}} \E([{\hat{\mu}_{a,k}(\X) - \mu_a(\X)}_{L^2(P_0)}]^2) = o(1)
	\end{align*}
	where $\mathcal{F}_{a,n} = \{\mu: \norm{\mu - \mu_a}_{L^2(P_0)} \leq \delta_n\}$.
\end{lemma}
\begin{proof}
	First, by definition, note that since with probability $1 - \Delta_n$, $\norm{\hat{\mu}_{a,k} - \mu_a}_{L^2(P_0)} \leq \delta_n$, with probability $1 - \Delta_n$, $\hat{\mu}_{a,k} \in \mathcal{F}_{a,n}$ since it satisfies the condition of $\mathcal{F}_{a,n}$. Furthermore, for any $\mu$ in $\mathcal{F}_{a,n}$, we know that $\norm{\mu - \mu_a}_{L^2(P_0)} \leq \delta_n$. Therefore, the supremum over $\mu \in \mathcal{F}_{a,n}$ also satisfies $\sup_{\mu\in \mathcal{F}_{a,n}} \norm{\mu - \mu_a}_{L^2(P_0)} \leq \delta_n$ (because it is a compact set). Finally, since $\sup_{\mu \in \mathcal{F}_{a,n}} \norm{\hat{\mu}_{a,k} - \mu_a}_{L^2(P_0)} \leq \delta_n$, $\sup_{\mu \in \mathcal{F}_{a,n}} \norm{\mu_{a,k} - \mu_a}_{L^2(P_0)}^2 \leq \delta^2_n$. Since $\delta_n = o(1)$, $\delta_2^n = o(1)$ as well. This completes the proof.
\end{proof}

\begin{lemma}\label{cf-remainder}
Defining $\psi(\bm W_i; \theta_a, \eta_a) = \frac{1_a(A_i)}{{\pi}_a}\bigg\{\Yai - \mu^*_a(\X_i) \bigg\} + \mu^*_a(\X_i) - \theta_a\label{linear-score}$, we have the following:
\begin{align*}
	\mathcal{I}_{k} := \frac{1}{\sqrt n_k} \sum_{i\in I_k} \left\{\psi(\bm W_i; \theta_a, \hat{\eta}_{a,k}) - \psi(\bm W_i; \theta_a, \eta_a) \right\} = o_P(1).
\end{align*}
\end{lemma}
\begin{proof}
	Here we show asymptotic negligibility for the difference in Averages of Estimating Function Evaluated at $\hat \eta_{a,k}$ versus $\eta_a$. This roughly follows the proof in \citet{chernozhukov2017double} for Step 3 of Theorem 3.1, but we define a different empirical process notation that allows us to have conditional independence in the setting of covariate adaptive randomization. Finally, in what follows, we are considering the following probabilities \textit{on the event that} 
	 $\hat{\mu}_{a,k} \in \mathcal{F}_{a,n}$, where $\mathcal{F}_{a,n}$ is defined in Lemma \ref{lemma:implies-sup}. Recall that by Assumption \ref{assump:cross-fit} combined with Lemma \ref{lemma:implies-sup}, this occurs with probability $1-\Delta_n$, with $\Delta_n \to 0$. Therefore, asymptotically, the following results hold.
	
	We define the following empirical process term:
\begin{align}
	\mathbb{G}_{n,k}^*[\psi(W; \theta, \eta)] = \frac{1}{\sqrt n_k} \sum_{i\in I_k} \left(\psi(\bm W_i; \theta, \eta) - \E \big\{ \psi(\bm W_i; \theta, \eta)\big|(\bm W_i)_{i\in I_k^c}, \mathcal{A}_n, \mathcal{Z}_n \big\} \right)\label{new-emp}
\end{align}
The second term is similar to how $P_{0}^n$ is defined in the proof of Theorem \ref{theo:asymptotic-normality} (i), but where we restrict the sum to the $k$th partition, and condition on the auxiliary sample $(\bm W_i)_{i\in I_k^c}$, as well as all of the treatment indicators and strata $(\mathcal{A}_n, \mathcal{Z}_n)$. We aim to show that, for any $k$,
\begin{align*}
	\mathcal{I}_{k} := \frac{1}{\sqrt n_k} \sum_{i\in I_k} \left\{\psi(\bm W_i; \theta_a, \hat{\eta}_{a,k}) - \psi(\bm W_i; \theta_a, \eta_a) \right\} = o_P(1).
\end{align*}
Term $(I)$ can be broken up into two parts, using the notation in \eqref{new-emp}. Define:
\begin{align*}
	\mathcal{I}_{3,k} &:= \mathbb{G}^*_{n,k}[\psi(W; \theta_a, \hat\eta_{a,k})] - \mathbb{G}^*_{n,k}[\psi(W; \theta_a, \eta_a)] \\
	\mathcal{I}_{4,k} &:= \frac{1}{\sqrt n_k}\sum_{i\in I_k} \left(\E \big\{ \psi(\bm W_i; \theta_a, \hat\eta_{a,k}) - \psi(\bm W_i; \theta_a, \eta_a)\big|(\bm W_i)_{i\in I_k^c}, \mathcal{A}_n, \mathcal{Z}_n \big\} \right) \\
	&= \frac{1}{\sqrt n_k}\sum_{i\in I_k} \left(\E \big\{ \psi(\bm W_i; \theta_a, \hat\eta_{a,k})\big|(\bm W_i)_{i\in I_k^c}, \mathcal{A}_n, \mathcal{Z}_n \big\} - \E \big\{ \psi(\bm W_i; \theta_a, \eta_a)\big|\mathcal{A}_n, \mathcal{Z}_n \big\} \right).
\end{align*}
Then it suffices to show that $\mathcal{I}_{3,k} = o_P(1)$ and $\mathcal{I}_{4,k} = o_P(1)$ because 
	$|\mathcal{I}_{k}| \leq |\mathcal{I}_{3,k}| + |\mathcal{I}_{4,k}|$.
Starting with $\mathcal{I}_{3,k}$, define $\phi(w) := \psi(w; \theta_a, \hat\eta_{a,k}) - \psi(w; \theta_a,\eta_a)$. Then, defining the short-hand notation $\E^*(\cdot) = \E(\cdot | (\bm W_i)_{i\in I_k}, \mathcal{A}_n, \mathcal{Z}_n )$, we have
\begin{align*}
	\E^*\{\mathcal{I}^2_{3,k}\} &= \E^*\{|\mathbb G_{nk}^* [\phi(\W)]|^2\} \\
	&= \frac{1}{n_k}\E^* \left\{\bigg(\sum_{i\in I_k} \left(\phi(\bm W_i) - \E^*(\phi(\bm W_i))\right)\bigg)^2 \right\}.
\end{align*}
Notice that if we expand the square, the cross-terms cancel out through conditioning on $(\mathcal{A}_n,\mathcal{Z}_n)$ in $\E^*$ so that elements $(\bm W_i, W_j)$ are independent. The cross-terms are
$\left[\E^*(\phi(\bm W_i) - \E^*(\phi(\bm W_i))\right] \left[\E^*(\phi(W_j) - \E^*(\phi(W_j))\right]$ which are zero. Therefore, we can continue the above as:
\begin{align*}
	\E^*\{\mathcal{I}^2_{3,k}\} &= \frac{1}{n_k} \E^* \left\{\sum_{i\in I_k} \left(\phi(\bm W_i) - \E^*(\phi(\bm W_i))\right)^2 \right\} \\
	&\leq \frac{1}{n_k} \sum_{i\in I_k} \E^*(\phi(\bm W_i)^2) \\
	&\leq \frac{1}{n_k} \sum_{i\in I_k} \sup_{\eta \in \mathcal{T}_N} \E(\phi(\bm W_i)^2 | \mathcal{A}_n,\mathcal{Z}_n).
\end{align*}
We would like to show that the summation above is $o_P(1)$. We do this by breaking up $\phi(\bm W_i)^2$. Note first that
\begin{align*}
		[\underbrace{\psi(\W; \theta_a, \eta) - \psi(\W; \theta_a, \eta_a)}_{\upsilon(\eta, \eta_a)}]^2 &\leq \underbrace{[1_a(A)\Ya[\pi^{-1} - \pi_a^{-1}]]^2}_{\upsilon_1(\eta, \eta_a)} + \underbrace{[\mu^*(\X) - \mu^*_a(\X)]^2}_{\upsilon_2(\eta, \eta_a)} \\
		&\quad + \underbrace{[1_a(A)[\mu^*(\X)\pi^{-1} - \mu^*_a(\X)\pi^{-1}_a]]^2}_{\upsilon_3(\eta, \eta_a)}.
	\end{align*}
We will bound each of the $\upsilon$ terms separately. Using Assumption \ref{assump:cross-fit} (b), we know that $|\mu^*(\X) \pi^{-1} - \mu_a^*(\X) \pi_a^{-1}| < \epsilon^{-2}|\mu^*(\X)\pi_a - \mu_a^*(\X)\pi| = \epsilon^{-2}|\pi_a (\mu^*(\X) - \mu_a^*(\X)) + \mu_a^*(\pi - \pi_a)|$. By the generalized mean inequality, $\exists C > 0$ such that
\begin{align*}
	(\epsilon^{-2}|\pi_a (\mu^*(\X) - \mu_a^*(\X)) + \mu_a(\pi - \pi_a)|)^2 &\leq \epsilon^{-4} C \left(\pi_a^2 (\mu^*(\X) - \mu_a^*(\X))^2 + \mu_a^2(\X) (\pi - \pi_a)^2\right) \\
	&\leq \epsilon^{-4} C \left(\mu^*(\X) - \mu_a^*(\X))^2 + \mu_a^2(\X) (\pi - \pi_a)^2\right).
\end{align*}
Turning to each of the $\upsilon$ terms, and recalling that we use $\hat{\pi}_a$ (in the whole sample) as our estimator for $\pi_a$ (and that it can be brought out of the conditional expectation because it is deterministic conditional on the whole set of indicators $(\mathcal{A}_n)$), we have:
\begin{align*}
	(1) \quad \frac{1}{n_k} \sum_{i\in I_k} \sup_{\eta \in \mathcal{T}_n} \E(\upsilon_1|\mathcal{A}_n,\mathcal{Z}_n) &\leq [\hat{\pi}^{-1}_{a,k} - \pi_a^{-1}]^2 \left[\frac{1}{n_k} \sum_{i\in I_k} \E([\Yai]^2|\Z_i)\right] \\
	(2) \quad \frac{1}{n_k} \sum_{i\in I_k} \sup_{\eta \in \mathcal{T}_n}  \E(\upsilon_2|\mathcal{A}_n,\mathcal{Z}_n) &\leq \frac{1}{n_k} \sum_{i\in I_k} \sup_{\mu^* \in \mathcal{F}^*_{a,n}} \E([\mu^*(\X_i) - \mu_a^*(\X_i)]^2 | \Z_i) \\
	(3) \quad \frac{1}{n_k} \sum_{i\in I_k} \sup_{\eta \in \mathcal{T}_n} \E(\upsilon_3|\mathcal{A}_n,\mathcal{Z}_n) &\leq \frac{1}{n_k} \sum_{i\in I_k} \sup_{\mu^* \in \mathcal{F}^*_{a,n}} \epsilon^{-4} C \E[\mu^*(\X_i) - \mu_a^*(\X_i))^2|\Z_i] \\
	&\quad\quad + (\hat\pi_{a,k} - \pi_a)^2\frac{1}{n_k} \sum_{i\in I_k} \epsilon^{-4} C\E[\mu_a^2(X_i)|\Z_i]
\end{align*}
For (1) to be $o_P(1)$, it suffices that $\hat{\pi}_{a,k}^{-1} \to \pi_a$ in probability (Assumption \ref{assump: car}), and for all strata $Z$, $\E([\Ya]^2|\Z) < \infty$. This must hold since we assume that $Y$ has a finite second-order moment, unconditionally. For (2) to be $o_P(1)$, it suffices that for each strata $Z$, $\sup_{\mu^* \in \mathcal{F}^*_{a,n}} \E([\mu^*(\X_i) - \mu_a^*(\X_i)]^2 | \Z_i) = o_P(1)$. This holds by Assumption 4 (a) combined with Lemmas \ref{lemma:mu-star} and \ref{lemma:strata-contradiction}. For (3) to be $o_P(1)$ the preceding two conditions must hold, with the additional condition that $\E(\mu_a^*(\X)^2 | \Z) < \infty$ for each strata $Z$ (Assumption \ref{assump:cross-fit} specifies this unconditionally for $\mu_a$, which implies that $\E(\mu_a^*(\X)^2 | \Z) < \infty$ as well; it must also hold conditionally by strata, since each strata has positive probability). Therefore, $\E^*\{\mathcal{I}^2_{3,k}\} = o_P(1)$. By Lemma 6.1 in \citet{chernozhukov2017double}, this implies that $\mathcal{I}_{3,k} = o_P(1)$.

For $\mathcal{I}_{4,k}$, we define the function for $r \in [0, 1]$: $$f_{i,k}(r) = \left(\E \big\{ \psi(\bm W_i; \theta_a, \eta_a + r(\hat\eta_{a,k} - \eta_a))\big|(\bm W_i)_{i\in I_k^c}, \mathcal{A}_n, \mathcal{Z}_n \big\} - \E \big\{ \psi(\bm W_i; \theta_a, \eta_a)\big|\mathcal{A}_n, \mathcal{Z}_n \big\} \right)$$ and define $\mathcal I_{4,k}(r) = \frac{1}{\sqrt n_k} \sum_{i\in I_k} f_{k,i}(r)$. Then $\mathcal{I}_{4,k} = \mathcal{I}_{4,k}(1)$. We take a Taylor expansion of $\mathcal I_{4,k}(r)$ around $r = 0$ by first expanding $f_{k,i}(r)$ as follows, where we note that we can remove that conditioning on $(\bm W_i)_{i\in I_k^c}$ whenever we do not have a $\hat\eta_{a,k}$ since the observations $(\bm W_i)_{i\in I_k}$ are independent through conditioning on $(\mathcal{A}_n,\mathcal{Z}_n)$:
\begin{align*}
	f_{k,i}(0) &= \E \big\{ \psi(\bm W_i; \theta_a, \eta_a)\big|(\bm W_i)_{i\in I_k^c}, \mathcal{A}_n, \mathcal{Z}_n \big\} - \E \big\{ \psi(\bm W_i; \theta_a, \eta_a)\big|\mathcal{A}_n, \mathcal{Z}_n \big\} = 0 \\
	f_{k,i}'(0) &= \frac{d}{dr} \E \big\{ \psi(\bm W_i; \theta_a, \eta_a + r(\hat\eta_{a,k} - \eta_a))\big|(\bm W_i)_{i\in I_k^c}, \mathcal{A}_n, \mathcal{Z}_n \big\}\Big|_{r = 0} \\
	f_{k,i}''(\tilde{r}) &= \frac{d^2}{dr^2} \E \big\{ \psi(\bm W_i; \theta_a, \eta_a + r(\hat\eta_{a,k} - \eta_a))\big|(\bm W_i)_{i\in I_k^c}, \mathcal{A}_n, \mathcal{Z}_n \big\}\Big|_{r = \tilde{r}}.
\end{align*}
We have that
\begin{align*}
	\psi(\bm W_i; \theta_a, \eta_a + r (\hat\eta_{a,k} - \eta_a)) &= \frac{1_a(A_i)}{\pi_a + r(\hat\pi_{a,k} - \pi_a)} \{\Yai - \mu^*_a(\X_i) - r(\hat\mu_{a,k}^*(\X_i) - \mu_a^*(\X_i))\} \\
	&\quad\quad + \mu^*_a(\X_i) + r(\hat\mu_{a,k}^*(\X_i) - \mu_a^*(\X_i)) - \theta_a.
\end{align*}
Taking the expectation, then two derivatives with respect to $r$, we have
\begin{align*}
	\frac{d}{dr} \E^* \big\{ \psi(\bm W_i; \theta_a, \eta_a + r(\hat\eta_{a,k} - \eta_a))\big\} &= -\frac{1_a(A_i) \E^*(\hat\mu_{a,k}^*(\X_i) - \mu_a^*(\X_i))}{\pi_a + r(\hat{\pi}_{a,k} - \pi_a)} \\
	&\quad\quad - \frac{1_a(A_i)\E^*[\Yai - \mu^*_a(\X_i) - r(\hat\mu_{a,k}^*(\X_i) - \mu_a^*(\X_i))][\hat{\pi}_{a,k} - \pi_a]}{[\pi_a + r(\hat\pi_{a,k} - \pi_a)]^2} \\
	&\quad\quad + \E^*(\hat\mu_{a,k}^*(\X_i) - \mu_a^*(\X_i)) \\
	\frac{d^2}{dr^2} \E^* \big\{ \psi(\bm W_i; \theta_a, \eta_a + r(\hat\eta_{a,k} - \eta_a))\big\} &= 2\frac{(\hat{\pi}_{a,k} - \pi_a)1_a(A_i)\E^*(\hat\mu_{a,k}^*(\X_i) - \mu_a^*(\X_i))}{[\pi_a + r(\hat{\pi} - \pi_a)]^2} \\
	&\quad\quad + \frac{1_a(A_i) \E^*[\Yai - \mu^*_a(\X_i) - r(\hat\mu^*_{a,k}(X_i) - \mu^*_a(\X_i))][\hat\pi_{a,k} - \pi_a]^2}{[\pi_a + r(\hat\pi_{a,k} - \pi_a)]^3}.
\end{align*}
Evaluating $f'$ at $r = 0$, we have
\begin{align*}
	f_{k,i}'(0) &= \E^*(\hat\mu_{a,k}^*(\X_i) - \mu_a^*(\X_i)) \left(1 - \frac{1_a(A_i)}{\pi_a} \right) - \frac{1_a(A_i)}{\pi_a^2} [\hat{\pi}_{a,k} - \pi_a] \E(\Yai - \mu_a^*(\X_i)|\Z_i).
\end{align*}
We know that $\frac{1}{\sqrt n_k} \sum_{i\in I_k} f_{k,i}(r) = o_P(1)$, using both Assumption \ref{assump: car} and \ref{assump:cross-fit}. Then,
\begin{align*}
	\frac{1}{\sqrt n_k} \sum_{i\in I_k} f'_{k,i}(0) &= \pi_a \sqrt{n_k} \sum_{l=1}^{L} \frac{1}{n_{k}(l)} \hat{p}_{k,l} \left(\sup_{\mu^* \in \mathcal{F}^*_{a,n}} \E(\mu^*(\X_i) - \mu_a^*(\X_i)|\Z_i=l)\right) \sum_{i\in I_{k,l}} (\pi_a - 1_a(A_i)) \\
	&\quad\quad + \frac{1}{\pi_a^2} [\hat{\pi}_{a,k} - \pi_a]\frac{1}{\sqrt n_k}\sum_{i\in I_k} 1_a(A_i) \E(\Yai - \mu^*_a(\X_i)|\Z_i) \\
	&= \pi_a \sum_{l=1}^{L} \sqrt{\hat{p}_{k,l}} \left(\sup_{\mu^* \in \mathcal{F}^*_{a,n}} \E(\mu^*(\X_i) - \mu_a^*(\X_i)|\Z_i=l)\right) \frac{1}{\sqrt{n_{k,l}}} \sum_{i\in I_{k,l}} (\pi_a - 1_a(A_i)) \\
	&\quad\quad + o_P(1)
\end{align*}
where $\hat{p}_{k,l} = n_k(l) / n_k$, i.e., the fraction within strata $l$ out of the sample $k$, and $\mathcal I_{k,l}$ are the indexes within $\mathcal{I}_k$ such that they are also in strata $l$.
Within each strata $l$, by Assumption \ref{assump: car} (see equation \eqref{strata-ai}), we have that $\frac{1}{\sqrt{n_{k,l}}} \sum_{i\in I_{k,l}} (\pi_a - 1_a(A_i)) = O_P(1)$. Furthermore, by Assumption 2*, $\left(\sup_{\mu^* \in \mathcal{F}^*_{a,n}} \E(\mu^*(\X_i) - \mu_a^*(\X_i)|\Z_i=l)\right) = o_P(1)$. Therefore, since we have a finite number of strata, $\frac{1}{\sqrt n_k} \sum_{i\in I_k} f_{k,i}(0) = o_P(1) O_P(1) = o_P(1)$.
For the second derivative, we have (again applying the prediction unbiasedness, and recalling that $\tilde{r} \leq 1$):
\begin{align*}
	\left|\frac{1}{\sqrt n_k} \sum_{i\in I_k} f''_{k,i}(\tilde{r})\right| &\leq 2\left |\frac{(\hat{\pi}_{a,k} - \pi_a) \frac{1}{\sqrt n_k} \sum_{i\in I_k} \sup_{\mu^* \in \mathcal{F}^*_{a,n}} \E(\mu^*(\X_i) - \mu_a^*(\X_i)|\Z_i)}{[\pi_a + \tilde r(\hat{\pi} - \pi_a)]^2}\right| \\
	&\quad\quad + \left|[\hat{\pi}_{a,k} - \pi_a]^2\frac{\frac{1}{\sqrt n_k} \sum_{i\in I_k} \sup_{\mu^* \in \mathcal{F}^*_{a,n}} \E(\mu^*(\X_i) - \mu^*_a(\X_i)|\Z_i)}{[\pi_a + \tilde r(\hat\pi_{a,k} - \pi_a)]^3}\right| \\
	&\quad\quad + \left|[\hat\pi_{a,k} - \pi_a]^2\frac{\frac{1}{\sqrt n_k} \sum_{i\in I_k}1_a(A_i)\E(\Yai - \mu^*_a(\X_i)|\Z_i)}{[\pi_a + \tilde r(\hat{\pi}_{a,k} - \pi_a)]^3}\right|.
\end{align*}
To show that $\frac{1}{\sqrt n_k} \sum_{i\in I_k} f''_{k,i}(\tilde{r}) = o_P(1)$, we first show that $\pi_a + \tilde{r}(\hat{\pi}_{a,k} - \pi_a)$ is bounded below by a constant. Notice that $\pi_a + \tilde{r}(\hat{\pi}_{a,k} - \pi_a) = (1-\tilde{r}) \pi_a + (\tilde{r}) \hat{\pi}_{a,k}$ is a weighed average of $\pi_a$ and $\hat{\pi}_{a,k}$. Since both $\pi_a$ and $\hat{\pi}_{a,k}$ must be larger than $\epsilon$, their average must also be larger than $\epsilon$. Therefore, recalling that $\hat{\pi}_{a,k} - \pi_a = O_P(1/\sqrt{n})$ we have
\begin{align*}
	\left|\frac{1}{\sqrt n_k} \sum_{i\in I_k} f''_{k,i}(\tilde{r})\right| &\leq 2 O_P(1/\sqrt{n}) \epsilon^2 \left| \frac{1}{\sqrt n_k} \sum_{i\in I_k} \sup_{\mu^* \in \mathcal{F}^*_{a,n}} \E(\mu^*(\X_i) - \mu_a^*(\X_i)|\Z_i)\right| \\
	&\quad\quad + O_P(1/n) \epsilon^3 \left| \frac{1}{\sqrt n_k} \sum_{i\in I_k} \sup_{\mu^* \in \mathcal{F}^*_{a,n}} \E(\mu^*(\X_i) - \mu_a^*(\X_i)|\Z_i)\right|  \\
	&\quad\quad + O_P(1/n) \epsilon^3 \left|\frac{1}{\sqrt n_k} \sum_{i\in I_k} 1_a(A_i) \E(\Yai - \mu^*_a(\X_i)|\Z_i) \right| \\
	&\leq 2O_P(1/\sqrt{n}) \epsilon^2 \left| \frac{1}{\sqrt n_k} \sum_{i\in I_k} \sup_{\mu^* \in \mathcal{F}^*_{a,n}} \E(\mu^*(\X_i) - \mu_a^*(\X_i)|\Z_i)\right| + O_P(1/n) \epsilon^3 O_P(1/\sqrt{n_k})
\end{align*} 
where the last term comes from either (B) or (U) holding, as before. Furthermore, recalling Assumption 2* as we did for $f'$, we have: 
\begin{align*}
	 \frac{1}{\sqrt n_k} \sum_{i\in I_k} \sup_{\mu^* \in \mathcal{F}^*_{a,n}} \E(\mu^*(\X_i) - \mu_a^*(\X_i)|\Z_i) &= \frac{1}{\sqrt n_k} \sum_{l=1}^{L} \sum_{i\in I_{k,l}} \sup_{\mu^* \in \mathcal{F}^*_{a,n}} \E(\mu^*(\X_i) - \mu_a^*(\X_i)|\Z_i) \\
	 &= \sum_{l=1}^{L}  \sqrt{\hat p_{k,l}} \sqrt{n_{k,l}} \sup_{\mu^* \in \mathcal{F}^*_{a,n}} \E(\mu^*(\X_i) - \mu_a^*(\X_i)|\Z_i) \\
	 &\leq \sum_{l=1}^{L} \sqrt{n_{k,l}} \sup_{\mu^* \in \mathcal{F}^*_{a,n}} \E(\mu^*(\X_i) - \mu_a^*(\X_i)|\Z_i) \\
	 &= \sum_{l=1}^{L} \sqrt{n_{k,l}} \cdot  o_P(1) \\
	 &= o_P(\sqrt{n_{k,l}}).
\end{align*}
Therefore, $\left|\frac{1}{\sqrt n_k} \sum_{i\in I_k} f''_{k,i}(\tilde{r})\right| = O_P(1/\sqrt{n})o_P(\sqrt{n_{k,l}}) + o_P(1) = o_P(1)$.
\end{proof}

 \subsection{Proof of Theorem \ref{theo:ml}}

    We now show that the influence function of the cross-fitted estimator under Assumptions \ref{assump: car} and 2* is the same as equation \eqref{influence-function}. As a result, all of our other theoretical results hold.

Let $\W = (A, X, \Ya)$. Our estimating function is $\psi$:
\begin{align}
\psi(\bm W_i; \theta_a, \eta_a) &= \frac{1_a(A_i)}{{\pi}_a}\bigg\{\Yai - \mu^*_a(\X_i) \bigg\} + \mu^*_a(\X_i) - \theta_a\label{linear-score}
\end{align}
where the nuisance parameter $\eta_a = (\pi_a, \mu^*_a)$, and $\mu^*_a(\X) = \mu_a(\X) - P_0 \mu_a + \theta_a$. So there is an implicit dependence between the nuisance parameter $\mu_a$ and $\theta_a$. The result of the Theorem 3.1 in \citet{chernozhukov2017double} is that $\sqrt{n}(\hat{\theta}_{CF,a} - \theta_a) = \frac{1}{\sqrt{n}} \sum_{i=1}^{n} \psi(\bm W_i; \theta_a, \eta_a) + o_P(1)$, which is the same influence function as in Theorem \ref{theo:asymptotic-normality} (i). We now follow the proof outline of Theorem 3.1 in \citet{chernozhukov2017double}, built off of our Lemma \ref{cf-remainder} that has modifications to the Step 3 of their theorem that are needed to allow for covariate-adaptive randomization.

First, we have that $\E(\psi(\W; \theta_a, \eta_a)) = 0$. Next, note that the score $\psi$ is linear in the sense that it can be written as $\psi(\W; \theta, \eta) = \psi^a(\W; \eta) \theta + \psi^b(\W; \eta)$,
	where $\psi^a(\W; \eta) = -1$, and $\psi^b(\W; \eta) = \frac{1_a(A)}{{\pi}}\left[\Ya - \mu^*(\X) \right] + \mu^*(\X)$ with $\eta = (\pi, \mu^*)$. Therefore, we can write, for any partition $k$,
	\begin{align*}
		\sqrt{n_k} \left(\check{\theta}_{a,k} - \theta_a\right) &= \sqrt{n_k} \left(\E_{n,k} [\psi^b(\W; \hat{\eta}_{a,k})] - \theta_a\right) \\
		&= \sqrt{n_k} \left(\E_{n,k}[\psi(\W; \theta_a, \hat{\eta}_{a,k}] \right) \\
		&= \frac{1}{\sqrt{n_k}} \sum_{i \in I_k}\psi(\bm W_i; \theta_a, \eta_a) + \mathcal{I}_k
	\end{align*}
	where $\mathcal{I}_k$ is from Lemma \ref{cf-remainder}. Now we can take the average across all $k = 1, ..., K$ partitions to get the cross-fitted estimator on the LHS, and the desired influence function on the RHS:
	\begin{align*}
	\sqrt{n} (\hat{\theta}_{a} - \theta_a) &= \sqrt{n} \left(\frac{1}{K}  \sum_{k=1}^{K} \check{\theta}_{a,k} - \theta_a \right) \\
	&= \sum_{k=1}^{K} \frac{\sqrt{n_k}}{\sqrt{K}} \left(\check{\theta}_{a,k} - \theta_a \right) \\
	&= \sum_{k=1}^{K} \sum_{i\in I_k} \frac{1}{\sqrt{n}} \psi(\bm W_i; \theta_a, \eta_a) + \frac{1}{\sqrt{K}} \sum_{k=1}^{K} \mathcal I_k \\
	&= \frac{1}{\sqrt n} \sum_{i=1}^{n} \psi(\bm W_i; \theta_a, \eta_a) + o_P(1)
	\end{align*}
	since each $\mathcal{I}_k = o_P(1)$ (Lemma \ref{cf-remainder}), their sum (over a finite number of partitions) is also $o_P(1)$.

\begin{remark}\label{remark:cf-g-pia}
	Theorem \ref{theo:ml} holds replacing $\hat{\pi}_{a,k}$ with $\hat{\pi}_a$ in $\hat{\theta}_{\rm CF, a}$ in \eqref{eq: cf}. That is,
	\begin{align}
		\frac{1}{K} \sum_{k=1}^{K} \frac{1}{n_k}\sum_{i\in I_k} \left[\frac{1_a(A_i)}{\hat{\pi}_{a
}}\bigg\{\Yai - \hat{\mu}_{a,k}(\X_i) \bigg\} + \hat{\mu}_{a,k}(\X_i) \right]\label{eq: cf-2}
	\end{align}
	has the same influence function given in \eqref{influence-function}.
	This is because in each step of the proof of Theorem \ref{theo:ml}, we could have equivalently used $\hat{\pi}_a$ without changing the result, asymptotically. Specifically, since we condition on the full $(\mathcal{A}_n,\mathcal{Z}_n)$ in the proof, we could have just as easily removed $\hat{\pi}_{a}$ from the conditional expectation as we did $\hat{\pi}_{a,k}$. Also, $O_p(n_k^{-1/2}) \equiv O_p(n^{-1/2})$, since $n_k / n$ is constant as $n$ grows.
\end{remark}

\begin{lemma}[Equivalence of Cross-Fitted AIPW Estimators in Simulation]\label{lemma: cf-aipw-equivalence}
	Let $\check{\mu}_{a}(\X_i) = \sum_{k=1}^{K} I\{i \in I_k\} \hat{\mu}_{a,k}(\X_i)$ be the estimated $\mu_a(\X_i)$ from AIPW using cross-fitting with $K$ partitions. Then if the partitions are equally sized, i.e., $n_k = n_j$ for all $k = 1, ..., K$, then the AIPW estimator in \eqref{eq: aipw} is equivalent to the cross-fitted AIPW estimator in \eqref{eq: cf-2}.
\end{lemma}

\begin{proof}
	The AIPW estimator using $\check{\mu}_{a}(\X_i) = \sum_{k=1}^{K} I\{i \in I_k\} \hat{\mu}_{a,k}(\X_i)$ is
	\begin{align}\label{eq:aipw-pack}
		\frac1n \sum_{i=1}^n \left[ \frac{1_a(A_i)}{\hat \pi_{a}}  \left\{  \Yai  - \check \mu_a(\X_i) \right\} + \check \mu_a (\X_i)\right].
	\end{align}
	With the additional assumption that $n_k = n_j$ for all $k = 1, ..., K$, we have
	\begin{align*}
		\eqref{eq:aipw-pack} &= \frac1n \sum_{i=1}^n \left[ \frac{1_a(A_i)}{\hat \pi_{a}}  \left\{  \Yai  - \sum_{k=1}^{K} I\{i \in I_k\} \hat{\mu}_{a,k}(\X_i) \right\} + \sum_{k=1}^{K} I\{i \in I_k\} \hat{\mu}_{a,k}(\X_i)\right] \\
		&= \frac{1}{K} \frac{1}{n_k} \sum_{k=1}^{K} \sum_{i \in I_k}\left[ \frac{1_a(A_i)}{\hat \pi_{a}}  \left\{  \Yai  - \sum_{k=1}^{K} I\{i \in I_k\}\hat{\mu}_{a,k}(\X_i) \right\} + \sum_{k=1}^{K} I\{i \in I_k\}\hat{\mu}_{a,k}(\X_i)\right] \\
		&= \frac{1}{K} \sum_{k=1}^{K} \frac{1}{n_k}  \sum_{i \in I_k}\left[ \frac{1_a(A_i)}{\hat \pi_{a}}  \left\{  \Yai  - \hat{\mu}_{a,k}(\X_i) \right\} + \hat{\mu}_{a,k}(\X_i)\right] = \eqref{eq: cf-2}.
	\end{align*}
\end{proof}

\newpage
\section{Guaranteed Efficiency Gain Under (B) + (G1): Proof of Theorem \ref{corollary:efficiency}}

We will show that under condition (B) we have a guaranteed efficiency gain if the following condition (G1) holds:
\begin{align*}
	\begin{array}{c}
{\rm diag} \left\{ \pi_a^{-1} \cov \{ y_a-\mu_a(\X), \mu_a(\X)\}, a=1,...,k \right\} - \cov \{ \Y - \mu (\X), \mu (\X) \} \\
= E\left[ \{  R_Y (\Z) -  R_X(\Z)\} \{  \Omega_{\rm SR} - \Omega(\Z) \}   R_X(\Z)  \right]. \end{array}
\end{align*}
where $R_Y(\Z) = \diag\{ \pi_a^{-1} \E(\Ya - \theta_a | \Z)\}$ and $R_X(\Z) = \diag\{\pi_a^{-1} \E(\mu_a(\X) - \E(\mu_a(\X))|\Z) \}$, so that $R(\Z) = R_Y(\Z) - R_X(\Z)$.
In order to achieve a guaranteed efficiency gain over ANOVA, a simple sample mean, we would need the difference between the covariance matrix of ANOVA and $\hat{\theta}_{AIPW}$ to be positive definite. In what follows, we show that a sufficient condition for this difference being positive definite is that (G1) holds. Let $\bm \Sigma^{\bar{Y}}_{(B)}$ be the variance-covariance matrix for the sample mean $\bar{Y}$ under $(B)$, given by the result in Theorem \ref{theo:asymptotic-normality} with $\tilde{\mu} = \bm 0$ and $R(\Z) = R_Y(\Z)$:
\begin{align*}
	\bm \Sigma^{\bar{Y}}_{(B)} &= \diag\{\pi_a^{-1} \Var(\Ya)\} - \E(R_Y(\Z) \{\Omega_{SR} - \Omega\} R_Y(\Z)).
\end{align*}
Then,
\begin{align*}
	\bm \Sigma^{\bar{Y}}_{(B)} - \bm \Sigma_{(B)} &= \underbrace{\diag\{\pi_a^{-1} \Var(\Ya)\} - \diag\{\pi_a^{-a} \Var\{ \Ya - \mu_a(X)\}}_{(I)} \\
  &\quad\quad \underbrace{- \Cov\{Y, \mu(X)\} - \Cov\{\mu(X), Y\} + \Var\{\mu(X)\}}_{(II)} \\
  &\quad\quad \underbrace{- \E(R_Y(\Z) \{\Omega_{SR} - \Omega(\Z) \} R_Y(\Z)) + E\left[ \{  R_Y (\Z) - R_X(\Z)\} \{  \Omega_{\rm SR} - \Omega(\Z) \} \{  R_Y (\Z) - R_X(\Z)\} \right]}_{(III)}. \\
  &= \diag\{\pi_a^{-1} \Var(\mu_a(X))\} + 2\diag\{\pi_a^{-a} \Cov\{\mu_a(X), \Ya - \mu_a(X)\} \\
  &\quad\quad  - \Cov\{Y - \mu(X), \mu(X)\} - \Cov\{\mu(X), Y - \mu(X)\} - \Var\{\mu(X)\} \\
  &\quad\quad - \E(R_Y(\Z) \{\Omega_{SR} - \Omega\} R_Y(\Z)) + E\left[ \{  R_Y (\Z) - R_X(\Z)\} \{  \Omega_{\rm SR} - \Omega(\Z) \} \{  R_Y (\Z) - R_X(\Z)\} \right]
\end{align*}
We can rewrite each of the above terms. Using the variance of a difference and the law of total variance:
\begin{align*}
    (I) &= \diag\{\pi_a^{-1} \Var(\mu_a(X))\} + 2\diag\{\pi_a^{-1} \Cov\{\mu_a(X), \Ya - \mu_a(X)\}\} \\
    &= \diag\{\pi_a^{-1} \E(\Var(\mu_a(X)|Z))\} + \diag\{\pi_a^{-1} \Var(\E(\mu_a(X)|Z))\} \\
    &\quad\quad + 2\diag\{\pi_a^{-1} \E(\Cov\{\mu_a(X), \Ya - \mu_a(X)|Z\})\} + 2\diag\{\pi_a^{-1} \Cov\{\E(\mu_a(X)|Z), \E(\Ya - \mu_a(X)|Z)\}\}.
\end{align*}
For $(II)$, we expand terms and then add and subtract the quadratic form of $R_{X}$:
\begin{align*}
    (II) &= - \E(R_X(\Z) \{\Omega_{SR} - \Omega(\Z)\} R_Y(\Z)) - \E(R_Y(\Z) \{\Omega_{SR} - \Omega(\Z)\} R_X(\Z)) + \E(R_X(\Z) \{\Omega_{SR} - \Omega(\Z)\} R_X(\Z)) \\
    &= - \E(R_X(\Z) \{\Omega_{SR} - \Omega(\Z)\} \{R_Y(\Z) - R_X(\Z)\}) - \E(\{R_Y(\Z) - R_X(\Z)\} \{\Omega_{SR} - \Omega(\Z)\} R_X(\Z)) \\
    &\quad\quad - \E(R_X(\Z) \{\Omega_{SR} - \Omega(\Z)\} R_X(\Z)).
\end{align*}
For $(III)$, we use law of total covariance and variance, and then use the definition of $R_X$ and $R_Y$. As an example, we have
\begin{align*}
    \Var(\mu(\X)) &= \E(\Var(\mu(\X)|\Z)) + \Var(\E(\mu(\X)|\Z)) \\
    &= \E(\Var(\mu(\X)|\Z)) + \E(\E(\mu(\X)|\Z) \E(\mu(\X)|\Z)^T) - \E(\mu(\X)) \E(\mu(\X))^T \\
    &= \E(\Var(\mu(\X)|\Z)) + \E(R_X(\Z) \pi \pi^T R_X(\Z)).
\end{align*}
The same steps hold for the covariance terms. Therefore,
\begin{align*}
    (III) &= - \E(\Cov\{Y - \mu(X), \mu(X)|\Z\}) - \E(\{R_Y(\Z) - R_X(\Z)\} \pi \pi^T R_X(\Z)|Z) \\
    &\quad\quad - \E(\Cov\{\mu(X), Y - \mu(X)|\Z\}) - \E(R_X(\Z)\pi \pi^T \{R_Y(\Z) - R_X(\Z)\}|Z) \\
    &\quad\quad - \E(\Var(\mu(\X)|\Z)) - \E(R_X(\Z) \pi \pi^T R_X(\Z)).
\end{align*}
We now can combine terms across $(I) + (II) + (III)$. For example, considering the variance terms for $\mu$, and recalling that $\Omega_{SR} = \diag(\pi) - \pi \pi^T$, we have:
\begin{align*}
    &\diag\{\pi_a^{-1} \E(\Var(\mu_a(X)|Z))\} + \diag\{\pi_a^{-1} \Var(\E(\mu_a(X)|Z))\} \\
    &\quad - \E(R_X(\Z) \{\Omega_{SR} - \Omega(\Z)\} R_X(\Z)) - \E(\Var(\mu(\X)|\Z)) - \E(R_X(\Z) \pi \pi^T R_X(\Z))  \\
    &= \diag\{\pi_a^{-1} \E(\Var(\mu_a(X)|Z))\} - \E(\Var(\mu(\X)|\Z)) + \E(R_X(\Z) \Omega(\Z) R_X(\Z)) \\
    &\quad + \E(R_X(\Z) \diag(\pi) R_X(\Z)) - \E(R_X(\Z) \pi \pi^T R_X(\Z)) - \E(R_X(\Z) \Omega_{SR} R_X(\Z)) \\
    &=  \diag\{\pi_a^{-1} \E(\Var(\mu_a(X)|Z))\} - \E(\Var(\mu(\X)|\Z)) + \E(R_X(\Z) \Omega(\Z) R_X(\Z)).
\end{align*}
We have a similar result for each of the covariance terms. Using the law of total covariance, and the fact that again $\Omega_{SR} = \diag(\pi) - \pi \pi^T$, we have:
\begin{align*}
    &\diag(\pi_a^{-1} \E(\Cov(y_a - \mu_a(X), \mu_a(X)|Z))\} - \E(\Cov(Y - \mu(X), \mu(X)|Z)) \\
    &\quad\quad = \diag(\pi_a^{-1} \Cov(y_a - \mu_a(X), \mu_a(X))\} - \Cov(Y - \mu(X), \mu(X)) - \E(\{R_Y(Z) - R_X(Z) \}\Omega_{SR} R_X(Z)).
\end{align*}
Therefore,
\begin{align*}
    &\diag(\pi_a^{-1} \E(\Cov(y_a - \mu_a(X), \mu_a(X)|Z))\} - \E(\Cov(Y - \mu(X), \mu(X)|Z)) \\
    &\quad\quad + \E(\{R_Y(Z) - R_X(Z) \}\Omega(Z) R_X(Z)) \\
    &= \diag(\pi_a^{-1} \Cov(y_a - \mu_a(X), \mu_a(X))\} - \Cov(Y - \mu(X), \mu(X)) - \E(\{R_Y(Z) - R_X(Z) \}\{\Omega_{SR} - \Omega(Z)\}R_X(Z)).
\end{align*}
This also holds for the transpose matrix, i.e., dealing with terms like $\Cov(\mu(X), Y - \mu(X))$.
Putting all of these parts together, we have the following result:
\begin{align*}
    (I) + (II) + (III) &= \diag\{\pi_a^{-1} \E(\Var(\mu_a(X)|Z))\} - \E(\Var(\mu(\X)|\Z)) + \E(R_X(\Z) \Omega(\Z) R_X(\Z)) \\
    &\quad + 2\diag(\pi_a^{-1} \E(\Cov(\mu_a(\X), \Ya - \mu_a(\X)|\Z)) \\
    &\quad - \Cov(Y - \mu(X), \mu(X)) - \Cov(Y - \mu(X), \mu(X))^T \\
    &\quad - \E(\{R_Y(\Z) - R_X(\Z)\} \{\Omega_{SR} - \Omega(\Z)\} R_X(\Z)) \\
    &\quad - \E(\{R_Y(\Z) - R_X(\Z)\} \{\Omega_{SR} - \Omega(\Z)\} R_X(\Z))^T.
\end{align*}
With the condition given in 1 (d), all terms cancel out besides the first line. In other words, under 1 (d),
\begin{align*}
    \bm \Sigma^{\bar{Y}}_{(B)} - \bm \Sigma_{(B)} = \diag\{\pi_a^{-1} \E(\Var(\mu_a(X)|Z))\} - \E(\Var(\mu(\X)|\Z)) + \E(R_X(\Z) \Omega(\Z) R_X(\Z)).
\end{align*}
Note that $\Omega(Z)$ is positive semidefinite. Furthermore, Lemma 1 from \citet{Ye2021better} states that if $M$ is a matrix with columns $m_1, ..., m_k$ and $\pi_1, ..., \pi_k$ are nonnegative constants that sum to one, then $\diag(\pi_t^{-1} m_t^T m_t) - M^T M$ is positive semidefinite. This means that $\diag\{\pi_a^{-1} \E(\Var(\mu_a(X)|Z))\} - \E(\Var(\mu(\X)|\Z))$ is positive semidefinite, as is its expectation. Therefore, $\bm \Sigma^{\bar{Y}}_{(B)} - \bm \Sigma_{(B)}$ is positive semidefinite and we have guaranteed efficiency gain.

\newpage
\section{Universal Applicability: Proof of Theorem \ref{theo:universality}}

Under (U), we have that $\E \{ \Y - \mu (\X)  \mid \Z\} = \theta - E\{ \mu (\X) \}$ almost surely. Since we defined $\tilde{\mu} = \mu(\X) - \E(\mu(\X)) - \theta$, this means that under (U), $\E \{ \Y - \tilde{\mu} (\X)  \mid \Z\} = \bm 0$. As a result, $\bm \phi^2 = \bm \phi^3 = \bm 0$ almost surely. Therefore, $\bm V^{(5)} := \bm V^{(3,5)} = \Var(\E(\tilde{\mu}|\Z))$.

Define  $\bm M_n^{(1, 4)} := \left[\bm V^{(1,4)}\right]^{-1}\sqrt{n} \E_n(\bm \phi^1 + \bm \phi^4)$ and $\bm M_n^{(5)} := \left[\bm V^{(3,5)}\right]^{-1}\sqrt{n} \E_n(\bm \phi^5)$. Then by the continuous mapping theorem, $\bm M_n^{(1, 4)} | \mathcal{A}_n,\mathcal{Z}_n \xrightarrow{d}  N(0, \bm I_k)$ and $\bm M_n^{(5)} \xrightarrow{d}  N(0, \bm I_k)$. Furthermore, these two normal distributions are asymptotically independent. Applying Lemma \ref{lemma:asymptotic-independence},
$(\bm M_n^{(1, 4)}, \bm M_n^{(5)}) \xrightarrow{d} (\bm M^{(1, 4)}, \bm M^{(5)})$, where $\bm M^{(1,4)}$ and $\bm M^{(5)}$ are independent multivariate normal random variables.
Therefore, by the continuous mapping theorem, we can both scale them by their respective variance-covariance matrices, and add the two random variables together to get that under (B),
\begin{align*}
	\sqrt{n} (\hat{\theta}_{AIPW} - \theta) \equiv \sqrt{n} \E_n(\bm \phi^1 + \bm \phi^4 + \bm \phi^5) + o_P(1)\xrightarrow{d} N(0, \bm \Sigma_{(U)})
\end{align*}
where $\bm \Sigma_{(U)} = \bm V^{(1,4)} + \bm V^{(5)}$. Note that $\bm \Sigma_{(U)}$ has no dependence on the covariate-adaptive randomization scheme, so the asymptotic distribution (including the asymptotic variance) of $\hat{\theta}_{AIPW}$ is universal under (U).

It is now straightforward to show guaranteed efficiency gain under (U) + (G2). Under (U), $R(Z) = \bm 0$. Therefore, the RHS of (G1) is $\bm 0$. If (G2) holds, then both covariance terms on the LHS of (G1) are also $\bm 0$. Since both sides of (G1) are $\bm 0$, (G1) is satisfied.
\newpage

\section{Joint Calibration}

\subsection{Joint Calibration Satisfies Assumptions 2 and 3}

Here we show that Assumptions 2 and 3 are satisfied for the jointly calibrated $\hat{\mu}_a^*(X_i)$ and a limiting function $\mu_a^*(X_i)$. Recall that for joint calibration, we have defined

\begin{equation*}
\begin{array}{l}
	\hat \mu_a^* (\X_i) =  \hat \balpha_{a}^T \Z_i +\hat {\bm  \beta}_a^T  \hat {\mu} (\X_i) , \\
 (\hat \balpha_{a} , \, \hat {\bm \beta}_a) = {\displaystyle \argmin_{ \balpha_{a}, \, \bm \beta_a  } \, \sum_{i: A_i = a}\left\{ Y_{a,i}  -  \balpha_{a}^T \Z_i  -  \bm  \beta_a ^T  \hat {\mu} (\X_i) \right\}^2}. \end{array}
 \end{equation*}
 
Define $\mu_a^*(X_i) = \bm \alpha^{*T}_a Z_i + \bm \beta_a^{*T} {\mu}(X_i)$, where $\mu(X_i) = (\mu_1(X_i), ..., \mu_k(X_i))^T$ is the vector of all $\mu$ functions. Define the true regression parameters as the solution to the corresponding population problem, where $\hat{\mu}$ has been replaced with $\mu$:
\begin{align}
	 (\balpha^*_{a} , \, \bm \beta^*_a) = {\displaystyle \argmin_{ \balpha_{a}, \, \bm \beta_a  } \, \E\left\{\left( \Ya  -  \balpha_{a}^T \Z  -  \bm  \beta_a ^T {\mu}^* (\X)\right)^2 \right\}}.\label{jc:pop-min}
\end{align}
By assumption, Assumptions 2 and 3 are satisfied for $\hat\mu_a(X)$ using some $\mu_a(X)$ for each $a$. Our goal is to show that, given Assumptions 2 and 3 for each $\hat\mu_a(X)$, Assumptions 2 and 3 are satisfied for $\hat \mu_a^* (\X)$ using the above $\mu_a^* (\X)$.

\subsubsection{Consistency of Regression Coefficients}\label{sec:reg-consistency}

We first establish that $\hat{\bm\alpha}_a \to \bm \alpha^{*}_a$ and $\hat{\bm\beta}_a \to \bm \beta^{*}_a$ in probability. If $\hat{\mu}$ were not random this would be trivial as they are just coefficients from OLS. Will figure out what $\bm \alpha^*_a$ and $\bm \beta^*_a$ are, and show that the corresponding $\hat{\bm \alpha}_a$ and $\hat{\bm \beta}_a$ converge in probability. Will likely need to use Assumptions 2 and 3 for each $\hat{\mu}_a$.

Following the same logic as in Section 3.1 of \citet{Ye2021better}, we have that
\begin{align*}
\begin{pmatrix}
    \hat{\bm \alpha}_a \\
    \hat{\bm \beta}_a
\end{pmatrix} = \frac{n}{n_a} \left\{ \sum_{i=1}^{n}
\begin{pmatrix}
    Z_i \\
    \hat{\mu}(\X_i)
\end{pmatrix}
\begin{pmatrix}
    Z_i \\
    \hat{\mu}(\X_i)
\end{pmatrix}^T
\right\}^{-1}
\sum_{i:A_i=a} \begin{pmatrix}
    Z_i Y_i \\
    \hat{\mu}(\X_i) Y_i
\end{pmatrix}
\end{align*}
Therefore, to prove that $\hat{\bm\alpha}_a \to \bm \alpha^{*}_a$, and $\hat{\bm\beta}_a \to \bm \beta^{*}_a$, it suffices to show that $\frac{1}{n} \sum_{i=1}^{n}
\begin{psmallmatrix}
    Z_i \\
    \hat{\mu}(\X_i)
\end{psmallmatrix}
\begin{psmallmatrix}
    Z_i \\
    \hat{\mu}(\X_i)
\end{psmallmatrix}^T \to \E\left\{ \begin{psmallmatrix}
    Z \\
    {\mu}(\X)
\end{psmallmatrix}
\begin{psmallmatrix}
    Z \\
    {\mu}(\X)
\end{psmallmatrix}^T\right\}$ in probability, and that $\frac{1}{n_a} \sum_{i:A_i=a} \begin{psmallmatrix}
    Z_i \Yai \\ 
    \hat{\mu}(\X_i) \Yai
\end{psmallmatrix} \to\E \left\{ \begin{psmallmatrix}
    Z \Ya \\ 
    {\mu}(\X) \Ya
\end{psmallmatrix}\right\}$ in probability. For each of the matrix/vector components, we can show this. The components without 
$\hat{\mu}$ are straightforward. Specifically, $\frac{1}{n} \sum_{i=1}^{n} Z_i Z_i^T \to \E(Z Z^T)$ by the law of large numbers since $Z_i$ are iid. We cite \citet{Ye2021better}'s work in Lemma 3 of their supplemental materials to claim that $\frac{1}{n_a} \sum_{i:A_i=a} Z_i \Yai \to \E(Z \Ya)$ in probability. Turning to terms with $\hat{\mu}$, we start with $\frac{1}{n_a} \sum_{i:A_i=a} \Yai \hat{\mu}(\X_i)$:
\begin{align*}
    \frac{1}{n_a} \sum_{i:A_i=a} \Yai \hat{\mu}(\X_i) - \E(Y_i \mu(\X_i)) &= \underbrace{\frac{1}{n_a} \sum_{i:A_i=a} \Yai {\mu}(\X_i) - \E(\Yai \mu(\X_i))}_{(1)} \\
    &\quad + \underbrace{\frac{1}{n_a} \sum_{i:A_i=a} \Yai \hat{\mu}(\X_i) - \frac{1}{n_a} \sum_{i:A_i=a} \Yai {\mu}(\X_i)}_{(2)}
\end{align*}
Term (1) is $o_P(1)$ by the same logic of \citet{Ye2021better} Lemma 3. For term (2), we have similar types of empirical process terms as we dealt with in the proof of Theorem 1, when finding the influence function. However, the difference here is that we do not need a specific rate; we just need it to be $o_P(1)$. We have:
\begin{align*}
    (2) &= \frac{1}{n_a} \sum_{i:A_i=a} \Yai \hat{\mu}(\X_i) - \frac{1}{n_a} \sum_{i:A_i=a} \Yai {\mu}(\X_i) \\
    &= \hat\pi_a^{-1} \frac{1}{n} \sum_{i:A_i=a} [\Yai \hat{\mu}(\X_i) - P_0(\Yai \hat{\mu}(\X_i))] - \hat\pi_a^{-1} \frac{1}{n} \sum_{i=1}^{n} 1_a(A_i) [\Yai {\mu}(\X_i) - P_0(\Yai \mu(\X_i))] \\
    &\quad + P_0(\Yai (\hat{\mu}(\X_i) - \mu(\X_i))) \\
    &= \hat\pi_a^{-1} \frac{1}{n} \sum_{i=1}^{n} 1_a(A_i) \left\{\Yai \hat{\mu}(\X_i) - P_0 \Yai \hat{\mu}(\X_i) - \Yai \mu(\X_i) + P_0 \Yai \mu(\X_i) \right\} + P_0(\Yai (\hat{\mu}(\X_i) - \mu(\X_i))).
\end{align*}
Since $\hat{\pi}_a^{-1} \to \pi_a^{-1}$ in probability, by Slutsky's Lemma, to show that (2) is $o_P(1)$, it suffices to show that
% Let , and where we have used the fact that $\hat{\pi}_a^{-1} \to \pi_a^{-1}$ in probability. In order for (2) to be $o_P(1)$, we need the following to hold:
\begin{enumerate}[label=(\Roman*)]
    \item $P_0(\Ya \hat{\mu}(\X) - \Ya \mu(\X)) = o_P(1)$. By the Cauchy-Schwartz Inequality, we have
    \begin{align*}
        |P_0(\Ya \hat{\mu}(\X) - \Yai \mu(\X))| &\leq \norm{\Ya}_{L_2} \norm{\hat{\mu}(\X) - \mu(\X)}_{L_2} \\
        &= o_P(1)
    \end{align*}
    by assumption since $\norm{\hat{\mu}(\X_i) - \mu(\X_i)}_{L_2} \to 0$ in probability, and since $\Ya$ has a finite second moment. Note that this also holds replacing $P_0(\Ya [\hat{\mu}(\X) - \mu(\X)])$ with $P_0(|\Ya| |\hat{\mu}(\X) - \mu(\X)|)$, therefore $P_0(|\Ya| |\hat{\mu}(\X) - \mu(\X)|)$ is also $o_P(1)$.
    \item $\frac{1}{n} \sum_{i=1}^{n} 1_a(A_i) \left\{\Yai \hat{\mu}(\X_i) - P_0 \Yai \hat{\mu}(\X_i) - \Yai \mu(\X_i) + P_0 \Yai \mu(\X_i) \right\} = o_P(1)$. This term is exactly the same as term $\mathcal{M}$ in the proof of Theorem 1(a), but instead of $\hat{\mu}_a(\X_i)$ and $\mu_a(\X_i)$, we have $\Yai \hat{\mu}(\X_i)$ and $\Yai \mu(\X_i)$ respectively. However, we do not want to follow the techniques of $\mathcal{M}$ in the proof of Theorem 1(a) because it requires too strict of new assumptions. We only need the term to be $o_P(1)$ rather than $o_P(1/\sqrt{n})$. So in this case, we can use some more basic techniques. Term (II) can be broken up into two terms, each of which is $o_P(1)$ by applications of Assumption \ref{assump: stability}, Cauchy-Schwartz, and Markov's Inequality:
    \begin{enumerate}[label=(\roman*)]
        \item $\frac{1}{n} \sum_{i=1}^{n} 1_a(A_i) \left\{\Yai[\hat{\mu}(\X_i) - \mu(\X_i)] \right\}$. Here we can apply the finite version of Cauchy-Schwartz:
        \begin{align*}
            \left|\frac{1}{n} \sum_{i=1}^{n} 1_a(A_i) \left\{\Yai[\hat{\mu}(\X_i) - \mu(\X_i)] \right\}\right| &\leq \sqrt{\frac1n \sum_{i:A_i=a} [\Yai]^2} \sqrt{\frac1n \sum_{i:A_i=a} [\hat{\mu}(\X_i) - \mu(\X_i)]^2} \\
            &\leq \sqrt{\frac1n \sum_{i=1}^{n} [\Yai]^2} \sqrt{\frac1n \sum_{i=1}^n [\hat{\mu}(\X_i) - \mu(\X_i)]^2}.
        \end{align*}
        We know that $\sqrt{\frac1n \sum_{i=1}^{n} [\Yai]^2} - \norm{\Yai}_{L_2} = o_P(1)$ by the law of large numbers and continuous mapping theorem. We also know that $\sqrt{\frac1n \sum_{i=1}^n [\hat{\mu}(\X_i) - \mu(\X_i)]^2} = o_P(1)$ by Assumption:
        \begin{align*}
            \lim_{n\to\infty} P\left(\frac1n \sum_{i=1}^n [\hat{\mu}(\X_i) - \mu_a(\X_i)]^2 > \epsilon\right) &= \lim_{n\to\infty} \E\left( P\left(\frac1n \sum_{i=1}^n [\hat{\mu}(\X_i) - \mu(\X_i)]^2 > \epsilon \Bigg| \hat{\mu}\right)\right) \\
            \text{(DCT)} \quad &= \E\left(\lim_{n\to\infty} P\left(\frac1n \sum_{i=1}^n [\hat{\mu}(\X_i) - \mu(\X_i)]^2 > \epsilon \Bigg| \hat{\mu}\right)\right) \\
            \text{(Markov's Inequality)} \quad &\leq \E\left(\lim_{n\to\infty} P\left(\frac1n \sum_{i=1}^n [\hat{\mu}(\X_i) - \mu(\X_i)]^2 \Bigg|\hat{\mu} \right) \right)/\epsilon \\
            \text{(Assumption \ref{assump: stability})} \quad &= \E(0) / \epsilon = 0.
        \end{align*}
        Therefore, by applying Slutsky's Theorem, we can see that the whole term is $o_P(1)$.
        \item $\frac{1}{n} \sum_{i=1}^{n} 1_a(A_i) P_0 \left\{\Yai[\hat{\mu}(\X_i) - \mu(\X_i)] \right\}$. To see that this term is $o_P(1)$, note that the $P_0$ term is not dependent on $i$, therefore, it becomes $\hat{\pi}_a P_0 \left\{\Ya[\hat{\mu}(\X) - \mu(\X)] \right\}$. Furthermore,
        \begin{align*}
            \left|\hat{\pi}_a P_0 \left\{\Yai[\hat{\mu}(\X) - \mu(\X)] \right\}\right| &\leq \hat{\pi}_a P_0 \left| \Ya[\hat{\mu}(\X) - \mu(\X)] \right| \\
            &\leq \hat{\pi}_a \norm{\hat{\mu}(\X) - \mu(\X)}_{L_2} \norm{\Ya}_{L_2} \\
            &= \hat{\pi}_a o_P(1) \\
            &= (\hat{\pi}_a - \pi_a) o_P(1) + \pi_a o_P(1) = o_P(1)
        \end{align*}
        where we have used Cauchy-Schwartz, the fact that $\norm{\hat{\mu}(\X) - \mu(\X)}_{L_2} = o_P(1)$ by Assumption \ref{assump: stability}, and that $\Ya$ has a bounded second moment.
    \end{enumerate}
\end{enumerate}

Next, we can see that $\frac1n \sum_{i=1}^{n} \Z_i \hat{\mu}(\X_i) \to \E(\Z \mu(\X))$ by using the following expansion:
\begin{align*}
    \frac1n \sum_{i=1}^{n} \Z_i \hat{\mu}(\X_i) - \E(\Z \mu(\X)) &= (P_n - P_0) \Z_i \mu(\X_i) + P_0(Z_i [\hat{\mu}(\X_i) - \mu(\X_i)]) + (P_n - P_0) (\Z_i \hat{\mu} - \Z_i \mu)
\end{align*}
The first term is $o_P(1)$ by the law of large numbers. The second term is $o_P(1)$ for the same reason as (I) above replacing $\Yai$ with $\Z_i$, noting that $\Z_i$ is just an indicator function. The last term is $o_P(1)$ because the functions $(\Z_i, X_i) \to \Z_i \hat{\mu}(\X_i)$ and $(\Z_i, X_i) \to \Z_i {\mu}(\X_i)$ satisfy Assumption \ref{assump: simple}, the Donsker condition, assuming that it is already satisfied for $\hat{\mu}(\X_i)$ and $\mu(\X_i)$.

Finally, we turn to $\hat{\mu}(\X_i) \hat{\mu}(\X_i)^T$. We would like to show that
\begin{align}\label{eq:mumu-const}
    \frac1n \sum_{i=1}^{n} \hat{\mu}(\X_i) \hat{\mu}(\X_i)^T \to \E(\mu(\X) \mu(\X)^T)
\end{align}
in probability. We focus on each of the components of the matrix $\E(\mu(\X_i) \mu(\X_i)^T)$, i.e., we would like to show that $\frac1n \sum_{i=1}^{n} \hat{ \mu}_a(\X_i) \hat{\mu}_b(\X_i) \to \E(\mu_a(\X_i) \mu_b(\X_i))$ in probability, where $a = 1, ..., k$, $b = 1, ..., k$. We cannot simply use the law of large numbers because of the estimated $\hat{\mu}$. However, we can decompose this problem into several parts. Let $\hat{g}_{a,b} = \hat{ \mu}_a(\X_i) \hat{\mu}_b(\X_i)$ and $g_{a,b} = { \mu}_a(\X_i) {\mu}_b(\X_i)$. Then,
\begin{align*}
    \frac1n \sum_{i=1}^{n} \hat{ \mu}_a(\X_i) \hat{\mu}_b(\X_i) - \E(\mu_a(\X) \mu_b(\X)) &= (P_n - P_0) g_{a,b} + P_n (\hat{g}_{a,b} - g_{a,b}).
\end{align*}
$(P_n - P_0) g_{a,b} = o_P(1)$ by the law of large numbers. Therefore, for \eqref{eq:mumu-const} to hold, we need $P_n (\hat{g}_{a,b} - g_{a,b}) = o_P(1)$. We can expand this term as:
\begin{align*}
    \frac1n \sum_{i=1}^{n} \left(\underbrace{\left[\hat{\mu}_a(\X_i) - \mu_a(\X_i) \right]\left[\hat{\mu}_b(\X_i) - \mu_b(\X_i) \right]}_{(III)} + \underbrace{{\mu}_b(\X_i)\left[ \hat{\mu}_a(\X_i) - \mu_a(\X_i)\right]}_{(IV)} + \underbrace{\mu_a(\X_i)\left[ \hat{\mu}_b(\X_i) - \mu_b(\X_i)\right]}_{(V)}\right)
\end{align*}
We will use Cauchy-Schwartz Inequality on each of these terms. In (II) (i) from before, we showed that $\sqrt{\frac1n \sum_{i=1}^n [\hat{\mu}(\X_i) - \mu(\X_i)]^2} = o_P(1)$, so this holds for the $\mu_a$ and $\mu_b$ components of the $\mu$ vector. Therefore, we have
\begin{align*}
    (III) &\leq \sqrt{\frac1n \sum_{i=1}^{n} (\hat{\mu}_a(\X_i) - \mu_a)^2} \sqrt{\frac1n \sum_{i=1}^{n} (\hat{\mu}_b(\X_i) - \mu_b)^2} = o_P(1) \\
    (IV) &\leq \left(\sqrt{\frac1n \sum_{i=1}^{n} \mu_b^2(X_i)} - \norm{\mu_b}_{L_2(P)}\right) \sqrt{\frac1n \sum_{i=1}^{n} (\hat{\mu}_a(\X_i) - \mu_a)^2} + \norm{\mu_b}_{L_2(P)}\sqrt{\frac1n \sum_{i=1}^{n} (\hat{\mu}_a(\X_i) - \mu_a)^2} \\
    &= o_P(1).
\end{align*}
Term (V) is identical to term (IV) with $\mu_a$ and $\mu_b$ switched. Therefore, $P_n (\hat{g}_{a,b} - g_{a,b}) = o_P(1)$.

\subsubsection{Assumption \ref{assump: stability} for Joint Calibration: Stability}

Assume that for every $a$, there exists a function ${\mu}_a \in L_2(P)$ on $\mathscr{X}$ such that, as $n \to \infty$ (Assumption \ref{assump: stability}), 
 $\norm{\hat{\mu}_a - \mu_a}_{L_2} \to 0 $ in probability. First note that $\mu_a^* \in L_2(P)$ as well, because:
 \begin{align*}
     \norm{\mu^*_a}_{L_2} &\leq \norm{\bm \alpha^{*T}_a}_2 + \sum_{a'=1}^{k} |\bm \beta_{a,a'}^{*T}| \norm{\mu_a'}_{L_2} < \infty.
 \end{align*}
 Then:
 \begin{align*}
 	\norm{\hat{\mu}_a^* - \mu_a^*}_{L_2} &\leq \norm{\hat{\bm \alpha}_a - \bm \alpha^*_a }_{L_2} + \sum_{a'=1}^{k} \norm{\hat{\bm \beta}_{a,a'} \hat{\mu}_{a'} - \bm \beta_{a,a'}^* \mu^*_{a'}}_{L_2} \\
 	&\leq \norm{\hat{\bm \alpha}_a - \bm {\alpha^*}_a}_2 + \sum_{a=1}^{k} |\hat{\bm \beta}_{a,a'} - \bm \beta_{a,a'}|\norm{\hat{\mu}_{a'}}_{L_2} + |\bm \beta_{a,a'}|\norm{\hat{\mu}_{a'} - \mu^*_{a'}}_{L_2} \\
  &\leq \norm{\hat{\bm \alpha}_a - \bm {\alpha^*}_a}_2 + \sum_{a'=1}^{k} |\hat{\bm \beta}_{a,a'} - \bm \beta_{a,a'}|\left(\norm{{\mu}_{a'}}_{L_2} + \norm{\hat{\mu}_{a'} - \mu_{a'}}_{L_2} \right) + |\bm \beta_{a,a'}|\norm{\hat{\mu}_{a'} - \mu^*_{a'}}_{L_2}.
 \end{align*}
 This entire expression on the RHS is $o_P(1)$ because of the consistency of the regression coefficients that we proved in Section \ref{sec:reg-consistency} and because of Assumption \ref{assump: stability} on $\hat{\mu}_a$. Therefore, Assumption \ref{assump: stability} holds for $\hat{\mu}^*_a$ and $\mu^*_a$.

\subsubsection{Assumption \ref{assump: simple} for Joint Calibration: Donsker Condition}
 
Assume that Assumption \ref{assump: simple} (i-iii) holds for $\mu_a(X_i)$ and $\hat{\mu}_a(X_i)$, $a = 1, ..., k$. Define
\begin{align*}
	\mathcal{F}^*_{a} := \{\bm \alpha^{T} Z_i + \bm \beta^{T} {\mu}(X_i): \mu_a \in \mathcal{F}_{a}, |\alpha_a - \alpha_a^*| < 1, |\beta_a - \beta_a^*| < 1, a = 1, ..., k\}.
\end{align*}
Then,
\begin{enumerate}[label=(\roman*)]
	\item $\mu_a^* \in \mathcal{F}^*_{a}$. By assumption, each $\mu_a \in \mathcal{F}_{a}$, therefore, for $\mu^*_a \in \mathcal{F}^*_{a}$.
	\item $P(\hat\mu^*_a  \in \mathcal{F}^*_{a} )\rightarrow 1$ as $n \to \infty$. We can write
 \begin{align*}
     P(\hat\mu^*_a  \notin \mathcal{F}^*_{a} ) \leq \sum_{a'=1}^{k} \left\{P(\hat{\mu}_{a'} \notin \mathcal{F}_{a'}) + P(|\hat{\bm \alpha}_{a,a'} - \bm \alpha_{a,a'}^*| \geq 1) + P(|\hat{\bm \beta}_{a,a'} - \bm \beta_{a,a'}^*| \geq 1) \right\} \to 0
 \end{align*}
 by Assumption \ref{assump: simple} on $\hat{\mu}_{a'}$ and by definition of convergence in probability.
\item $\int_0^1 \sup_Q \sqrt {\log {N( \mathcal{F}^*_{a}, \|\cdot \|_{L_2(Q)} ,s)} } \ d s < \infty$. Let $T_a^{\mu}$ be a minimal $\epsilon$-cover for $\mathcal{F}_a$, $T_l^{\alpha}$ be a minimal $\epsilon$-cover for the unit ball in 1-dimension of $\alpha_l - \alpha_l^*$ (letting $\bm \alpha_a = (\alpha_1, ..., \alpha_L)$ and ignoring the $a$ subscript on $\bm \alpha_a$ for simplicity), and similarly $T_{a}^{\beta}$ be a minimal $\epsilon$-cover for the unit ball in 1-dimension of $\beta_a - \beta_a^*$  (letting $\bm \beta_a = (\beta_1, ..., \beta_k)$ and ignoring the $a$ subscript on $\bm \beta_a$ for simplicity).

Take any function $f \in \mathcal{F}_a^*$. Select the elements of $T_a^{\mu}$, $T_l^{\alpha}, l = 1, ..., L$, and $T_a^{\beta}, a = 1, ..., k$ such that they cover the $\alpha_l$'s, $\beta_a$'s and $\mu_a$'s of the chosen $f$, and denote them as $\alpha_l'$, $\beta_a'$, and $\mu_a'$, respectively. All of these functions are contained within the class
\begin{align*}
    T := \{ (\alpha_1, ..., \alpha_L)^T Z_i + (\beta_1, ..., \beta_k)^T (\mu_1(X_i), ..., \mu_k(X_i)): \alpha_l - \alpha_l^* \in T_l^{\alpha}, \beta_a - \beta^*_a \in T_a^{\beta}, \mu_a \in T_a^{\mu} \}.
\end{align*}
We define the function $f'(Z_i, X_i) := (\alpha_1', ..., \alpha_L')^T Z_i + (\beta_1', ..., \beta_k')^T (\mu_1'(X_i), ..., \mu_k'(X_i))$ as the covering function for $f$. Then:
\begin{align*}
    \norm{f - f'}_{L_2} &\leq \sum_{l=1}^{L} |\alpha_l - \alpha'_l| + \sum_{a=1}^{k} |\beta_a - \beta_a'| \norm{\mu_a}_{L_2} + |\beta'_a|\norm{\mu_a - \mu_a'}_{L_2} \\
    &< L \epsilon + 2k C \epsilon
\end{align*}
where $C > \max_a\{|1 + \beta_a^*|, \norm{\mu_a}_{L_2}, a = 1, ..., k\}$, which we know is $< \infty$ because $\beta_a^* < \infty$ and functions in $\mathcal{F}_a$ are square integrable. Therefore, we can find a constant $1 \leq C' < \infty$ such that $\norm{f - f'}_{L_2} < C'\epsilon$.

Now that we have shown that $T$ is a $C'\epsilon$ cover for $\mathcal{F}_a^*$, we will show that the uniform entropy integral also holds. We know that $N(\mathcal{F}_a^*, \norm{\cdot}_{L_2}, C'\epsilon) \leq \prod_{l=1}^{L} |T_{l}^\alpha| \prod_{a=1}^{k}  |T_{a}^{\beta}||T_a^{\mu}|$. Therefore,
\begin{align*}
    \log N(\mathcal{F}_a^*, \norm{\cdot}_{L_2}, C'\epsilon) \leq \sum_{a=1}^{k} \log N(\mathcal{F}_a, \norm{\cdot}_{L_2}, \epsilon) + \sum_{a=1}^k (M/\epsilon) + \sum_{l=1}^{L} (M/\epsilon)
\end{align*}
for some constant $M < \infty$, where the covering numbers for the unit ball in 1-dimension is given by Example 5.8 in Wainwright (2019). Taking the square root, the supremum, then the integral, we have:
\begin{align*}
    \int_0^{1} \sup_Q \sqrt{ \log N(\mathcal{F}_a^*, \norm{\cdot}_{L_2(Q)}, C'\epsilon) } d\epsilon \leq \sum_{a=1}^{k} \int_0^{1} \sup_Q \sqrt{\log N(\mathcal{F}_a, \norm{\cdot}_{L_2}, \epsilon)} d\epsilon + (k + L) M \int_0^{1} \sqrt{(1/\epsilon)} d\epsilon.
\end{align*}
$\int_0^{1} \sqrt{(1/\epsilon)} d\epsilon$ is finite, and $\int_0^{1} \sup_Q \sqrt{\log N(\mathcal{F}_a, \norm{\cdot}_{L_2}, \epsilon)} d\epsilon$ for all $a$ is finite by Assumption \ref{assump: simple}. Therefore, the RHS of the above is finite. Finally, for the LHS, let $\epsilon' = C'\epsilon$. Then, LHS is equivalent to
\begin{align*}
    (C')^{-1} \int_0^{C'} \sup_Q \sqrt{ \log N(\mathcal{F}_a^*, \norm{\cdot}_{L_2(Q)}, \epsilon') } d\epsilon'
\end{align*}
Since $1 \leq C' < \infty$, we can see that $\int_0^{1} \sup_Q \sqrt{ \log N(\mathcal{F}_a^*, \norm{\cdot}_{L_2(Q)}, \epsilon) } d\epsilon < \infty$.
\end{enumerate}

 \newpage
\subsection{Proof of Theorem \ref{theo: double}}

We have already shown that Assumptions 2 and 3 hold for $\hat{\mu}^*_a$ with respect to the function $\mu_a^*(X_i) = \bm \alpha^{*T}_a Z_i + \bm \beta_a^{*T} {\mu}(X_i)$. Let $\hat{W} = (Z^T, \hat{\mu}(X))$, ${W} = (Z^T, \mu(X))$, and $\Sigma_W = \Var(W)$. Let $\gamma_a = (\bm \alpha^{*T}_a, \bm \beta_a^{*T})^T$ and $\hat{\gamma}_a = (\hat{\bm \alpha}^{T}_a, \hat{\bm \beta}_a^{T})^T$. With this simplified notation, $\hat{\mu}^*_a(X) = \hat{\gamma}_a^T \hat{W}$, and $\mu_a^*(X) = \gamma_a^TW$. Based on the fact that we have shown that our OLS regression coefficients are consistent for their population minimizers based on \eqref{jc:pop-min}, we know that $\hat{\gamma}_a \to \gamma_a \equiv \Var(W)^{-1} \Cov(W, y_a)$.

To claim universality with (U), it suffices to show that for every $a$, 
	$E\{ \Ya - \mu_a^*(\X) + P_0 \mu_a^* - \theta_a \mid \Z\} = 0$.
This is satisfied by the population score equations based on how we have defined $\bm \alpha^*_a$ and $\bm \beta^*_a$, because for each $l = 1, ..., L$, we have
	$E[1(Z_i=l) \{\Ya - \mu_a^*(\X)\}] = 0$.
Therefore, we can directly apply Theorem \ref{theo:universality} to obtain the variance of $\hat{\mu}_a^*$. Furthermore, now that we have established universality, we only need the following to hold in order to claim guaranteed efficiency gain $\Cov\{ \Ya - \mu^*_a(X), \mu^*_b(X)\} = 0$,
for all $a = 1, ..., k$, $b = 1, ..., k$.
By the definition of covariance, it suffices to show that both $\E(( \Ya -  \mu^*_a(X)) \mu_b^*(X)) = 0$ and $\E(\Ya - \mu^*_a(X)) = 0$. The latter is satisfied by population score equations for group $a$ since we have a $\Z$-specific intercept, and because we can take expectation over $\Z$. For $\E((\Ya - \mu^*_a(X)) \mu_b^*(X))$, we have
\begin{align*}
	\E((\Ya - \mu^*_a(X)) \mu_b^*(X)) &= \E((\Ya - \mu^*_a(X)) \{\bm \alpha^{*T}_b Z + \bm \beta_b^{*T} {\mu}(X)\}) \\
	&= \bm \alpha^{*T}_b \E(Z (\Ya- \mu^*_a(X))) + \bm \beta^{*T}_b \E(\mu(X) (\Ya- \mu^*_a(X))).
\end{align*}
Both of the terms above are also zero by the population score equations for group $a$. Using the above result, we can use our results from Theorem \ref{theo:universality} applied under (G2) that the asymptotic variance of $\hat{\mu}_a^*$ is
\begin{align*}
    V &= {\rm{diag}} \left[\pi_a^{-1} \Var\{y_a - \mu^*_a(X)\}, a = 1, \dots, k \right] + \Cov\{\mu^*(X), Y\} \\
    &= {\rm{diag}} \left[\pi_a^{-1} \Var\{y_a - \gamma_a^T W\}, a = 1, \dots, k \right] + \Cov((\mu^*_1(X), ..., \mu^*_k(X))^T, (y_1, ..., y_k)^T).
\end{align*}
The $(a, b)$th entry of $\Cov((\mu^*_1(X), ..., \mu^*_k(X))^T, (y_1, ..., y_k)^T)$ is given by
\begin{align*}
    \Cov(\mu^*_1(X), y_a) &= \Cov(\gamma_a^T W, y_b) \\
    &= \gamma_a^T \Cov(W, y_b) \\
    &= \Cov(W, y_a)^T \Var(W)^{-1} \Cov(W, y_b) \\
    &= \Cov(W, y_a)^T \Var(W)^{-1} \Var(W) \Var(W)^{-1} \Cov(W, y_b) \\
    &= \gamma_a^T \Var(W) \gamma_b
\end{align*}
Therefore, the whole matrix $\Cov((\mu^*_1(X), ..., \mu^*_k(X))^T, (y_1, ..., y_k)^T)$ can be written as $\Gamma^T \Sigma_W \Gamma$ where $\Gamma = (\gamma_1, \dots, \gamma_k)$. Thus,
\begin{align*}
    V &= {\rm{diag}} \left[\pi_a^{-1} \Var\{y_a - \mu^*_a(X)\}, a = 1, \dots, k \right] + \Gamma^T \Sigma_W \Gamma.
\end{align*}

Based on the fact that we have shown that (U) and (G2) are satisfied, this means that $\hat{\theta}_{\rm JC}$ has guaranteed efficiency gain over the sample mean $\bar{Y}$, as that is the direct result of Theorem \ref{theo:universality}.

\subsection{Joint Calibration with Cross-Fitting}

\begin{corollary}
	The estimator $\hat{\theta}_{\rm JC,a}$ that uses the cross-fit $\check{\mu}_{a}(\X_i) = \sum_{k=1}^{K} I\{i \in I_k\} \hat{\mu}_{a,k}(\X_i)$ as defined in Lemma \ref{lemma: cf-aipw-equivalence} rather than the standard $\hat{\mu}_a(\X_i)$, has influence function given by
\begin{align}\label{influence-function}
		\psi_a(A_i, \X_i, Y_i) := \frac{I(A_i= a)}{\pi_a } \{ \Yai - \mu_a^* (\X_i) - (\theta_a  - P_0\mu_a^* (\X_i)) \} + \mu_a^*(\X_i) - P_0\mu_a^* (\X_i)
\end{align}
where $\mu^*_a(\X_i) = \sum_{l=1}^{L} I(Z_i = z_l) \alpha_{a,l} + \bm \beta^T_{a}{\mu}(\X_i)$
\end{corollary}
Let $\check{\mu} = (\check{\mu}_1, ..., \check{\mu}_k)$, Recall that with $\hat \mu_a^* (\X_i) = \sum_{l=1}^L I(\Z_i =z_l) \hat \alpha_{al} +\hat {\bm  \beta}_a^T  \check {\mu} (\X_i)$,
\begin{align*}
	\hat{\theta}_{\rm JC,a} - \theta_a &= \frac1n \sum_{i=1}^n \left[ \frac{1_a(A_i)}{\hat \pi_{a}}  \left\{  \Yai  - \hat \mu_a^*(\X_i) \right\} + \hat \mu_a^* (\X_i)\right] - \theta_a.
\end{align*}
Then, assuming that $\hat{\alpha}_{a,l}$ and $\hat{\beta}_{a}$ are consistent for $\alpha_{a,l}$ and $\beta_a$ respectively, we have
\begin{align*}
	\hat{\theta}_{\rm JC,a} - \theta_a &= \frac1n \sum_{i=1}^n \Bigg[ \frac{1_a(A_i)}{\hat \pi_{a}}  \bigg\{\Yai - \sum_{l=1}^L I(\Z_i =z_l) \hat \alpha_{al} - \hat {\bm  \beta}_a^T  \check {\mu} (\X_i) \bigg\} \\
	&\quad\quad\quad\quad\quad + \sum_{l=1}^L I(\Z_i =z_l) \hat \alpha_{al} +\hat {\bm  \beta}_a^T  \check {\mu} (\X_i)\Bigg] - \theta_a \\
	&= \frac1n \sum_{i=1}^n \Bigg[ \frac{1_a(A_i)}{\hat \pi_{a}}  \bigg\{\Yai - \sum_{l=1}^L I(\Z_i =z_l)  \alpha_{al} -  {\bm  \beta}_a^T  \check {\mu} (\X_i) \bigg\} \\
	&\quad\quad\quad\quad\quad + \sum_{l=1}^L I(\Z_i =z_l)  \alpha_{al} + {\bm  \beta}_a^T  \check {\mu} (\X_i)\Bigg] - \theta_a \\
	&\quad\quad - \underbrace{\frac{1}{n} \sum_{i=1}^{n} \left[\frac{1_a(A_i)}{\hat \pi_a} \left(\sum_{l=1}^{L} I(Z_i=z_l) (\hat{\alpha}_{al} - \alpha_{al}) + (\hat{\bm \beta}_a - \bm \beta_a)^T \check{\mu}(\X_i)\right)\right]}_{R_1} \\
	&\quad\quad + \underbrace{\frac{1}{n} \sum_{i=1}^{n} \left[\sum_{l=1}^{L} I(Z_i=z_l) (\hat{\alpha}_{al} - \alpha_{al}) + (\hat{\bm \beta}_a - \bm \beta_a)^T \check{\mu}(\X_i)\right]}_{R_2}.
\end{align*}
First, it suffices to show that $R_1 - R_2$ is $o_P(1/\sqrt{n})$. Starting with the $(\hat{\alpha}_{al} - \alpha_{al})$ terms, we have
\begin{align*}
	\frac{1}{n} \sum_{i=1}^{n} \left[\frac{1_a(A_i)}{\hat \pi_a} - 1 \right] \left[\sum_{l=1}^{L} I(Z_i=z_l) (\hat{\alpha}_{al} - \alpha_{al}) \right] &= \sum_{l=1}^{L} \frac{1}{n} \sum_{i:Z_i=z_l} \left[\frac{1_a(A_i)}{\hat \pi_a} - 1 \right] (\hat{\alpha}_{al} - \alpha_{al}).
\end{align*}
Therefore, we aim to show that in each level of $l$, we have
\begin{align*}
	(\hat{\alpha}_{al} - \alpha_{al}) \frac{1}{n} \sum_{i:Z_i=z_l} \left[\frac{1_a(A_i)}{\hat \pi_a} - 1 \right] = o_P(1/\sqrt{n}). 
\end{align*}
Since $(\hat{\alpha}_{al} - \alpha_{al}) = o_P(1)$, it suffices to show that $\frac{1}{n} \sum_{i:Z_i=z_l} \left[\frac{1_a(A_i)}{\hat \pi_a} - 1 \right] = O_P(1/\sqrt{n})$:
\begin{align*}
	\frac{1}{n} \sum_{i:Z_i=z_l} \left[\frac{1_a(A_i)}{\hat \pi_a} - 1 \right] &= (\hat{\pi}_a) \frac{1}{n} \sum_{i:Z_i=z_l}  \left(1_a(A_i) - \hat{\pi}_a \right) \\
	&= (\hat{\pi}_a) (n(l)/n) \frac{1}{n(l)} \sum_{i:Z_i=z_l} (1_a(A_i) - \hat{\pi}_a) \\
	&= (\hat{\pi}_a) (n(l)/n) \left(n_a(l)/n(l) -\hat{\pi}_a \right) \\
	&= (\hat{\pi}_a) (n(l)/n)\left([n_a(l)/n(l) - \pi_a] + [\pi_a -\hat{\pi}_a] \right) \\
	&= (\hat{\pi}_a) (n(l)/n) O_P(1/\sqrt{n}) \\
	&= \pi_a P(\Z_i=l) O_P(1/\sqrt{n}) + o_P(1/\sqrt{n}) \\
	&= O_P(1/\sqrt{n}).
\end{align*}
For the $(\hat{\bm \beta}_a - \bm \beta_a)^T$ term, we will use the results of previous lemmas and theorems. We have that the following vector can be re-written by adding and subtracting the scalar $(\bar{Y}_a - \theta_a)$:
\begin{align*}
	\frac{1}{n} \sum_{i=1}^{n} \left[\frac{1_a(A_i)}{\hat{\pi}_a} \check{\mu}(\X_i) - \check{\mu}(\X_i) \right] &= -\frac{1}{n} \sum_{i=1}^{n} \left[\frac{1_a(A_i)}{\hat{\pi}_a}(\bm \Yai - \check{\mu}(\X_i)) + \check{\mu}(\X_i) - \theta_a \right] + (\bar{Y}_a - \theta_a).
\end{align*}
The above term is a vector, i.e., with $(\bar{Y}_a - \theta_a) = (\bar{Y}_a - \theta_a)\bm 1_k$ and $\bm \Yai = \Yai \bm 1_k$. From Theorem \ref{theo:asymptotic-normality}, we know that $(\bar{Y}_a - \theta_a) = O_P(1/\sqrt{n})$. Using the equivalence result in Lemma \ref{lemma: cf-aipw-equivalence}, and Remark \ref{remark:cf-g-pia}, $-\frac{1}{n} \sum_{i=1}^{n} \left[\frac{1_a(A_i)}{\hat{\pi}_a}(\bm \Yai - \check{\mu}(\X_i)) + \check{\mu}(\X_i) - \theta_a \right]$ is also $O_P(1/\sqrt{n})$ by Theorem \ref{theo:asymptotic-normality}. Although the vector $\check{\mu}$ is $(\check{\mu}_1, ..., \check{\mu}_k)$ which may not be prediction unbiased for $j = 1, ..., k$, $j \neq a$, we still have the desired influence function because we are using the AIPW estimator, so we have prediction unbiasedness guaranteed. Therefore, we have
\begin{align*}
	(\hat{\bm \beta}_a - \bm \beta_a)^T \frac{1}{n} \sum_{i=1}^{n} \left[\frac{1_a(A_i)}{\hat{\pi}_a} \check{\mu}(\X_i) - \check{\mu}(\X_i) \right] &= (\hat{\bm \beta}_a - \bm \beta_a)^T \bm O_P(1/\sqrt{n}) + (\hat{\bm \beta}_a - \bm \beta_a)^T \bm O_P(1/\sqrt{n}) \\
	&= o_P(1/\sqrt{n}).
\end{align*}
Finally, we have
\begin{align*}
	\hat{\theta}_{\rm JC,a} &= \frac1n \sum_{i=1}^n \Bigg[ \frac{1_a(A_i)}{\hat \pi_{a}}  \bigg\{\Yai - \sum_{l=1}^L I(\Z_i =z_l)  \alpha_{al} -  {\bm  \beta}_a^T  \check {\mu} (\X_i) \bigg\} \\
	&\quad\quad\quad\quad\quad + \sum_{l=1}^L I(\Z_i =z_l)  \alpha_{al} + {\bm  \beta}_a^T  \check {\mu} (\X_i)\Bigg] + o_P(1/\sqrt{n}).
\end{align*}
We will show that $\check\mu_a^* (\X_i) := \sum_{l=1}^L I(\Z_i =z_l) \alpha_{al} + {\bm  \beta}_a^T  \check {\mu} (\X_i)$ satisfies Assumption 2* as long as $\hat{\mu}_{a,k}$ satisfies Assumption 2*, making the above equivalent to using a different cross-fitting estimator. Specifically, we know that with probability $1 - \Delta_{n}$, $\|\hat \mu_{a,k}  - \mu_a \|_{L^2 (P_0)} \leq \delta_n $ for all $a$ treatment groups. Then also with probability $1 - \Delta_n$,
\begin{align*}
	\norm{\check\mu_a^* (\X_i) - \mu_a^*}_{L_2(P_0)} \leq \bm \beta_a^T \norm{\check{\mu} - \mu} \leq \bm \beta_a^T  \bm 1_k \delta_n.
\end{align*}
Since $\delta_n = o(1)$, $\bm \beta_a^T  \bm 1_k \delta_n = o(1)$, so Assumption 2* is satisfied and $\hat{\theta}_{\rm JC,a}$ has the desired influence function.

\end{document}